\documentclass{article}
\usepackage{arxiv}

\pdfoutput=1

\usepackage[utf8]{inputenc}
\usepackage{listings}
\usepackage[T1]{fontenc}
\usepackage{inconsolata}

\usepackage{float}
\usepackage{amssymb,amsmath,amsthm}

\usepackage{xcolor}

\definecolor{lgreen} {RGB}{180,210,100}
\definecolor{dblue}  {RGB}{20,66,129}
\definecolor{ddblue} {RGB}{11,36,69}
\definecolor{lred}   {RGB}{220,0,0}
\definecolor{nred}   {RGB}{224,0,0}
\definecolor{norange}{RGB}{230,120,20}
\definecolor{nyellow}{RGB}{255,221,0}
\definecolor{ngreen} {RGB}{98,158,31}
\definecolor{dgreen} {RGB}{78,138,21}
\definecolor{nblue}  {RGB}{28,130,185}
\definecolor{jblue}  {RGB}{20,50,100}

\newcommand{\idlow}[1]{\mathord{\mathcode`\-="702D\it #1\mathcode`\-="2200}}

\usepackage{todonotes}

\definecolor{GreenYellow}       {RGB}{217, 229, 6} 	    
\definecolor{Yellow}            {RGB}{254, 223, 0} 	    
\definecolor{Goldenrod}         {RGB}{249, 214, 22} 	
\definecolor{Dandelion}         {RGB}{253, 200, 47} 	
\definecolor{Apricot}           {RGB}{255, 170, 123} 	
\definecolor{Peach}             {RGB}{255, 127, 69} 	
\definecolor{Melon}             {RGB}{255, 129, 141} 	
\definecolor{YellowOrange}      {RGB}{240, 171, 0} 	    
\definecolor{Orange}            {RGB}{255, 88, 0} 	    
\definecolor{BurntOrange}       {RGB}{199, 98, 43} 	    
\definecolor{Bittersweet}       {RGB}{189, 79, 25} 	    
\definecolor{RedOrange}         {RGB}{222, 56, 49} 	    
\definecolor{Mahogany}          {RGB}{152, 50, 34} 	    
\definecolor{Maroon}            {RGB}{152, 30, 50} 	    
\definecolor{BrickRed}          {RGB}{170, 39, 47} 	    
\definecolor{Red}               {RGB}{255, 0, 0}        
\definecolor{BrilliantRed}      {RGB}{237, 41, 57} 	    
\definecolor{OrangeRed}         {RGB}{231, 58, 0} 	    
\definecolor{RubineRed}         {RGB}{202, 0, 93}       
\definecolor{WildStrawberry}    {RGB}{203, 0, 68} 	    
\definecolor{Salmon}            {RGB}{250, 147, 171} 	
\definecolor{CarnationPink}     {RGB}{226, 110, 178} 	
\definecolor{Magenta}           {RGB}{255, 0, 144} 	    
\definecolor{VioletRed}         {RGB}{215, 31, 133} 	
\definecolor{Rhodamine}         {RGB}{224, 17, 157} 	
\definecolor{Mulberry}          {RGB}{163, 26, 126} 	
\definecolor{RedViolet}         {RGB}{161, 0, 107} 	    
\definecolor{Fuchsia}           {RGB}{155, 24, 137} 	
\definecolor{Lavender}          {RGB}{240, 146, 205} 	
\definecolor{Thistle}           {RGB}{222, 129, 211} 	
\definecolor{Orchid}            {RGB}{201, 102, 205} 	
\definecolor{DarkOrchid}        {RGB}{153, 50, 204} 	
\definecolor{Purple}            {RGB}{182, 52, 187} 	
\definecolor{Plum}              {RGB}{79, 50, 76} 	    
\definecolor{Violet}            {RGB}{75, 8, 161} 	    
\definecolor{RoyalPurple}       {RGB}{82, 35, 152} 	    
\definecolor{BlueViolet}        {RGB}{33, 7, 106} 	    
\definecolor{Periwinkle}        {RGB}{136, 132, 213} 	
\definecolor{CadetBlue}	  	    {RGB}{95, 158, 160} 	
\definecolor{CornflowerBlue}  	{RGB}{99, 177, 229} 	
\definecolor{MidnightBlue}	  	{RGB}{0, 65, 101} 	    
\definecolor{NavyBlue}          {RGB}{0, 70, 173}       
\definecolor{RoyalBlue}         {RGB}{0, 35, 102}       
\definecolor{Blue}              {RGB}{0, 24, 168}       
\definecolor{Cerulean}          {RGB}{0, 122, 201}      
\definecolor{Cyan}              {RGB}{0, 159, 218}      
\definecolor{ProcessBlue}       {RGB}{0, 136, 206}      
\definecolor{SkyBlue}           {RGB}{91, 198, 232}     

\definecolor{Turquoise}         {RGB}{0, 255, 239} 	    

\definecolor{TealBlue}          {RGB}{0, 124, 146} 	    
\definecolor{Aquamarine}        {RGB}{0, 148, 179} 	    
\definecolor{BlueGreen}         {RGB}{0, 154, 166} 	    
\definecolor{Emerald}           {RGB}{80, 200, 120} 	
\definecolor{JungleGreen}       {RGB}{0, 115, 99} 	    
\definecolor{SeaGreen}          {RGB}{0, 176, 146} 	    
\definecolor{Green}             {RGB}{0, 173, 131} 	    
\definecolor{ForestGreen}       {RGB}{0, 105, 60} 	    
\definecolor{PineGreen}         {RGB}{0, 98, 101} 	    
\definecolor{LimeGreen}         {RGB}{50, 205, 50} 	    
\definecolor{YellowGreen}       {RGB}{146, 212, 0} 	    
\definecolor{SpringGreen}       {RGB}{201, 221, 3} 	    
\definecolor{OliveGreen}        {RGB}{135, 136, 0} 	    
\definecolor{RawSienna}         {RGB}{149, 82, 20} 	    
\definecolor{Sepia}             {RGB}{98, 60, 27} 	    
\definecolor{Brown}             {RGB}{134, 67, 30}      
\definecolor{Tan}               {RGB}{210, 180, 140}	
\definecolor{Gray}              {RGB}{139, 141, 142} 	

\definecolor{Black}		  	    {RGB}{30, 30, 30}       
\definecolor{White}		  	    {RGB}{255, 255, 255}    

\definecolor{Numa}               {RGB}{255, 0, 0}        
\newcommand{\numacol}{red}

\usepackage{etoolbox}

\newtheorem{theorem}{Theorem}
\newtheorem{lemma}{Lemma}

\newcommand{\indentrule}{\color{code_indent}\vrule\hspace{2pt}}
\definecolor{code_indent}{HTML}{CCCCCC}

\newcommand{\ccmq}{\emph{SMQ}}
\newcommand{\smq}{\ccmq{}}

\newcommand{\stealsize}{\texttt{SIZE$_{steal}$}}

\newcommand{\stealprob}{$p_{\mathit{steal}}$}
\newcommand{\del}{\texttt{delete()}}
\newcommand{\ins}{\texttt{insert(..)}}

\newcommand{\chunksize}{\texttt{CHUNK\_SIZE}}
\newcommand{\obimdel}{$\Delta$}

\newcommand{\insprob}{$p_{\mathit{insert}}$}
\newcommand{\delprob}{$p_{\mathit{delete}}$}
\newcommand{\insbatch}{\texttt{BATCH$_{\mathit{insert}}$}}
\newcommand{\delbatch}{\texttt{BATCH$_{\mathit{delete}}$}}

\newcommand{\speed}{Speedup}
\newcommand{\workinc}{Work Increase}

\renewcommand{\paragraph}[1]{\vspace{0.3em}\noindent\textbf{#1}}

\newcounter{figureAsListingCtr}

\title{Multi-Queues Can Be State-of-the-Art Priority Schedulers}

\author{
  Anastasiia Postnikova \\
  ITMO University\\
  \texttt{postnikovaanastasiaa@gmail.com} \\
  \And
  Nikita Koval \\
  JetBrains Research\\
  \texttt{nikita.koval@jetbrains.com}
  \And
  Giorgi Nadiradze\\
  IST Austria\\
  \texttt{giorgi.nadiradze@ist.ac.at} \\
  \And
  Dan Alistarh\\
  IST Austria\\
  \texttt{dan.alistarh@ist.ac.at} \\
}

\begin{document}


\maketitle

\begin{abstract}
    Designing and implementing efficient parallel priority schedulers is an active research area. 
    An intriguing proposed design is the \emph{Multi-Queue}: given $n$ threads and $m\ge n$ distinct priority queues, task insertions are performed uniformly at random, while, to delete, a thread picks two queues uniformly at random, and removes the observed task of higher priority. 
    This approach scales well, and has probabilistic rank guarantees: roughly, the rank of each task removed, relative to remaining tasks in all other queues, is $O(m)$ in expectation. 
    Yet, the performance of this pattern is below that of well-engineered schedulers, which eschew theoretical guarantees for practical efficiency. 
    
    We investigate whether it is possible to design and implement a Multi-Queue-based task scheduler that is both highly-efficient and has analytical guarantees. 
    We propose a new variant called the \emph{Stealing Multi-Queue} (\ccmq{}), a \emph{cache-efficient} variant of the Multi-Queue, which leverages both \emph{queue affinity}---each thread has a \emph{local} queue, from which tasks are usually removed; but, with some probability, threads also attempt to \emph{steal} higher-priority tasks from the other queues---and \emph{task batching}, that is, the processing of several tasks in a single insert / remove step. 
    These ideas are well-known for task scheduling \emph{without priorities};  our theoretical contribution is showing that, despite relaxations, this design can still provide rank guarantees, which in turn implies bounds on total work performed. 
    We provide a general \ccmq{} implementation which can surpass state-of-the-art schedulers such as Galois and PMOD in terms of performance on popular graph-processing benchmarks. Notably, the performance improvement comes mainly from the \emph{superior rank guarantees} provided by our scheduler, confirming that analytically-reasoned approaches can still provide performance improvements for priority task scheduling. 
\end{abstract}

\section{Introduction}
Scalable concurrent priority schedulers are a key ingredient for efficiently parallelizing algorithms arising in graph processing, computational geometry, or scientific simulation. In such settings, algorithms are usually structured as a series of prioritized tasks, each of which accesses a subset of the shared algorithm state, computes a local update, and applies the update to the shared state. 
Many classic problems fit this pattern, such as graph processing algorithms (e.g., BFS or MST computations),  Dijkstra's single-source shortest paths (SSSP), or Delaunay mesh triangulation. 

Often, algorithm semantics imply a natural task priority ordering. 
A classic example is Dijkstra's SSSP algorithm, where the first task to be processed should be the one corresponding to the active node that is closest to the source. 
Concurrent schedulers often relax this sequential order to enable parallellism, but also, importantly, to reduce the overhead of the task scheduler's implementation, which can become extremely contended if perfect priority order is enforced~\cite{lenharth2015concurrent}. 
However, relaxing the priority order excessively can also decrease performance, as it leads to wasted work: in Dijkstra's SSSP algorithm, for example, processing a node out-of-order at a distance that is higher than its minimal distance from the source is useless, since the node will need to be re-processed later, when its distance is updated to the correct one. 
Generally, relaxed priority schedulers induce non-trivial trade-offs between the higher scalability of the scheduling mechanism itself, allowed by out-of-order execution, and the potential for wasted work caused by excessive speculation. 

The sustained recent progress on concurrent priority scheduling can be viewed through the lens of this trade-off. 
The influential Galois line of work, e.g.~\cite{Nguyen2013, pingali2011tao} proposed a family of highly-efficient schedulers which significantly relax priority order, specifically focusing on performance rather than priority guarantees~\cite{lenharth2015concurrent}. As a consequence, the scheduler has low overhead, but may induce high wasted work~\cite{PMOD}. This approach provided state-of-the-art performance upon its publication, and has inspired significant follow-up work, notably in terms of hardware implementations, e.g.~\cite{abeydeera2017sam, jeffrey2016unlocking}.  
Yesil et al.~\cite{PMOD} performed an in-depth analysis of the relaxation-vs-wasted-work trade-off in concurrent priority schedulers, and proposed a new scheduler called PMOD, combining the high scalability of Galois-type schedulers, with a dynamic priority management heuristic which reduces wasted work.

A parallel research thread has been on providing efficient relaxed schedulers \emph{with guarantees} on the maximum amount of priority rank relaxation, e.g.~\cite{SprayList, MQ, klsm, sagonas2016contention}. Of these, arguably the most popular design is the Multi-Queue~\cite{MQ}, which works roughly as follows: given $n$ threads, we instantiate $m \geq n$ concurrent priority queues, which will store tasks. To \emph{insert} a task, a thread simply places it into a random priority queue. To \emph{remove} a task, the thread picks \emph{two} priority queues at random, and removes the top element of higher priority. (For simplicity, we discuss concurrency-related details in the later sections.) 
Alistarh et al.~\cite{AKLN17} showed that a sequential variant of the Multi-Queue ensures that the rank of an element removed is always $O(m)$, in expectation, and $O(m \log m )$, with high probability in the number of queues $m$. Follow-up work showed that similar guarantees can be extended to concurrent executions~\cite{alistarh2018distributionally}, and gave work bounds for some task-based algorithms when executed via Multi-Queue-like priority schedulers~\cite{alistarh2018relaxed, alistarh2019efficiency, aksenov2020scalable}. Despite their guarantees, Multi-Queue schedulers are known to have lower overall performance relative to efficient scheduling heuristics~\cite{lenharth2015concurrent, PMOD}. 

\paragraph{Contribution.} In this paper, we show that these two lines of work can be unified, providing a highly-efficient, practical concurrent priority scheduler, while still maintaining theoretical rank guarantees, under analytical assumptions, for variants of the scheduler. 

An obvious reason for the lower performance of Multi-Queues relative to practical heuristics is  poor cache efficiency, as the basic process requires a high number of random accesses to maintain rank guarantees. 
A simple practical approach to address such issues, which we also adopt, is to affinitize threads to queues, assigning some subset of queues ``preferentially'' to each thread. 
For insertions, each thread can pick one of these preferential queues with (higher) probability $p_{\idlow{insert}}$, relative to inserting into other queues. 

We would like to apply a similar approach for deletions. 
Yet, allowing fully-local removals would perturb the two-choice process, and cause  divergence~\cite{PTW15, AKLN17}.
Nevertheless, we show that the process can be adapted as follows: to remove, with probability $p_{\idlow{steal}}$
the thread considers $\emph{stealing}$ tasks from a randomly chosen other queue, comparing the top element of a local queue with that on top of a globally-chosen random queue, and removing the higher-priority element. Otherwise, the thread directly removes from a local queue. 
As most insertions and deletions using the above scheme are local, this can result in a very unbalanced task distribution among queues, our analysis will adapt a stochastic scheduling model, by which threads are scheduled according to a scheduling distribution 
$\vec{\pi} = (\pi_1, \pi_2, \ldots, \pi_n)$, where thread $i$ is scheduled in each step with probability $\pi_i$, and we place upper and lower bounds on the maximum and minimum scheduling probabilities. Under these assumptions, we will be able to still provide rank guarantees. 

A second performance issue with standard Multi-Queues is that, in practice, the overhead of inserting or removing a task can be large relative to the task execution time. 
The standard approach to address this issue is \emph{task batching}, by which multiple tasks are inserted or removed at a single step. We show that the above random process can be resilient to task batching. 
Both the above approaches are well-known in the context of task scheduling \emph{without priorities}, e.g.~\cite{blumofe1995cilk, blumofe1999scheduling}, but have not been well-explored in the context of \emph{priority} scheduling. 

Our theoretical contribution is a generalization of the Multi-Queue analysis of~\cite{AKLN17} showing that, under assumptions, the above process, which we call \emph{the stealing Multi-Queue (SMQ)}, induces a non-trivial trade-off between the stealing probability, the scheduling properties, and the average rank of elements removed. For instance, assuming a Multi-Queue formed of $m=n$ queues (one queue per thread), with task batches of size $O(B)$, and a balanced thread scheduling distribution, we show that the expected rank removed at a step is $O(Bm)$, and $O( B m \log (Bm)) $ with high probability. 
These bounds hold irrespective of the running time. 

On the practical side, our work starts from an examination of the possible performance benefits of Multi-Queues relative to scheduling heuristics~\cite{Nguyen13, PMOD}, but also of their performance bottlenecks. 
The main benefit, which motivates our investigation, is the lower wasted work in real tasks, correlated to their rank guarantees. 
On the other hand, as mentioned, standard Multi-Queues have poor cache locality, relatively high per-task cost, and, so far, have had very limited specific implementation support. Our implementation addresses these shortcomings, via the following optimizations. 

We begin by investigating the ``optimal'' data structure for implementing individual queues. 
While previous schedulers partitioned the tasks into sub-buffers, maintained either manually or semi-automatically by priority range, we adopt an efficient variant of a \emph{sequential heap} for local task structure. 
To allow for  efficient stealing, we affix a \emph{stealing buffer} to each thread, into which the queue's owner periodically places tasks, which can be either stolen by other threads or later processed by the queue owner. 
With this in place, we examine various mechanisms for allowing concurrent access to heaps and stealing buffers, and implement a task batching mechanism similar to the one described above.  

This basic mechanism allows for several extensions, in particular a non-trivial NUMA-aware variant, which defines affinities and probabilities such that we seek to minimize out-of-socket accesses. 

\paragraph{Experimental Results.} We provide a general \ccmq{} implementation on top of the Galois graph processing framework, which includes standard and NUMA-aware variants of the \ccmq{}. 
Experiments show that our designs can surpass state-of-the-art schedulers such as OBIM and PMOD in terms of throughput and scalability when executing popular graph algorithms such as SSSP or A*. Of note, much of the performance improvement comes from significantly less wasted work, which is linked to the improvement in rank guarantees provided by our scheduler. 

\paragraph{Related Work.} To our knowledge, the  Multi-Queue-like data structure was given by the parallel branch-and-bound framework by Karp and Zhang~\cite{KarZha93}, which distributed tasks randomly among queues, assigned to processors, and also remove tasks uniformly at random. 
We stress however that their proposal is in the context of classic task scheduling in the PRAM model, and that their design does not provide rank guarantees under asynchrony. 
Our construction starts from the Multi-Queue of Rihani, Sanders, and Dementiev~\cite{MQ},  who introduced this design and provided a simple argument showing that the expected rank of the \emph{first} removed element is $O(m)$. 
Follow-up work by Alistarh et al.~\cite{AKLN17} provided a more general and involved argument, showing that, for a sequential variant of the Multi-Queue, the expected rank of \emph{any} removal is $O(m)$, by linking Multi-Queues with the classic $(1 + \beta)$ random process of Peres, Talwar, and Wieder~\cite{PTW15}. Follow-up work by the same authors~\cite{alistarh2018distributionally} extended the analysis to a \emph{concurrent} version of the Multi-Queue, under analytical assumptions. 

Relative to this work, we adapt the standard Multi-Queue semantics so that they result in efficient implementations, in particular with respect to caching, add queue locality and task batching to the original design, and then adapt the analysis approach of~\cite{AKLN17} to prove rank bounds for the resulting algorithm. 

There has been a tremendous amount of work on efficient scheduling heuristics for fine-grained task-based programs, especially in the context of graph processing~\cite{low2014graphlab, gonzalez2012powergraph, PMOD, Nguyen13, shun2013ligra, dhulipala2017julienne, dhulipala2018theoretically}. 
A complete survey is beyond our scope, so we focus on the two works that are closest to ours. 
The first is~\cite{Nguyen13}, which details the design, implementation, and practical performance of the Galois system, focusing on the OBIM (Ordered By Integer Metric) scheduler. In brief, this scheduler assigns one \emph{bag} (unordered set) per \emph{task priority class}, which is empirically defined. 
Each bag is implemented as one or more FIFO queues, one per socket. 
Enqueues insert into the bag corresponding to the task priority, creating the bag if required, and each queue element is mapped to a batch ({chunk}) of tasks. 
Threads dequeue chunks from their socket's queue; if that is empty, the thread steals from a remote queue. Tasks in a batch are performed one at a time. 
The list of bags is maintained in a global map, which is mirrored  locally by each thread for cache efficiency. 

Yesil et al.~\cite{PMOD} start from the observation that the communication-avoiding pattern of OBIM can lead to significant wasted work, due to the relatively high number of priority inversions. 
Hence, they propose a heuristic which defines \emph{priority groups}, which change dynamically at runtime. 
More precisely, PMOD tries to adapt the number of different priority bags in OBIM, by merging sets of similar priorities, so as to reduce the number of empty bags during runtime, thus trying to ensure that threads always have work to do. 
This is implemented via a dynamic merging mechanism. Conversely, the algorithm detects when there are \emph{too few} priority bags, therefore splitting bags which are too full. These operations are controlled via carefully-designed heuristics. 

Our design shares some features with these schedulers: in particular, we also recognize the importance of ``localizing'' the queues, and of batching for cache efficiency. However, it also differs in key ways, required to provide providing \emph{rank guarantees} for the resulting scheduler. 
(None of these two previous heuristics have rank guarantees, and we do not believe that such guarantees could be shown without significant modifications.) 
For instance, in keeping with the Multi-Queue design, do not split tasks per priority ``level,'' and instead maintain a local heap structure for each queue.  
Further, the stealing mechanism we use is different from both the standard Multi-Queue, and from the previous priority scheduling heuristics.


\section{The Stealing Multi-Queue}\label{sec:algorithm}

\subsection{The Classic Multi-Queue Design}\label{subsec:optimizations}

The classic Multi-Queue uses $m$ sequential queues, each protected by a lock, and distributes requests among them. Typically, $m$ is taken to be the number of threads $T$ multiplied by a constant factor $C \geq 2$, making it likely that individual operations will not interfere with each other when taking locks. 
An \texttt{insert(x)} comes, it chooses uniformly random queue and tries to lock it. When successful, it adds \texttt{x} into it and releases the lock. If the lock acquisition fails, the operation restarts. 
Similarly, \texttt{delete()} picks \textit{two} different queues uniformly at random, and removes from the higher priority top element. 
Then, it tries to lock the chosen queue and retrieve the top task from it, releasing the lock at the end. If the lock acquisition fails, the operation restarts.

Listing~\ref{lst:classic_mq} presents a pseudo-code of a simplified Multi-Queue version which may return \texttt{null} in \texttt{delete()} when the queue it removes an element from becomes empty. The version described here faithfully models the Galois implementation of Multi-Queues~\cite{Nguyen13}. 

\begin{lstlisting}[
label={lst:classic_mq}, 
caption={The classic Multi-Queue implementation.}
]
class MultiQueue<E> {
 val queues = Queue<E>[C * T] 
 
 fun insert(task: E) = while(true) {
 #\indentrule#  q := queues[random(0, queues.size)]
 #\indentrule#  if !tryLock(q): continue
 #\indentrule#  q.add(task)
 #\indentrule#  unlock(q)
 #\indentrule#  return
 }
 
 fun delete(): E? = while(true) {
 #\indentrule#  i1, i2 := distinctRandom(0, queues.size)
 #\indentrule#  q1 := queues[i1]; q2 := queues[i2]
 #\indentrule#  if !tryLock(q1, q2): continue
 #\indentrule#  q := q1.top() < q2.top() ? q1 : q2
 #\indentrule#  task := q.extractTop()
 #\indentrule#  unlock(q1, q2)
 #\indentrule#  return element
 }
}
\end{lstlisting}

\paragraph{Optimization 1: Task Batching.}
A standard way to reduce the ratio between synchronization cost and task execution time in schedulers is \emph{task batching}: \texttt{delete()} operations can retrieve multiple tasks from the same queue at once, storing them into a fixed-size thread-local buffer, and \texttt{insert(..)}\nobreakdash-s put the tasks into another thread-local buffer, flushing it to a random queue when the number of buffered tasks exceeds the buffer capacity.
Most practical graph processing frameworks implement some variant of this optimization, e.g.~\cite{Nguyen13}.
Our analysis can bound its impact on the rank guarantees of the Multi-Queue. 

The benefits of this approach are that (1) it reduces the number of lock acquisitions by a constant factor which is approximately the batch size and 
(2) it reduces the number of cache misses and contention compared to the version which accesses different queues on each \texttt{insert(..)} and \texttt{delete()}. 
However, it clearly impacts rank guarantees, since multiple elements are retrieved from the same queue so that other queues and further task insertions are ``ignored'' until the buffered tasks are processed. 
When this optimization is applied to insert operations, buffered tasks cannot be processed, so it also makes the implementation less fair. 

\paragraph{Optimization 2: Temporal Locality.}
A similar, but different approach to reduce the cache coherence overhead is to use the same queue for a sequence of \texttt{delete()} or \texttt{insert()} operations, making it more likely that the corresponding data is already cached at the current core. Specifically, in this variant, the thread flips a biased coin before each new operation to decide whether to keep using the same queue as in the previous operation, or potentially pick another queue, according to the algorithm. (The coin is biased towards locality.) 
Our analysis approach can also provide rank bounds for this approach, although we will focus on analyzing the more performant \emph{stealing} variant. 

The difference between this method and task buffering is that updates to the queue (e.g. newly inserted tasks by another thread) would be visible to the current thread in this case. 
However, this approach is more costly, as it requires synchronization upon every operation, although obtaining the same lock multiple times, when uncontended, is relatively cheap. 

We can examine the difference in terms of Dijkstra's SSSP algorithm. In a step, the algorithm usually relaxes multiple edges, and thus adds several tasks to the queue. With temporal locality, the sequence of \texttt{insert(..)} invocations may insert part of all of these tasks into the same queue. Moreover, it is possible to insert several elements into the same queue with a single lock acquisition: with the lock acquired, the algorithm flips a coin after the insertion of each task to determine whether to keep inserting elements, or whether to potentially switch queues. 
With this optimization, the cost of lock acquisitions can be lower than in the classic Multi-Queue, but is still greater than with task buffering, where a fixed set of tasks is always inserted. The advantage, however, is in relatively better rank guarantees.  





\subsection{The Stealing Multi-Queue}\label{sec:stealing_mq}\label{subsec:smq}

The experimental data in Section~\ref{sec:experiments} show that \emph{task batching} and \emph{temporal locality} optimizations can improve the performance of Multi-Queue scheduling for graph algorithms by up to $3\times$ relative to the classic variant. 
However, accessing queues by different threads and lock acquisitions still have high performance cost. Therefore, we constructed a new variant which we call the \emph{Stealing Multi-Queue (SMQ)}, which eschews locks, and improves cache locality beyond the optimizations discussed above, while maintaining rank guarantees under analytic assumptions. 

Without task priorities, the classic way to implement schedulers, e.g.~\cite{blumofe1999scheduling}, is to use thread-local queues and allow \emph{stealing}, so that threads can both add and retrieve tasks from their own queues, and steal tasks when the queues become empty. 

The idea behind \smq{} is similar {---} it also allows task stealing, but it uses thread-local \emph{priority} queues. To guarantee fairness, \smq{} steals tasks not only when it finds the thread-local queue empty, but also with a constant probability in each step. Specifically, in each step, with probability $p_{\idlow{steal}},$ the thread compares the priority of the top element in its local queue with that of a randomly chosen queue. 
In Section~\ref{sec:proofs} we provide a rank analysis for this process under analytic assumptions. 

Stealing lends itself to additional optimizations. 
For example, we employ the task batching optimization for task stealing: threads do not steal a single task at a time, but a whole batch. 
Intuitively, we aim to take advantage of the fact that in e.g. graph algorithms, tasks with similar priorities refer to nodes that are close to eachother in the graph, and therefore it would be more efficient that they are processed together.

\paragraph{High-Level Algorithm.}
Listing~\ref{listing:smq} presents a high-level pseudocode for the \smq{} algorithm. Thread-local queues are stored in the \texttt{queues} array, where \texttt{queues[t]} is associated with thread \texttt{t} (line~\ref{line:smq:queues}). Additionally, each threads owns a  \texttt{stolenTasks} buffer of capacity \stealsize{}~\texttt{-~1} that is thread-local and stores the tasks stolen from another queue (line~\ref{line:smq:stolenTasks}). 

\begin{lstlisting}[label={listing:smq}, caption={
High-level Stealing Multi-Queue (\smq{}) algorithm. It assumes that \texttt{Queue} is provided with an additional \texttt{steal(k)} function that retrieves top \texttt{k} tasks (or less, if the queue size is lower). The implementation of this \texttt{steal(k)} function is discussed in Section~\ref{sec:smq_impl}.
}]
class StealingMultiQueue<E> {
 // #\color{Mahogany}queues[t]# is associated with thread `t`
 val queues: Queue<E>[threads] #\label{line:smq:queues}#
 threadlocal val stolenTasks = #\nolinebreak#Buffer<E>(#\stealsize{}#-1) #\label{line:smq:stolenTasks}#
  
 fun insert(task: E) =  #\label{line:smq:insert0}#
   queues[curThread()].addLocal(task) #\label{line:smq:insert1}#
    
 fun delete(): E? {
 #\indentrule#  // Do we have previously stolen tasks?
 #\indentrule#  if stolenTasks.isNotEmpty(): #\label{line:smq:delete:processStolen:start}#
 #\indentrule#  #\indentrule#  return stolenTasks.removeFirst() #\label{line:smq:delete:processStolen:end}#
 #\indentrule#  // Should we steal?
 #\indentrule#  with #\stealprob{}# probability { #\label{line:smq:delete:steal:start}#
 #\indentrule#  #\indentrule#  task := trySteal()  #\label{line:smq:delete:steal:trySteal}#
 #\indentrule#  #\indentrule#  if task != null: return task
 #\indentrule#  } #\label{line:smq:delete:steal:end}#
 #\indentrule#  // Try to retrieve the top task
 #\indentrule#  // from the thread-local queue
 #\indentrule#  task := queues[curThread()].extractTopLocal() #\label{line:smq:delete:extractTop:start}#
 #\indentrule#  if (task != null) return task #\label{line:smq:delete:extractTop:end}#
 #\indentrule#  // The local queue is empty, try to steal
 #\indentrule#  return trySteal()  #\label{line:smq:delete:stealEmpty}#
 }
  
 fun trySteal(): T? { #\label{line:smq:trySteal:start}#
 #\indentrule#  // Choose a random queue and check whether 
 #\indentrule#  // its top task has higher priority
 #\indentrule#  t := curThread()
 #\indentrule#  qId := random(0, queues.size)
 #\indentrule#  if queues[qId].top() < queues[t].top():
 #\indentrule#  #\indentrule#  // Try to steal a better task!
 #\indentrule#  #\indentrule#  stolen := queues[qId].steal(STEAL_SIZE)
 #\indentrule#  #\indentrule#  if stolen.isEmpty(): return null // failed
 #\indentrule#  #\indentrule#  // Return the first task and add the others
 #\indentrule#  #\indentrule#  // to the thread-local buffer of stolen ones
 #\indentrule#  #\indentrule#  stolenElements.add(stolen[1:])
 #\indentrule#  #\indentrule#  return stolen[0]
 } #\label{line:smq:trySteal:end}#
}
\end{lstlisting}

The \texttt{insert(..)} operation is straightforward{---}it simply adds the specified task into the thread-local queue (lines~\ref{line:smq:insert0}--\ref{line:smq:insert1}).

The \texttt{delete()} operation first processes any buffered stolen tasks (lines~\ref{line:smq:delete:processStolen:start}--\ref{line:smq:delete:processStolen:end}). 
In case the buffer is empty, it steals tasks from another queue with probability \stealprob{} and returns the  stolen task of highest priority if this succeeds (lines~\ref{line:smq:delete:steal:start}--\ref{line:smq:delete:steal:end}). 
If the algorithm did not steal or stealing has failed (the \texttt{trySteal()} invocation at line~\ref{line:smq:delete:steal:trySteal} returned \texttt{null}), \texttt{delete()} retrieves the top task from the thread-local queue and returns it (lines~\ref{line:smq:delete:extractTop:start}--\ref{line:smq:delete:extractTop:end}).
However, the thread-local queue can be empty. In this case, the algorithm tries to steal tasks from another queue (line~\ref{line:smq:delete:stealEmpty}).
As before, this implementation may return \texttt{null} if not task is found.

The stealing logic is implemented in the \texttt{trySteal()} function (lines~\ref{line:smq:trySteal:start}--\ref{line:smq:trySteal:end}) which attempts to steal \stealsize{} tasks from a random queue if the top element from the thread-local queue has lower priority, and returns the top task as a result, storing all the resulting tasks into the thread-local \texttt{stolenTasks} buffer.

\section{Analysis of the Stealing Multi-Queue} \label{sec:proofs}

We consider the following simplified variant of the \smq{} algorithm in the analysis. 

\begin{lstlisting}[label={listing:theoreticalalgo}, caption={
 The stealing Multi-Queue algorithm assumed in the analysis. It assumes that each queue is equipped with \texttt{extractTopB()} method which retrieves and returns its top $B$ elements.
}]
class StealingMultiQueue<E> {
 // #\color{Mahogany}queues[t]# is associated with thread `t`
 val queues: Queue<E>[threads] 
 
 fun insert(task: E) = 
   queues[curThread()].add(task)  
    
 fun delete(): List<E> {
 #\indentrule#  with #\stealprob{}# probability { // should we steal?
 #\indentrule#  #\indentrule#  return trySteal() 
 #\indentrule#  }
 #\indentrule#  // Never empty in the theoretical model
 #\indentrule#  return queue[curThread()].extractTopB() 
 }
  
 fun trySteal(): List<E> {
 #\indentrule#  // Choose a random queue and check whether its
 #\indentrule#  // top task has higher priority, #\color{Mahogany}$\Rightarrow$# lower rank
 #\indentrule#  t := curThread()
 #\indentrule#  qId := random(0, queues.size)
 #\indentrule#  if (queues[qId].top() < queues[t].top()):
 #\indentrule#  #\indentrule#  return queues[qId].extractTopB()
 #\indentrule#  }
 #\indentrule#  return queues[t].extractTopB()
 }
\end{lstlisting}

\paragraph{Analytical Model.} 
We now provide rank bounds for the above algorithm, in a simplified analytical model. As in~\cite{AKLN17}, we assume that all element insertions occur initially and in increasing rank order. Additionally, we assume that tasks are inserted at random among the queues, and that \texttt{extractTopB()} operations occur atomically, so that we can analyze a sequential ``linearized'' version of the process.

Importantly, we assume a stochastic scheduling model similar to that of~\cite{alistarh2015lock}, where we are given a thread scheduling distribution $\vec{\pi} = (\pi_1, \pi_2, \ldots, \pi_n)$ for the $n$ threads, where $\pi_i$ is the probability with which thread $\pi_i$ is scheduled to perform  operation in a step.
Since we assume that all insertions occur initially, threads will not observe empty queues, so our main objective will be to characterize the expected rank bound of the elements removed, relative to all elements present in the queues. 
The \texttt{extractTopB()} operation retrieves $B$ elements, for constant $B$; however, for convenience, we will first describe the case $B = 1$ first. 

Let \stealprob{} be the stealing probability, that is: once a thread is scheduled to perform a delete  operation, 
with probability \stealprob{} stealing occurs: it picks a second queue uniformly at random among all queues,
and ``steals'' its top element if the rank of its top element is smaller than the rank of the element on top its local queue.
If no stealing occurs, the thread simply deletes the top element of its local queue and returns it.
Finally, we assume that there exists a constant $\gamma$ such that, for each thread $i$, $1-\gamma \le \frac{1}{\pi_i n} \le 1+\gamma$, for $\gamma \le 1/2$. Intuitively, the parameter $\gamma$
bounds how unfair the thread scheduler may be. 
For instance, the value $\gamma = 0$ means that the thread scheduler is completely uniform. 

\paragraph{The Main Theorem.}
Given the above definitions, our main claim is the following. 


\begin{theorem}
\label{thm:main}
Assume the thread scheduling probability distribution $\vec{\pi} = (\pi_1, \pi_2, \ldots, \pi_n)$ with constant $\gamma \geq 0$ such that, for each thread $i$, $1-\gamma \le \frac{1}{\pi_i n} \le 1+\gamma$, for $\gamma \le 1/2$,\footnote{If the scheduling distribution is \emph{uniform}, then $\gamma = 0$.} and let $p_{\idlow{steal}}$ be the stealing probability.
If $\gamma\Big(\frac{1}{p_{\idlow{steal}}}-1\Big) \le \frac{1}{2n}$, then the SMQ process which removes $B$ elements during the delete operation satisfies that,
for any time step $t$, the expected maximum rank of elements on top of queues is $O\Bigg(\frac{nB(1+\gamma)}{p_{\idlow{steal}}}\Big(\log{n}+\log{\frac{(1+\gamma)}{p_{\idlow{steal}}}}\Big)\Bigg)$
and the expected average rank (maximum and average ranks are computed over $Bn$ elements: $B$ elements on top of $n$ queues) is at most 
$O\Bigg(\frac{nB(1+\gamma)}{p_{\idlow{steal}}}\log{\frac{(1+\gamma)}{p_{\idlow{steal}}}}\Bigg)$.
\end{theorem}

\paragraph{Discussion.}
The above claim is somewhat abstract, so let us build intuition via some examples. 
Consider first the case of a \emph{uniform} scheduling distribution, $\gamma = 0$, and no task batching, $B = 1$. 
Then, if the stealing probability is \emph{constant}, the expected average rank is $O(n)$, and the expected maximum rank is $O(n \log n)$. If we wish to steal less often, i.e. $p_{\idlow{steal}} = O(1 / n)$, then the expected rank cost becomes $O(n^2 \log n)$ in both cases. 

More generally, the claim also shows that, if $p_{\idlow{steal}}$ is large enough, we consistently obtain good rank bounds. Given parameter $0 \le q \le 1$, if $p_{\idlow{steal}} \geq 1-\frac{1}{n^{q}+1} \ge \frac{1}{2}$ and  $\gamma \le \frac{n^q}{2n}$, then
$\gamma\Big(\frac{1}{p_{\idlow{steal}}}-1\Big) \le \frac{1}{2n}$ and 
$\frac{1+\gamma}{p_{\idlow{steal}}}=O(1)$.
Hence, the expected average rank bound 
becomes $O(n)$, irrespective of the scheduling imbalance.
We can also use smaller $p_{\idlow{steal}}$ at the expense of larger rank bounds
and smaller tolerated $\gamma$. If $\gamma \le \frac{p_{\idlow{steal}}}{2n} \le \frac{1}{2n}$,  threads are scheduled almost uniformly.
Here, the expected average rank bound becomes $O\Big(\frac{n \log{\frac{1}{p_{\idlow{steal}}}}}{p_{\idlow{steal}}}\Big)$.

\paragraph{Analysis Overview.}
Our analysis generalizes the argument of~\cite{AKLN17} to the more intricate variant of the Multi-Queue process which we consider. Due to space constraints, we only sketch the argument here, and provide the full proof as supplementary material. 

The first step is similar to~\cite{AKLN17}: we describe a coupling which equates the discrete $SMQ$ process described above to a \emph{continuous} balls-into-bins process.
We replace the $n$ queues with $n$ bins, each of which initially contains a single ball of label $0$. 
We start by inserting infinitely many balls into the bins, so that for each bin $i$ the difference between the labels of two consecutive balls
is an exponential random variable with mean $\pi_i$. 
Then, we perform  $T$ insertions in the modified $SMQ$ process (recall that elements are inserted in the increasing rank order) as follows :
for each element with rank $t \le T$, we insert it into $i$ if the label (ball) with rank $t$ is inserted in the bin $i$.
Finally, we remove the labels which have rank larger than $T$ from the bins. 
Note that after the insertion phase the rank distributions of labels (balls) in bins and elements in queues are equal.
The first technical step is a Lemma proving that elements are inserted into the queues 
in the same way as the original $SMQ$ process.

Assuming that the number of initial insertions $T$ is large, we extend the coupling to the second \emph{removal phase}. 
In each removal step, we first pick a ``local'' bin $i$ with probability $\pi_i$, and then flip a coin to decide whether to steal or not. With probability \stealprob{}, we decide to steal. 
If so, we pick a second bin uniformly at random from among all $n$, and examine the two balls of lowest label (highest priority) on ``top'' of the two bins. 
Following the priority process, we remove the ball of lowest label among the two, uncovering the next ball in the bin. 
If the coin flip dictated that we \emph{not steal}, then we directly remove the ball on top of bin $i$. (Notice that, in both cases, the label increment for a chosen bin is exponentially-distributed).
If batch size $B>1$, we remove $B$
 labels from the bin.
 In the $SMQ$ process we make the same random choices as in the balls and bins process and follow the same removal procedure, and since the rank distributions are equal after the insertion phase, it is easy to see that they will stay equal after every removal. That is, if the ball is removed from the bin $k$ then the element is removed from the queue $k$.
 
 The \emph{rank cost} at a step is the \emph{rank} of the label of the removed ball among all labels still present in all the bins. Because of the coupling, this will imply that the rank cost in the $SMQ$ process is also bounded in expectation.

Let $\ell_i(t)$ be the label on top of bin $i$ at time step $t$.
Let $x_i(t)=\ell_i(t)/n$ be the normalized value of this label, 
and let $\mu(t)=\frac{1}{n} \sum_{i=1}^n x_i(t)$ be the average normalized label at time step $t$.
As in~\cite{PTW15, AKLN17}, we will be analyzing the potential function 
\begin{equation}
    \Gamma(t)=\sum_{i=1}^n e^{-\alpha(x_i(t)-\mu(t))}+\sum_{i=1}^n e^{\alpha(x_i(t)-\mu(t))},
\end{equation} 
for a suitable constant $\alpha > 0$, which we will define later.

A key ingredient of this analysis is the $(1+\beta)$-choice random process~\cite{PTW15}, which is similar to the continuous process we defined above, with the difference that in order to delete an element, with probability $(1-\beta)$ we choose a single bin uniformly at random and delete from it, and with probability $\beta$
we choose two bins uniformly at random and delete from the one which has a ball of lower label on top. (Thus, in expectation, we perform $(1 + \beta)$ choices at a step.) 

\paragraph{Bounding the Potential.}
Fix an arbitrary time step $t$, and the labels  $x_1(t), x_2(t), ... x_n(t)$  on top of the bins.
Let $\Gamma_c(t+1)$ be the potential at time step $t+1$ if the label is deleted by our 
continuous process, and let $\Gamma_{\beta}(t+1)$ be the potential at time step $t+1$
if the label is deleted by $1+\beta$ process.
Our strategy is to show that there exists $\beta$ such that 
\begin{align}
\label{eq:dominance}
    \mathbb{E}[\Gamma_c(t+1) &|(x_i(t)_{i = 1, n}] \nonumber \le   \mathbb{E}[\Gamma_{\beta}(t+1)|(x_i(t)_{i = 1, n}].
\end{align}

At this point, our argument diverges from that of~\cite{AKLN17}. 
First, we define $S_c(i)$ to be the probability that we delete a label from one of the first $i$ bins in our process continuous process, and let $S_{\beta}(i)$ be the same probability for the $(1+\beta)$ process. 
Our first technical step is to show that, under reasonably chosen parameter values $\gamma\Big(\frac{1}{p_{\idlow{steal}}}-1\Big) \le \frac{1}{2n}$ and $\beta=\frac{p_{\idlow{steal}}}{2(1+\gamma)}$, we have that 
for any $1 \le i \le n, S_c(i) \ge S_{\beta}(i).$ 

Based on this, our key technical step will be to use coupling in order to show that the potential dominance claimed above holds for suitable chosen parameter values. 
Subsequently, we can use \cite[Lemma 3]{AKLN17} to show that for any step $t+1$:
\begin{align*}
\mathbb{E}[\Gamma_{c}(t+1)|(x_i(t)_{i = 1, n}] \le & \mathbb{E}[\Gamma_{\beta}(t+1)|(x_i(t)_{i = 1, n}] \\\le & \Big(1-\Omega\left(\frac{\beta^2}{n}\right)\Big)\Gamma(t)\nonumber +poly(\frac{1}{\beta}).
\end{align*}

The rest of the proof leverages the implied potential bound to obtain \emph{rank} bounds on the elements removed.
First, we use the above supermartingale type bound to prove that for any time step $t \ge 0$:  
$\mathbb{E}[\Gamma_{c}(t)] \le O\Big(poly\left(\frac{1}{\beta}\right)n\Big).$
Finally, we can use Theorems 4 and 5 in the full version of \cite{AKLN17} to show
that for the continuous process and time  $t \ge 0$, the expected maximum rank of labels on top of bins at time step $t$ is
$O\Big(\frac{n}{\beta}(\log {n}+\log {\frac{1}{\beta}})\Big)$ and expected average rank is 
$O\Big(\frac{n}{\beta}\log {\frac{1}{\beta}}\Big)$. Plugging in $\beta=\frac{p_{\idlow{steal}}}{2(1+\gamma)}$ in the above rank provides the proof of the theorem for $B=1$. The proof for $B > 1$ is slightly more involved, and is described in the supplementary material.

\newpage
\section{SMQ Implementation Details}\label{sec:smq_impl}

In Section~\ref{sec:stealing_mq} we presented the general SMQ design, skipping the details of the stealing implementation. 
Here, we address this gap. 
In our investigation, we have first used \emph{concurrent skip-lists} as thread-local queues as well as \emph{sequential $d$-ary heaps} augmented with special \emph{stealing buffers}.
We have found the latter approach to perform consistently better, and therefore we focus on it here.

\paragraph{\smq{} via d-ary Heaps with Stealing Buffers.}
Listing~\ref{listing:smq_heaps} presents the pseudo-code for this design. We use sequential $d$-ary heaps (typically, with parameter $d=4$) as thread-local queues for storing tasks and separate the synchronization and stealing from the sequential heap implementation. 

\begin{lstlisting}[label={listing:smq_heaps}, caption={
Stealing Multi-Queue implementation via $d$-ary heaps with stealing buffers.
}]
class HeapWithStealingBufferQueue<E> {
 val q = Heap<E>() // local d-ary heap #\label{line:smq_heaps:q}#
 // Other threads can steal from this buffer
 val stealingBuffer = Buffer<E>(#\stealsize{}#) #\label{line:smq_heaps:b}#
 // 64-bit register with the current buffer epoch 
 // and the "tasks are stolen" flag
 val (epoch, stolen): (Int, Bool) = (0, true)  #\label{line:smq_heaps:ef}#

 fun addLocal(task: E) { #\label{line:smq_heaps:add:start}#
 #\indentrule#  q.add(task) // add to the local queue #\label{line:smq_heaps:add:q}#
 #\indentrule#  if stolen: fillBuffer()  #\label{line:smq_heaps:add:b}#
 }  #\label{line:smq_heaps:add:end}#
  
 fun extractTopLocal(): E? { #\label{line:smq_heaps:top:start}#
 #\indentrule#  if stolen: fillBuffer() 
 #\indentrule#  return q.extractTop()
 } #\label{line:smq_heaps:top:end}# 
  
 fun top(): E? = while(true) { 
 #\indentrule#  // Read the current epoch and the flag
 #\indentrule#  (curEpoch, curStolen) := (epoch, stolen) #\label{line:smq_heaps:top:read}#
 #\indentrule#  if stolen: return null // can we steal? #\label{line:smq_heaps:top:check}#
 #\indentrule#  top := stealingBuffer.first() // read the top #\label{line:smq_heaps:top:top}#
 #\indentrule#  if curEpoch != epoch: continue // restart #\label{line:smq_heaps:top:recheck}#
 #\indentrule#  return top // return the top buffer element
 }
  
 fun steal(size: Int): List<E> = while(true) {  #\label{line:smq_heaps:steal:start}#
 #\indentrule#  // Read the current epoch and the flag
 #\indentrule#  (curEpoch, curStolen) := (epoch, stolen) #\label{line:smq_heaps:steal:read}#
 #\indentrule#  if stolen: return emptyList() // can't steal #\label{line:smq_heaps:steal:check}#
 #\indentrule#  // Read the tasks non-atomically
 #\indentrule#  tasks := stealingBuffer.read()  #\label{line:smq_heaps:steal:buffer}#
 #\indentrule#  atomic { // atomically update (epoch, stolen) #\label{line:smq_heaps:steal:atomic0}#
 #\indentrule#  #\indentrule#  if epoch != curEpoch || stolen: continue
 #\indentrule#  #\indentrule#  stolen = true // the tasks have been stolen!
 #\indentrule#  } #\label{line:smq_heaps:steal:atomic1}#
 #\indentrule#  return tasks
 } #\label{line:smq_heaps:steal:end}#
  
 fun fillBuffer() { // stolen #\color{Mahogany}==# true #\label{line:smq_heaps:fb:end}#
 #\indentrule#  stealingBuffer.clear()  #\label{line:smq_heaps:fb:clear}#
 #\indentrule#  // Fill the buffer, keep the flag `true`
 #\indentrule#  repeat(STEAL_SIZE) {  #\label{line:smq_heaps:fb:fill0}#
 #\indentrule#  #\indentrule#  task := q.extractTop()
 #\indentrule#  #\indentrule#  if task == null: break
 #\indentrule#  #\indentrule#  stealingBuffer.add(task)
 #\indentrule#  } #\label{line:smq_heaps:fb:fill1}#
 #\indentrule#  // Increment the epoch and re-set the flag
 #\indentrule#  (epoch, stolen) = (epoch + 1, false) #\label{line:smq_heaps:fb:apply}#
 } #\label{line:smq_heaps:fb:start}#
}
\end{lstlisting}


To steal efficiently, we maintain the metadata such as a buffer epoch and a ``tasks are stolen'' flag in a single 64-bit integer field (line~\ref{line:smq_heaps:ef}). When the tasks in the buffer are stolen, the flag should be set to \texttt{true}. Simultaneously, when the buffer is filled with new tasks, the epoch increases.

Thus, the stealing procedure presenting in the \texttt{steal(..)} function (lines~\ref{line:smq_heaps:steal:start}--\ref{line:smq_heaps:steal:end}) is organized as follows. First, it atomically reads the current epoch and the flag (line~\ref{line:smq_heaps:steal:read}). If the tasks have been already stolen, the function fails and returns an empty list (line~\ref{line:smq_heaps:steal:check}). Next, it reads the buffer in a non-atomic way in the hope that the tasks are still in the buffer (line~\ref{line:smq_heaps:steal:buffer}). Then, the algorithm atomically checks that the epoch has not been updated, and changes the flag from \texttt{false} to \texttt{true} (line~\ref{line:smq_heaps:steal:atomic0}--\ref{line:smq_heaps:steal:atomic1}). On success, it returns the already read tasks; otherwise, it restarts the procedure from the beginning.

In the \texttt{add(..)} operation, the task is added to the thread-local heap (line~\ref{line:smq_heaps:add:q}) followed by a \texttt{fillBuffer()} invocation if the tasks from the buffer are stolen (line~\ref{line:smq_heaps:add:b}). The \texttt{extractTop()} operation also delegates the work to the sequential heap, filling the buffer with new tasks if needed (lines~\ref{line:smq_heaps:top:start}--\ref{line:smq_heaps:top:end}).

The \texttt{top()} operation uses the same approach as stealing: it reads the current epoch and flag (line~\ref{line:smq_heaps:top:read}), checks that the tasks are not stolen (line~\ref{line:smq_heaps:top:check}), reads the top task from the buffer (line~\ref{line:smq_heaps:top:top}), and checks that the epoch has not been changed (line~\ref{line:smq_heaps:top:recheck}). When the epoch is the same, the operation returns the already read task; otherwise, it restarts.

To refill the buffer (lines~\ref{line:smq_heaps:fb:start}--\ref{line:smq_heaps:fb:end}), the algorithm clears it first (line~\ref{line:smq_heaps:fb:clear}) and puts \stealsize{} top tasks from the local sequential heap (lines~\ref{line:smq_heaps:fb:fill0}--\ref{line:smq_heaps:fb:fill1}). During this procedure, the ``tasks are stolen'' flag is \texttt{true}, so other threads do not interfere. In the end, the algorithm applies  refilling by atomically increasing the epoch and resetting the flag (line~\ref{line:smq_heaps:fb:apply}).


\paragraph{NUMA-Awareness.}\label{sec:numa}
The classic Multi-Queue design does not consider Non-Uniform Memory Access (NUMA) effects, and threads are likely to access queues located on another socket, which can affect performance. 
We present a simple strategy which significantly reduces overheads, and fits the analysis of Section~\ref{sec:proofs}.

Assume $N$ NUMA nodes with $T_i$ threads each (here, $i$ is the node index). 
In this configuration, we create $T_i \times C$ queues for each node. This way, the total number of queues stays the same as before. The straightforward idea is for threads to preferentially access queues from the same node, so all the synchronization conflicts are likely to be resolved on the last-level cache (usually, L3) of the current NUMA node. 
For this, we use a weighted probability distribution to be used when choosing a new queue to sample as part of either the Multi-Queue or the Stealing Multi-Queue. 
For a given thread, all $T_i \times C$ queues associated with its own node will have  weight $1$ while all other queues, associated with other nodes, will have weight $1 \over K$, where $K > 1$ is a constant. Intuitively, by adjusting the multiplier $K$, we balance fairness with the number of ``out-of-node'' accesses.

Specifically, given a thread from note $i$, the total weight of all queues is $W_i = \sum_{j \neq i}^{} {T_j \cdot C \over k} + T_i \times C$.
The probability of choosing a queue from the same node by a thread from the node $i$ is $T_i C / W_i$. Thus the expected number of ``internal'' queue choices by all threads in node $i$ is $E_i = T_i^2 \cdot C / W_i$, and the expected ratio of ``internal'' queue choices by \emph{all} threads in \emph{all} nodes is $E = \sum E_i$. This metric represents the ``NUMA-friendliness'' of the algorithm.
Typically, since we use all cores, the number of threads on each NUMA node is equal to $T / N$, where $T = \sum_i T_i$. When $K > N$, which is reasonable in practice, we simplify to the expression of $E$ to
\noindent\vspace{-0.5em}
\[E_{int} \approx T \times \left(1 - {1 \over K}\right).\] 


In our implementation, we aim to keep this ratio constant as we increase the thread count, and therefore we make $K$ depend linearly depend on the number of threads $T$.


\section{Evaluation} \label{sec:experiments}
We now examine how the various optimizations we discussed impact the performance of the Multi-Queue, 
and compare relative to high-performance scheduling heuristics, such as OBIM and and its adaptive PMOD version, Spray-List~\cite{SprayList}, and the random-enqueue local-dequeue (RELD) algorithm from~\cite{RELD}. 
We examine throughput versus a single-threaded baseline (and hence also scalability) but also  the total amount of additional work incurred due to scheduling. 
We build on Galois~\cite{Nguyen13}, a popular graph processing framework, 
which is still maintained and provides efficient implementations of the above two schedulers. 
Our implementation will be made open-source upon publication.


\paragraph{Hardware.}
Experiments were performed an AMD machine featuring two EPYC 7702 64-Core processors with hyper-treading for 256 total hardware threads, and an Intel machine with features four Intel Xeon Gold 5218 processors with 16 cores each for 128 hardware threads in total.

\begin{table}[t]
    \centering
    {\small
    \begin{tabular}{|p{1.55cm}|p{0.9cm}|p{1.2cm}|p{7cm}|}
    \hline
\large{\textbf{Graph}} & \large{\textbf{|V|}}&\large{\textbf{|E|}} & \large{\textbf{Description}}    
\\
\hline
 \textbf{USA} & \centering\large$24M$ & \centering\large$58M$ & Full USA roads \\\hline
 \textbf{WEST} & \centering\large$6M$  & \centering\large$15M$ & Roads of the western part of the USA \\\hline
 \textbf{TWITTER} & \centering\large$41M$ & \centering\large$1468M$ & Directed graph describing follower relation in Twitter	 \\\hline
  \textbf{WEB} & \centering\large$50M$& \centering\large$1930M$ & Directed web crawl of .sk domain \\\hline
  \end{tabular}
  }
    \vspace{0.3em}
    \caption{Input graphs for the algorithms used to benchmark different schedulers. }
    \label{table:graphs}
\end{table}

\begin{figure*}
    \centering
    \includegraphics[width=0.92\textwidth]{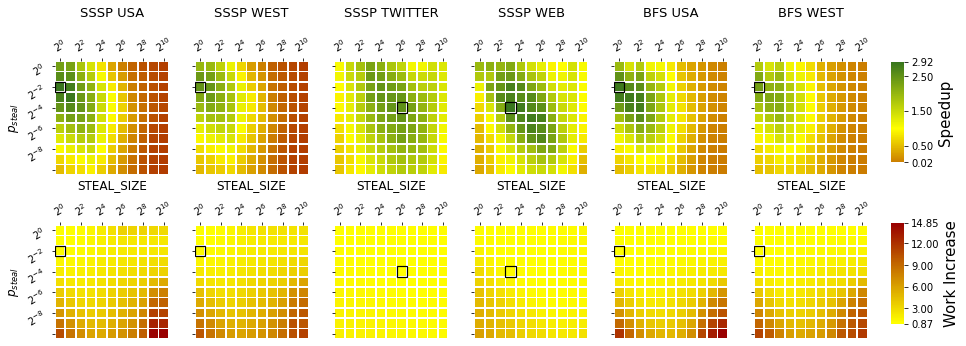}
    {
    \color{lightgray}
    \centerline{\rule{0.8\textwidth}{0pt}}
    }
    \vspace{-2em}
    \caption{Ablation of stealing probability $p_{\idlow{steal}}$ and steal buffer size, for \smq{} implemented using $d$-ary heaps, relative to wasted work, for a subset of benchmarks. Experiments are executed on the AMD machine on $256$ threads. The baseline is the Multi-Queue on $256$ threads with $C$ = 4. The fastest configuration for each benchmark is highlighted with a black border. Best viewed in color.}
    \label{fig:mq_insertbatch_deletebatch_paper}
\end{figure*}

\paragraph{Benchmarks.}
We use the real-world road and social graphs listed in Table~\ref{table:graphs} for experiments. 
The first two graphs represent the full USA road network and its West part. 
The second two graphs represent follower relations in the Twitter social network and a web crawl of .sk domain; edge weights are uniform random  in the range \texttt{[0,255]}.

We use various graph algorithms to test schedulers under different workloads: \textbf{Single-Source Shortest Paths (SSSP):} The Galois implementation of SSSP based on delta-stepping. We evaluate it on the USA and WEST road graphs, and on TWITTER and WEB social network graphs. \textbf{Breadth-First Search (BFS):} The classic  traversal algorithm a graph, where the weight of each edge is 1, evaluated on USA and WEST road graphs, and on TWITTER and WEB social network graphs.
 \textbf{A*:} This algorithm calculates the distance between two vertices, guided by the expected distance to the destination vertex from the currently visiting one. As a heuristic, the equi-rectangular approximation is used. We evaluate it on the USA and WEST road graphs.
 \textbf{Minimum Spanning Tree (MST):} We use Boruvka’s algorithm to find a spanning tree over all vertices with minimum total edge weight, with task priority equal to the degree of the associated vertex. We evaluate it on the USA and USA-WEST graphs.

\paragraph{Metrics.} 
For all runs, we record end-to-end times for the given tasks, as well as total number of tasks executed, to measure wasted work relative to the baseline. 
We use 10 repetitions, and show the average. 
(Standard deviation is low, so we omit confidence intervals for readability.)

\paragraph{Methods and Tuning.} 
We use the PMOD and OBIM baselines provided as part of Galois~\cite{Nguyen13}. 
We follow PMOD~\cite{PMOD} for parameter and benchmark setups. 
Specifically, we start from their optimized choices for the $\Delta$ and \texttt{CHUNK\_SIZE} parameters, for both OBIM and adaptive PMOD schedulers, and perform additional tuning to maximize  throughput on our setup. 
Full experimental data is presented in Appendix Section~\ref{appendix:obim_pmod_tuning}.
We also examine variants of the classic Multi-Queue, including Random Enqueue Local Dequeue (RELD), with and without the suggested buffering and locality optimizations. (In brief, we found that these optimizations, without stealing, improved the throughput of the baseline implementation by up to $3.4\times$).
 Further, we implemented the \ccmq{} proposal described in Section~\ref{sec:algorithm} both with local skip-lists, and using local heaps. 
We also implemented the buffering optimization. 

\begin{figure*}
    \centering
    \includegraphics[width=\textwidth]{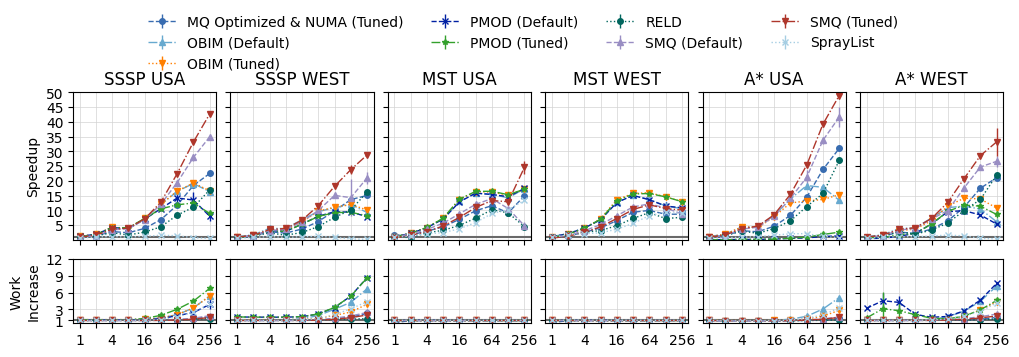}
    \includegraphics[width=\textwidth]{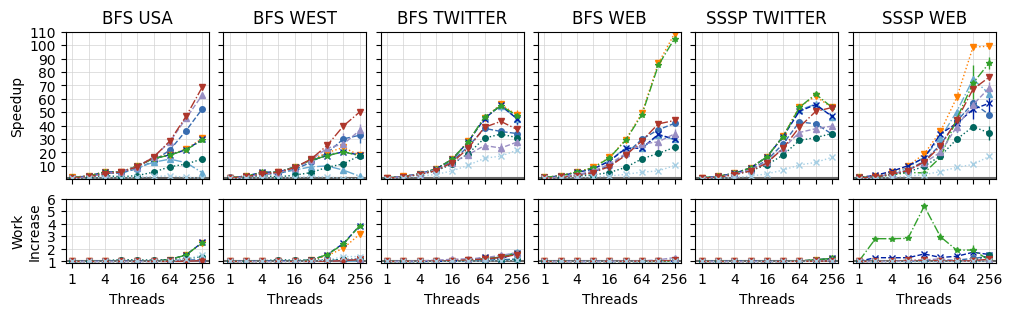}
    \caption{Comparison between the best and the default variants of SMQ, optimized classic MQ, tuned OBIM, PMOD, and other advanced schedulers on the AMD platform. Speedups are versus the baseline Multi-Queue, running on a single thread. See text for full details.} 
    \label{fig:smq_best}
\end{figure*} 

\paragraph{Impact of Optimizations and Parameters.} 
Our first set of experiments aims to examine the impact of the \emph{task batching} and \emph{temporal locality} optimizations on the performance of Multi-Queue-based proposals. We begin with the classic Multi-Queue and its variants. We found that most variants ``scale,'' in the sense that their maximal throughput is achieved at the highest number of threads. (Lower performance after 128 threads is explained by the fact that hyper-threads share their cache, thus technically halving its size.) Generally, our  optimizations significantly improve the performance of the Multi-Queue by up to $3.4\times$. Appendix~\ref{appendix:mq_optimizations} presents a set of experiments on both AMD and Intel platforms, evaluating in terms of wall-clock speedup, and total number of tasks performed, relative to the sequential baseline. 

Next, we examine the impact on the various parameter values, in particular \emph{local queue implementation}, \emph{stealing buffer size} and \emph{stealing probability}, on the performance of the \emph{Stealing Multi-Queue (SMQ).}
Figure~\ref{fig:mq_insertbatch_deletebatch_paper} presents results for the heap-based SMQ, in terms of heatmaps examining \emph{speedup} (top row) and \emph{wasted work} (work increase) for the various parameter values. (Appendix~\ref{appendix:smq:tuning} contains a full set of results on all graphs and on both AMD and Intel platforms, also examining  the skip-list variant.)
We observe that the speedups are fairly stable across parameter values, except at very high batch sizes, and that loss of performance  correlates well with increase in number of tasks (work) performed. 


\paragraph{Comparison with PMOD and OBIM.} 
Figure~\ref{fig:smq_best} presents the results of the speedup and total work comparison between the SMQ (both using heaps and skip-lists), the best-performing variant of the optimized Multi-Queue with NUMA-aware sampling, PMOD, OBIM, RELD, and Spray-List schedulers, executed on the AMD machine. 
(Please see the Appendix~\ref{appendix:final:magnified} for magnified versions of this graph on both AMD and Intel platforms.) 
We present two variants of the SMQ with respect to the tuning process. 
\texttt{SMQ (Tuned)} adopts parameter values derived via task-specific tuning (see Figure~\ref{fig:mq_insertbatch_deletebatch_paper}), while \texttt{SMQ (Default)} picks a set of reasonable default parameters ($\idlow{STEAL\_SIZE} = 4, p_{\idlow{steal}} = 1/8, K = 8$) across all benchmarks. 

We note the following. First, an examination of the throughput graphs (top) shows that the SMQ provides the highest throughput at 256(128 for Intel) threads in 10 out of the 12 experiments, and is virtually tied in the 6th experiment (MST). 
The only experiments where it has lower overall performance relative to OBIM and PMOD is the BFS experiment on social network graphs. 
The simple explanation is that, here, throughput is more important than task ordering: the task priorities are ``flat'' due to the graph having high expansion, and there is essentially no difference in terms of the total number of tasks performed by the different schedulers. 
On all other tasks, we find that the SMQ provides similar or higher performance relative to scheduling heuristics, of up to $1.84 \times$.  
Focusing on the SSSP and A* experiments, the main reason is the lower number of total tasks performed, as well as a relatively low per-operation cost due to batching and the stealing buffer implementation. Of note, the number of tasks performed appears to be near-constant with respect to the number of threads/queues.
We further emphasize that this finding is valid even in the absence of task-specific tuning. 

Next, we notice that, perhaps surprisingly, even the standard Multi-Queue can provide competitive results on a range of benchmarks, as long as it uses the optimized and NUMA-aware variants. However, its throughput does not surpass that of the SMQ, motivating the stealing optimization. 
Finally, we note that both OBIM and PMOD provide very competitive results, and good scalability in all benchmarks. 
We encourage the reader to examine the Appendix for additional experiments and ablation studies. 

\section{Discussion}

We presented an in-depth investigation of scalable priority scheduling for graph algorithms, focusing on a new Multi-Queue variant which we show to be competitive with state-of-the-art scheduling heuristics, while still providing theoretical guarantees. 
In future work, we plan to examine comparisons with alternative parallelization  approaches~\cite{dhulipala2018theoretically}, and other applications, such as   iterative machine learning algorithms e.g.~\cite{aksenov2020scalable}.


\bibliography{references}
\bibliographystyle{plain}

\appendix

\onecolumn{}

\section{Full Analysis}

\paragraph{General Approach.} Our analysis will generalize the argument of~\cite{AKLN17} to the more intricate variant of the Multi-Queue process which we consider. 
The first step in both analyses is similar: we describe coupling which equates the discrete $SMQ$ process described above to a \emph{continuous} balls into bins process.

We start with the insertion phase.
Imagine we replace the $n$ queues with $n$ bins, each of which initially contains a single ball of label $0$. 
We start by inserting infinitely many balls into bins, so that for each bin $i$ the difference between the labels of two consecutive balls
is an exponential random variable with mean $\pi_i$. 
Then, we perform  $T$ insertions in the modified $SMQ$ process (recall that elements are inserted in the increasing rank order) as follows :
for each element with rank $t \le T$, we insert it into the queue $i$ if the label(ball) with rank $t$ is inserted in the bin $i$.
Finally, we remove the labels which have rank larger than $T$ from the bins. 
Note that after the insertion phase the rank distributions of labels (balls) in bins and elements in queues are equal.
In order to show that the coupling is valid we need the following Lemma, which proves that elements are inserted into the queues 
in the same way as the original $SMQ$ process.

\begin{lemma}[Coupling, Theorem 2 in \cite{AKLN17}]  \label{lem:couplingAKLN}
Let $I_{i \leftarrow j}$ be the event that the label with rank $j$ is located in bin $i$ and let $Pr[I_{i \leftarrow j}]$
be its probability. We have that $I_{i \leftarrow j}$
is independent from $I_{i' \leftarrow j'}$
for all $j \neq j'$ and 
$Pr[I_{i \leftarrow j}]=\pi_i$.
\end{lemma}

Assuming that the number of initial insertions-$T$ is large, we describe how coupling works in a second \emph{removal phase}. 
In each removal step, we first pick a ``local'' bin $i$ with probability $\pi_i$, and then flip an additional coin to decide whether to steal or not. With probability \stealprob{}, we decide to steal. 
If so, we pick a second bin uniformly at random from among all $n$, and examine the two balls of lowest label (highest priority) on ``top'' of the two bins. 
Following the priority process, we remove the ball of lowest label among the two, uncovering the next ball in the bin. 
If the coin flip dictated that we \emph{not steal}, then we directly remove the ball on top of bin $i$. (Notice that, in both cases, the label increment for a chosen bin is exponentially-distributed).
When batch size $B>1$, we remove $B$
 labels from the bin, but for the simplicity we are going to deal with the $B=1$ case first.
 In the $SMQ$ process we make the same random choices as in the balls and bins process and follow the same removal procedure, and since the rank distributions are equal after the insertion phase, it is easy to see that they will stay equal after every removal. That is, if the ball is removed from the bin $k$ then the element is removed from the queue $k$. The \emph{rank cost} at a step is the \emph{rank} of the label of the removed ball among all labels still present in all the bins. 
Thus, to minimize this cost, we would like to always remove the ball of lowest label, but this is obviously unlikely due to the random nature of our process. However, we will show that the rank cost at each step is still well-bounded in expectation. 
Because of the coupling, this will imply that the rank cost in the $SMQ$ process is well-bounded in expectation as well, we therefore focus on
the analysis of the ranks in the balls into bins process.

Let $\ell_i(t)$ be the label on top of bin $i$ at time step $t$.
Let $x_i(t)=\ell_i(t)/n$ be the normalized value of this label, 
and let $\mu(t)=\frac{1}{n} \sum_{i=1}^n x_i(t)$ be the average normalized label at time step $t$.
As in~\cite{PTW15, AKLN17}, we will be analyzing the potential function 
\begin{equation}
    \Gamma(t)=\sum_{i=1}^n e^{-\alpha(x_i(t)-\mu(t))}+\sum_{i=1}^n e^{\alpha(x_i(t)-\mu(t))},
\end{equation} 
for a suitable constant $\alpha$, which we will define later.

We now overview the analysis of~\cite{AKLN17}, and then proceed to outline the major differences. 
A key ingredient of this analysis is the $(1+\beta)$-choice random process~\cite{PTW15}, which is similar to the continuous process we defined above, with the difference that in order to delete an element, with probability $(1-\beta)$ we choose a single bin uniformly at random and delete from it, and with probability $\beta$
we choose two bins uniformly at random and delete from the one which has a ball of lower label on top. (Thus, in expectation, we perform $(1 + \beta)$ choices at a step.)

In our analysis, we will aim to ``loosely couple'' the SMQ process with a variant of the $(1 + \beta)$-choice process, with parameters $\beta=\Omega(\gamma)$ and $\alpha=\Theta(\beta)$. As we will see, this coupling will not be exact (unlike the coupling between the discrete and continuous processes above). Yet, we will be able to use the properties of the $(1 + \beta)$-choice process to bound the properties of the SMQ process. 

Returning to the $(1 + \beta)$-choice process, 
Lemma 2 in \cite{AKLN17} shows the following potential bound, for any step $t \ge 0$:
\begin{align} \label{eqn:gammaboundstep}
\mathbb{E}[\Gamma(t+1)|x_1(t),x_2(t),...,x_n(t)] \le \Big(1-\Omega\left(\frac{\beta^2}{n}\right)\Big)\Gamma(t)+poly\left(\frac{1}{\beta}\right).
\end{align}
Assuming that Inequality~(\ref{eqn:gammaboundstep}) holds, Lemma 3 in \cite{AKLN17} proves that for any time step $t \ge 0$:
\begin{align} \label{eqn:gammaboundglobal}
\mathbb{E}[\Gamma(t)] \le O\Big(poly\left(\frac{1}{\beta}\right)n\Big).
\end{align}
Finally, given that (\ref{eqn:gammaboundglobal}) holds, Theorems 4 and 5 in \cite{AKLN17} show
that for $t \ge 0$, expected maximum rank of labels on top of bins at time step $t$ is
$O\Big(\frac{n}{\beta}(\log {n}+\log {\frac{1}{\beta}})\Big)$ and expected average rank is 
$O\Big(\frac{n}{\beta}\log {\frac{1}{\beta}}\Big)$.

Relative to this argument, we will aim to show that there exists  $\beta$ such that Equation (\ref{eqn:gammaboundstep}) holds for continuous process. This will imply  bounds on the expected average rank and expected maximum rank based on $\beta$.

We would like to point at out that we provide Lemma and Theorem numbers based on full version of \cite{AKLN17}.

\paragraph{Bounding the Potential.}
Fix an arbitrary time step $t$, and the labels  $x_1(t), x_2(t), ... x_n(t)$  on top of the bins.
Let $\Gamma_c(t+1)$ be the potential at time step $t+1$ if the label is deleted by our 
continuous process, and let $\Gamma_{\beta}(t+1)$ be the potential at time step $t+1$
if the label is deleted by $1+\beta$ process.
Our goal is to show that there exists $\beta$ such that 
\begin{align}
    \mathbb{E}[\Gamma_c(t+1) |x_1(t),x_2(t),...,x_n(t)] \le  &\mathbb{E}[\Gamma_{\beta}(t+1)||x_1(t),x_2(t),...,x_n(t)].
\end{align}

We assume that bins are sorted in the increasing order of their top labels.
Let $S_c(i)$ be the probability that we delete a label from one of the first $i$ counters in our process continuous process, and let $S_{\beta}(i)$ be the same probability for the $(1+\beta)$ process. 
We can prove the following property:
\begin{lemma} \label{lem:majorpr}
If $\gamma\Big(\frac{1}{p_{\idlow{steal}}}-1\Big) \le \frac{1}{2n}$ and $\beta=\frac{p_{\idlow{steal}}}{2(1+\gamma)}$, we get that 
for any $1 \le i \le n, S_c(i) \ge S_{\beta}(i).$   
\end{lemma}
\begin{proof}

First, notice that we have that:
\begin{align*}
S_{\beta}(i)=(1-\beta)\frac{i}{n}+\beta\frac{i^2+2i(n-i)}{n^2}=\frac{i}{n}+\beta\frac{i(n-i)}{n^2}.
\end{align*}

Further, recall that, for each bin $i$, $1-\gamma \le \frac{1}{\pi_i n} \le 1+\gamma$, for $\gamma \le 1/2$.
Here we slightly abuse the notation for $\pi_i$, since we assumed that bins are sorted in increasing label order, but this does not change the proof since we will only need to show the lower bound on $\pi_i$, which holds for every $i$.
We have that for any $i$, $\pi_i \ge \frac{1}{n(1+\gamma)}$.
Let $P_i=\sum_{j=1}^i \pi_i \ge \frac{i}{n(1+\gamma)}$.
We get that :
\begin{align*}
    S_{c}(i) &= P_i+(1-P_i)p_{\idlow{steal}}\frac{i}{n}=P_i(1-p_{\idlow{steal}}\frac{i}{n})+p_{\idlow{steal}}\frac{i}{n} \\ &\ge \frac{i}{(1+\gamma)n}(1-p_{\idlow{steal}}\frac{i}{n})+p_{\idlow{steal}}\frac{i}{n}.
\end{align*}
Thus:
\begin{align*}
    S_c(i)-S_{\beta}(i) &
    \ge
 \frac{i}{(1+\gamma)n}-p_{\idlow{steal}}\frac{i^2}{(1+\gamma)n^2} -(1-p_{\idlow{steal}})\frac{i}{n}-\beta\frac{i(n-i)}{n^2}.
\end{align*}
Hence, to complete the proof we need to show that
\begin{align} \label{eqn:betaupper}
   \beta &\le \Big(\frac{i}{(1+\gamma)n}-p_{\idlow{steal}}\frac{i^2}{(1+\gamma)n^2}-(1-p_{\idlow{steal}})\frac{i}{n}\Big)\frac{n^2}{i(n-i)} \nonumber \\ &=
\frac{n}{(n-i)}\Big(\frac{1}{1+\gamma}-(1-p_{\idlow{steal}})-p_{\idlow{steal}}\frac{i}{(1+\gamma)n}\Big).
\end{align}
Recall that $1 \le i \le n$ is an integer and for $i=n$, $S_c(i)=S_{\beta}(i)=1$.
Thus, we need to show that (\ref{eqn:betaupper}) holds for $1 \le i \le n-1$.
Next, we can prove that $\frac{n}{(n-i)}\Big(\frac{1}{1+\gamma}-(1-p_{\idlow{steal}})-p_{\idlow{steal}}\frac{i}{(1+\gamma)n}\Big)$ is minimized for $i=n-1$.
Thus, after plugging $i=n-1$ in (\ref{eqn:betaupper}), we need to show that 
\begin{align*}
   \beta &\le
n\Big(\frac{1}{1+\gamma}-(1-p_{\idlow{steal}})-p_{\idlow{steal}}\frac{(n-1)}{(1+\gamma)n}\Big) =p_{\idlow{steal}}\frac{n}{1+\gamma}\Bigg(\frac{1}{n}-\gamma\Big(\frac{1}{p_{\idlow{steal}}}-1\Big)\Bigg) .
\end{align*}

We can now set $\gamma\Big(\frac{1}{p_{\idlow{steal}}}-1\Big) \le \frac{1}{2n}$ and $\beta=\frac{p_{\idlow{steal}}}{2(1+\gamma)}$, and the above inequality holds.
This completes the proof of the lemma.
\end{proof}

We can now show that the potential bound (3) holds. Formally:
\begin{lemma} \label{lem:couplingstep}
Fix any time step $t$ and labels $x_1(t),x_2(t),...,x_n(t)$. Let $\gamma\Big(\frac{1}{p_{\idlow{steal}}}-1\Big) \le \frac{1}{2n}$ and $\beta=\frac{p_{\idlow{steal}}}{2(1+\gamma)}$. Also, let $w_t$ be the random weight which we use to generate new labels for both continuous and $1+\beta$ processes. Then: 
\begin{align}
    \mathbb{E}[\Gamma_{c}(t+1)-\Gamma_{\beta}(t+1)|x_1(t),x_2(t),...,x_n(t),w_t] \le 0.
\end{align}
\end{lemma}
\begin{proof}
To show the proof of the lemma we use the coupling similar to the one used in \cite{PTW15}, Theorem 3.1. We again assume that bins are sorted in the increasing label order. The coupling works as follows:
At step $t$ we pick probability $0 \le p < 1$ uniformly at random,
for our continuous process we delete label from bin $i$, such that $S_{c}(i-1) \le p < S_c(i)$
(we assume that $S_c(0)=0$) and for $1+\beta$ process we delete label from bin $j$, if 
$S_{\beta}(j-1) \le p < S_{\beta}(j)$. 
We set $x_{c,i}(t+1)=x_i(t)+w_t/n$ and $x_{\beta,j}(t+1)=x_j(t)+w_t/n$.
Here $x_{c,i}(t+1)$ and $x_{\beta,j}(t+1)$ are new normalized labels on top of bins $i$ and $j$
for our continuous and $1+\beta$ processes correspondingly (The rest of the labels do not change).
It is straightforward to verify that coupling is valid, since $S_c(i)-S_c(i-1)$ is exactly
the probability of deleting label from bin $i$ in our continuous process (the same thing is valid for $1+\beta$ process). We would like to note the we are not able to use Theorem 3.1 in \cite{PTW15} directly since it assumes that $w_t=1$.
Lemma \ref{lem:majorpr} gives us that $i \le j$, and hence $x_i(t) \le x_j(t)$
Recall that $\mu(t)=\frac{1}{n} \sum_{k=1}^n x_k(t)$, for both processes $\mu(t+1)=\mu(t)+w_t/n^2$.
and for $k \neq i,j$, $x_k(t+1)=x_k(t)$, hence to prove that 
potential at step $t$ is smaller for continuous process we just need to check 
at how new labels of bins $i$ and $j$ effect potentials.

First we show that 
\begin{align*}
    e^{\alpha(x_i(t)+w_t/n-w_t/n^2-\mu(t))}+e^{\alpha(x_j(t)-w_t/n^2-\mu(t))} \le e^{\alpha(x_i(t)-w_t/n^2-\mu(t))}+e^{\alpha(x_j(t)+w_t-w_t/n^2-\mu(t))}.
\end{align*}
which after diving both sides by $e^{\alpha(-w_t/n^2-\mu(t))} \ge 0$ is the same as
\begin{align*}
    e^{\alpha(x_i(t)+w_t/n)}+e^{\alpha(x_j(t))} \le e^{\alpha(x_i(t))}+e^{\alpha(x_j(t)+w_t/n)}.
\end{align*}
After rearranging terms, this can be rewritten as
\begin{align*}
    (e^{\alpha w_t/n}-1)(e^{\alpha(x_j(t))}-e^{\alpha(x_i(t))}) \ge 0.
\end{align*}
The above inequality holds since $\alpha \ge 0$, $w_t \ge 0$ and $x_j(t) \ge x_i(t)$.
Next, we show that

\begin{align*}
    e^{-\alpha(x_i(t)+w_t/n-w_t/n^2-\mu(t))}+e^{-\alpha(x_j(t)-w_t/n^2-\mu(t))} \le e^{-\alpha(x_i(t)-w_t/n^2-\mu(t))}+e^{-\alpha(x_j(t)+w_t-w_t/n^2-\mu(t))}.
\end{align*}
which is the same as
\begin{align*}
    e^{-\alpha(x_i(t)+w_t/n)}+e^{-\alpha(x_j(t))} \le e^{-\alpha(x_i(t))}+e^{-\alpha(x_j(t)+w_t/n)}.
\end{align*}
Rearranging terms, this can be rewritten as
\begin{align*}
    (e^{-\alpha w_t/n}-1)(e^{-\alpha(x_j(t))}-e^{-\alpha(x_i(t))}) \ge 0.
\end{align*}
The above inequality clearly holds since $\alpha \ge 0$, $w_t \ge 0$ and $x_j(t) \ge x_i(t)$ (in this case, both terms are negative).
\end{proof}

\paragraph{Proof of Theorem~\ref{thm:main}.} Finally, we can prove our main result.
First, fix normalized labels on the top of the bins: \\$x_1(t), x_2(t), ..., x_n(t)$.
Lemma 2 in full version of \cite{AKLN17} shows that 
for any step $t \ge 0$:
\begin{align} \label{eqn:gammaperstepthm}
\mathbb{E}[\Gamma_{\beta}(t+1)|x_1(t),x_2(t),...,x_n(t)] \le \Big(1-\Omega(\frac{\beta^2}{n})\Big)\Gamma(t) +poly(\frac{1}{\beta}).
\end{align}
and Lemma \ref{lem:couplingstep} gives us that for $\beta=\frac{p_{\idlow{steal}}}{2(1+\gamma)}$
\begin{align*}
\mathbb{E}[\Gamma_{c}(t+1)-\Gamma_{\beta}(t+1)|x_1(t),x_2(t),...,x_n(t),w_t] \le 0.
\end{align*}
We remove conditioning on $w_t$ and slightly abuse the notation in $1+\beta$ process:
we assume that $w_t$ is $Exp(\pi_i)$ if continuous process deletes from bin $i$,
this slightly changes $1+\beta$ process (which might delete from different bin), but proof of (\ref{eqn:gammaperstepthm}) in \cite{AKLN17}
will still be correct since it only uses that $w_t=Exp(\pi_i)$, (To be more precise, the proof uses expectation 
and moment generating function of exponential random variable)
for $1-\gamma \le \frac{1}{\pi_i n} \le 1+\gamma$, and $\gamma \le 1/2$ and does not make any assumptions about $i$.
Hence we get that 
\begin{align*}
\mathbb{E}[\Gamma_{c}(t+1)-\Gamma_{\beta}(t+1)|x_1(t),x_2(t),...,x_n(t)] \le 0.
\end{align*}
and 
this means that 
\begin{align} 
\mathbb{E}[\Gamma_{c}(t+1)|x_1(t),x_2(t),...,x_n(t)] &\le \Big(1-\Omega(\frac{\beta^2}{n})\Big)\Gamma(t)\nonumber +poly(\frac{1}{\beta}).
\end{align}
Lemma 3 in the full version of \cite{AKLN17} proves that for any time step $t 
\ge 0$:
\begin{align} \label{eqn:gammaboundglobalthm}
\mathbb{E}[\Gamma_c(t)] \le O\Big(poly(\frac{1}{\beta})n\Big).
\end{align}
Finally, we can use Theorems 4 and 5 in the full version of \cite{AKLN17} to show
that for the continuous process and time  $t \ge 0$, the expected maximum rank of labels on top of bins at time step $t$ is
$O\Big(\frac{n}{\beta}(\log {n}+\log {\frac{1}{\beta}})\Big)$ and expected average rank is 
$O\Big(\frac{n}{\beta}\log {\frac{1}{\beta}}\Big)$. This is possible since these theorems only use the upper bound
on $\mathbb{E}[\Gamma(t)]$, regardless of which process we used to derive this bound).
Plugging in $\beta=\frac{p_{\idlow{steal}}}{2(1+\gamma)}$ in the above rank bounds and
using Lemma \ref{lem:couplingAKLN} completes the proof of the theorem for $B=1$.
We proceed by specifying what will change in the proof if $B>1$.
Lemma $\ref{lem:couplingAKLN}$ holds as before.
In the continuous process we delete $B$ labels instead of just one.
This means that if continuous process deletes $B$ labels from bin $i$ at step $t$.
We have that $\ell_i(t+1)=\ell_i(t)+\sum_{k=1}^B Exp(\pi_i)$ (Recall that $\ell_i(t)$ is a label
on top of bin $i$ at step $t$). Hence, we use a normalized label: $x_i(t)=\frac{\ell_i(t)}{Bn}$.
Consider random variable $\sum_{k=1}^B \frac{Exp(\pi_i)}{B}$ (this is by how much normalized label on top of bin $i$ increases, if we ignore factor of $1/n$). We have that $\mathbb{E}\Big[\sum_{k=1}^B \frac{Exp(\pi_i)}{B}\Big]=\mathbb{E}\Big[Exp(\pi_i)\Big]$
Also, by convexity of exponential function and Jensen's inequality we have that 
the moment generating function of $\sum_{k=1}^B Exp(\pi_i)/B$ is upper bounded by moment
generating function of $Exp(\pi_i)$ (This is important since the proof of (\ref{eqn:gammaperstepthm}) in \cite{AKLN17} uses
moment generating function). This means that we can apply all the steps of the proof exactly as in the case of $B=1$.
The only difference will be last step, where we scale back by $1/nB$ instead of $1/n$.
Hence bounds on expected maximum rank and expected average rank become $O\Big(\frac{nB}{\beta}(\log {n}+\log {\frac{1}{\beta}})\Big)$ and 
$O\Big(\frac{nB}{\beta}\log {\frac{1}{\beta}}\Big)$.

\clearpage
\section{Tuning OBIM and PMOD Schedulers}\label{appendix:obim_pmod_tuning}
\begin{figure*}[h]
    \centering
    \includegraphics[width=1\textwidth]{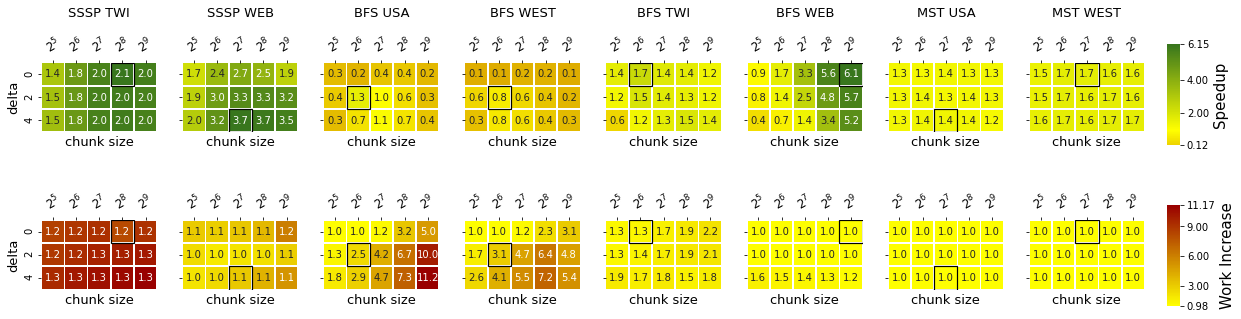}
    \includegraphics[width=0.48\textwidth]{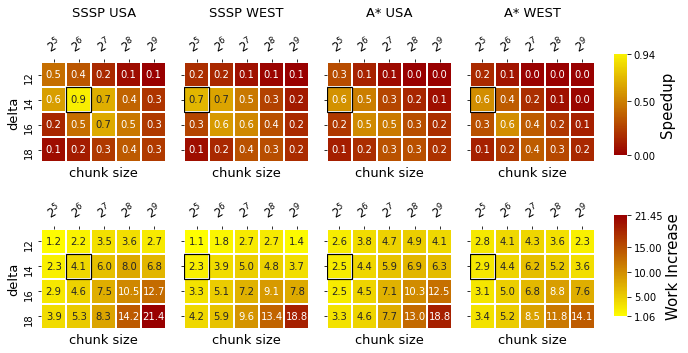}
    \vspace{-1em}
    \caption{Ablation of \obimdel{} and \chunksize{} parameters for OBIM. Experiments execute on $256$ threads on the \textbf{AMD} platform. The baseline is the classic Multi-Queue on $256$ threads with $C$ = 4. The fastest configuration for each benchmark is highlighted with a black border. Best viewed in color.}
    \label{fig:obim_delta_amd}
\end{figure*}

\begin{figure*}[h]
    \centering
    \includegraphics[width=1\textwidth]{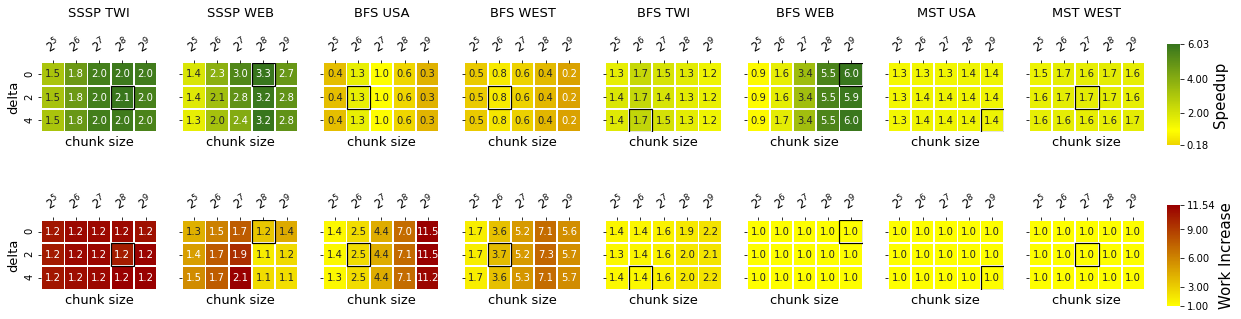}
    \includegraphics[width=0.48\textwidth]{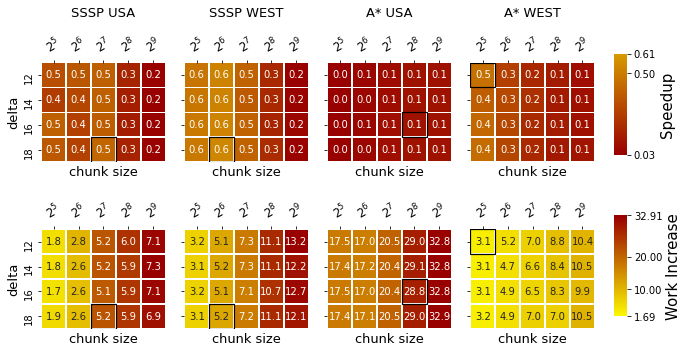}
    \vspace{-1em}
    \caption{Ablation of \obimdel{} and \chunksize{} parameters for PMOD. Experiments execute on $256$ threads on the \textbf{AMD} platform. The baseline is the classic Multi-Queue on $256$ threads with $C$ = 4. The fastest configuration for each benchmark is highlighted with a black border. Best viewed in color.}
    \label{fig:pmod_delta_amd}
\end{figure*}

\begin{figure*}[h]
    \centering
    \includegraphics[width=1\textwidth]{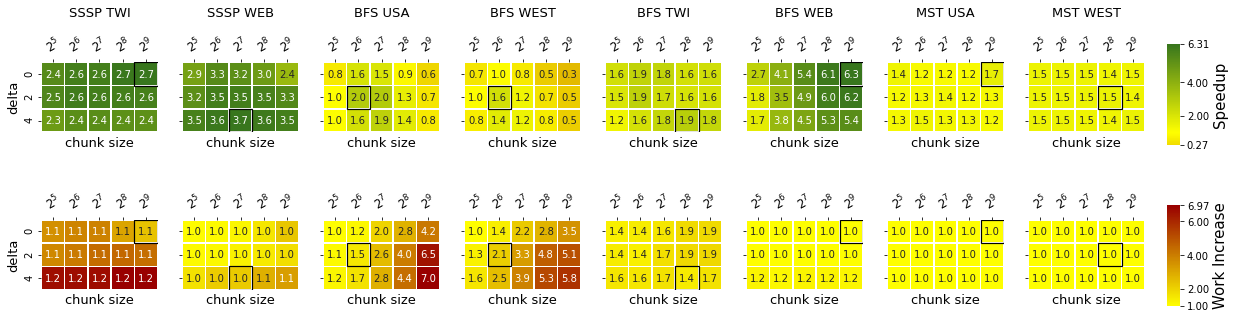}
    \includegraphics[width=0.48\textwidth]{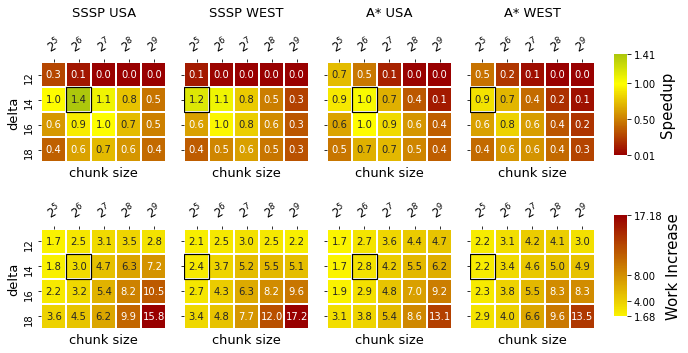}
    \vspace{-1em}
    \caption{Ablation of \obimdel{} and \chunksize{} parameters for OBIM. Experiments execute on $128$ threads on the \textbf{Intel} platform. The baseline is the classic Multi-Queue on $128$ threads with $C$ = 4. The fastest configuration for each benchmark is highlighted with a black border. Best viewed in color.}
    \label{fig:obim_delta_intel}
\end{figure*}

\begin{figure*}[h]
    \centering
    \includegraphics[width=1\textwidth]{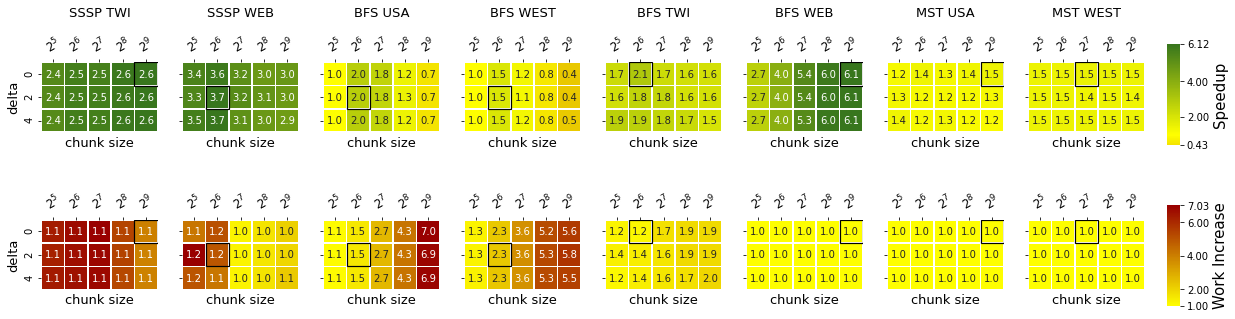}
    \includegraphics[width=0.48\textwidth]{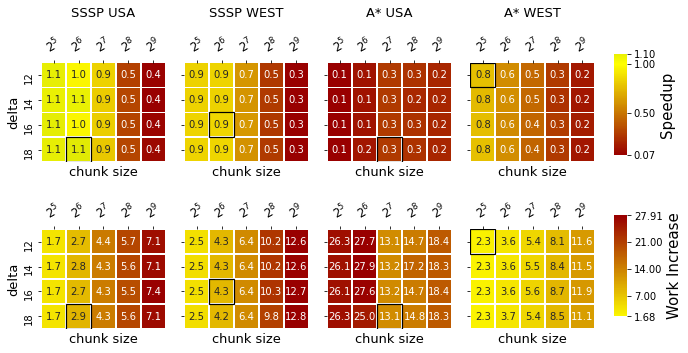}
    \vspace{-1em}
    \caption{Ablation of \obimdel{} and \chunksize{} parameters for PMOD. Experiments execute on $128$ threads on the \textbf{Intel} platform. The baseline is the classic Multi-Queue on $128$ threads with $C$ = 4. The fastest configuration for each benchmark is highlighted with a black border. Best viewed in color.}
    \label{fig:pmod_delta_intel}
\end{figure*}


\clearpage
\onecolumn
\section{Classic Multi-Queue Optimizations} \label{appendix:mq_optimizations}
Usually, Multi-Queues use $C \times T$ sequential queues under the hood, where $T$ is the number of threads and $C$ is some constant between $2$ and $8$. We consider different $C$ values to find the best one, and present the corresponding experimental data in Tables~\ref{table:c_amd}--\ref{table:c_intel}.

After that, we evaluate the \emph{task batching} and \emph{temporal locality} optimizations in Subsection~\ref{subsec:optimizations}. In short, task batching reduces the ratio between synchronization cost and task execution time by retrieving multiple tasks from the same queue at once in \del{} and buffering them in \ins{}. The temporal locality optimization reduces the cache coherence overhead by using the same queue for a sequence of \del{} or \ins{} operations, changing the ``temporally local'' queue with a constant probability. We evaluate all the four combinations of these optimizations, and present the corresponding results in Figures~\ref{fig:mq_insertprob_deleteprob}--\ref{fig:mq_insertbatch_deletebatch2}.  In addition, we compare the optimally configured combinations {---} the results are presented in Figures~\ref{fig:mq_combinations}--\ref{fig:mq_combinations2}. All the experiments are executed on both AMD and Intel platforms.

\begin{table}[h]
\small
\begin{center}
\begin{tabular}{ |c|c|c|c|c|c|c|c|c|c|c|c| }
\hline
 & \normalsize{\textbf{2}} & \normalsize{\textbf{3}} & \normalsize{\textbf{4}} & \normalsize{\textbf{5}} & \normalsize{\textbf{6}} & \normalsize{\textbf{7}} & \normalsize{\textbf{8}} \\
\hline
\normalsize{\textbf{BFS USA}} & 17.18 & 21.33 & 23.41 & 25.76 & 27.28 & \color{Numa}{\textbf{28.27}} & 27.44 \\
\hline
\normalsize{\textbf{BFS WEST}} & 17.64 & 20.31 & 20.69 & 22.05 & 23.10 & 23.65 & \color{Numa}{\textbf{23.77}} \\
\hline
\normalsize{\textbf{BFS TWI}} & \color{Numa}{\textbf{30.20}} & 30.06 & 29.28 & 27.81 & 29.09 & 28.98 & 28.29 \\
\hline
\normalsize{\textbf{BFS WEB}} & 15.95 & 16.41 & 17.65 & 18.03 & 18.06 & \color{Numa}{\textbf{18.46}} & 18.42 \\
\hline
\normalsize{\textbf{SSSP USA}} & 17.83 & \color{Numa}{\textbf{18.26}} & 17.53 & 18.10 & 17.35 & 16.95 & 16.34 \\
\hline
\normalsize{\textbf{SSSP WEST}} & 13.18 & \color{Numa}{\textbf{13.80}} & 13.52 & 12.67 & 11.40 & 11.74 & 11.29 \\
\hline
\normalsize{\textbf{SSSP TWI}} & \color{Numa}{\textbf{27.01}} & 26.86 & 26.51 & 24.98 & 24.14 & 25.11 & 25.03 \\
\hline
\normalsize{\textbf{SSSP WEB}} & 26.70 & 26.70 & 27.30 & 26.91 & 28.93 & \color{Numa}{\textbf{29.15}} & \color{Numa}{\textbf{29.15}} \\
\hline
\normalsize{\textbf{MST USA}} & 13.04 & 11.97 & 12.55 & 11.46 & 12.71 & \color{Numa}{\textbf{13.25}} & 12.87 \\
\hline
\normalsize{\textbf{MST WEST}} & 7.63 & 7.38 & 7.73 & 7.73 & \color{Numa}{\textbf{7.87}} & 7.22 & 7.40 \\
\hline
\normalsize{\textbf{A* USA}} & 23.47 & 24.74 & \color{Numa}{\textbf{26.27}} & 26.12 & 25.57 & 24.79 & 23.96 \\
\hline
\normalsize{\textbf{A* WEST}} & 18.01 & \color{Numa}{\textbf{19.52}} & 18.78 & 18.31 & 17.57 & 16.64 & 16.06 \\
\hline
\end{tabular}
\end{center}
\vspace{0.5em}
\caption{Speedup of the classic Multi-Queue with various $C$ executed on $256$ threads on the \textbf{AMD} platform. The baseline is sequential priority queue execution on a single thread. The best speedups are highlighted with {\color{Numa}{\textbf{\numacol{}}}}.}
\label{table:c_amd}
\end{table}

\begin{table}[h]
\small
\begin{center}
\begin{tabular}{ |c|c|c|c|c|c|c|c|c|c|c|c| }
\hline
 & \normalsize{\textbf{2}} & \normalsize{\textbf{3}} & \normalsize{\textbf{4}} & \normalsize{\textbf{5}} & \normalsize{\textbf{6}} & \normalsize{\textbf{7}} & \normalsize{\textbf{8}} \\
\hline
\normalsize{\textbf{BFS USA}} & 9.45 & 11.65 & 12.70 & 13.49 & 13.74 & \color{Numa}{\textbf{13.77}} & \color{Numa}{\textbf{13.77}} \\
\hline
\normalsize{\textbf{BFS WEST}} & 9.35 & 11.00 & 12.19 & 12.76 & 13.01 & 12.95 & \color{Numa}{\textbf{13.60}} \\
\hline
\normalsize{\textbf{BFS TWI}} & 29.00 & 29.04 & 28.62 & 26.81 & \color{Numa}{\textbf{29.12}} & 28.99 & 27.91 \\
\hline
\normalsize{\textbf{BFS WEB}} & 12.32 & 13.05 & 13.52 & 13.94 & 13.91 & 14.88 & \color{Numa}{\textbf{15.03}} \\
\hline
\normalsize{\textbf{SSSP USA}} & 10.40 & 12.68 & 13.49 & \color{Numa}{\textbf{13.53}} & 13.35 & 13.04 & 13.19 \\
\hline
\normalsize{\textbf{SSSP WEST}} & 9.45 & 9.85 & \color{Numa}{\textbf{10.96}} & 10.83 & 10.53 & 10.04 & 9.74 \\
\hline
\normalsize{\textbf{SSSP TWI}} & 27.30 & 26.98 & \color{Numa}{\textbf{27.64}} & 27.30 & 26.36 & 26.30 & 26.90 \\
\hline
\normalsize{\textbf{SSSP WEB}} & 23.65 & 24.58 & 25.17 & 25.22 & 23.79 & 26.13 & \color{Numa}{\textbf{26.51}} \\
\hline
\normalsize{\textbf{MST USA}} & 4.50 & 4.73 & 4.66 & 4.44 & 4.58 & 4.73 & \color{Numa}{\textbf{4.86}} \\
\hline
\normalsize{\textbf{MST WEST}} & 4.78 & 4.66 & \color{Numa}{\textbf{4.80}} & 4.46 & 4.59 & 4.33 & 4.75 \\
\hline
\normalsize{\textbf{A* USA}} & 14.64 & 17.42 & 17.73 & 18.31 & \color{Numa}{\textbf{18.71}} & 18.58 & 18.44 \\
\hline
\normalsize{\textbf{A* WEST}} & 12.29 & 14.10 & 14.55 & \color{Numa}{\textbf{14.59}} & 14.25 & 14.01 & 13.60 \\
\hline
\end{tabular}
\end{center}
\vspace{0.5em}
\caption{Speedup of the classic Multi-Queue with various $C$ executed on $128$ threads on the \textbf{Intel} platform. The baseline is sequential priority queue execution on a single thread. The best speedups are highlighted with {\color{Numa}{\textbf{\numacol{}}}}.}
\label{table:c_intel}
\end{table}

\clearpage
\subsection{Classic Multi-Queue Optimizations on AMD: insert=Temporal Locality, delete=Temporal Locality } 

\begin{figure*}[h]
    \centering
    \includegraphics[width=0.95\textwidth]{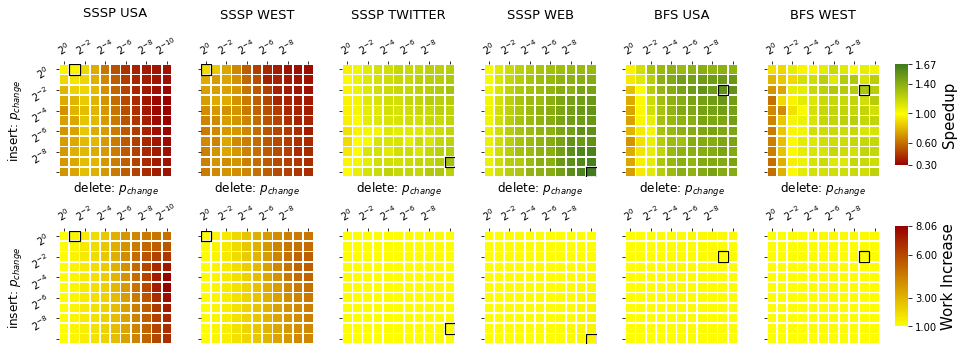}
    \includegraphics[width=0.95\textwidth]{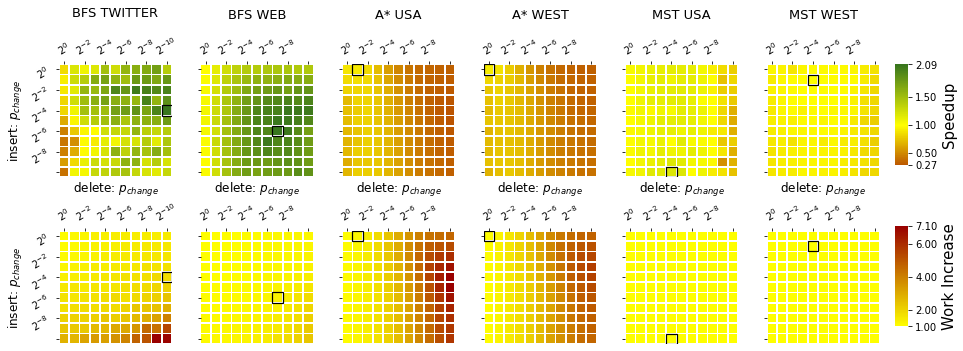}
    \vspace{-1em}
    \caption{Ablation of queue change probabilities (Temporal Locality) for both \ins{} and \del{}. Experiments execute on $256$ threads on the \textbf{AMD} platform. The baseline is the classic Multi-Queue on $256$ threads with $C$ = 4. The fastest configuration for each benchmark is highlighted with a black border and listed in Table~\ref{table:mqpp_amd}.}
    \label{fig:mq_insertprob_deleteprob}
\end{figure*}

\begin{table}[h]
\small
\begin{tabular}{|c|c|c|c|c|c|c| }
\hline
 & \normalsize{\textbf{SSSP USA}}&\normalsize{\textbf{SSSP WEST}} &\normalsize{\textbf{SSSP TWITTER}}&\normalsize{\textbf{SSSP WEB}} 
 & \normalsize{\textbf{BFS USA}} & \normalsize{\textbf{BFS WEST}} 
 \\
\hline
\insprob{} & 1/1 & 1/1 & 1/512 & 1/1024 & 1/4 & 1/4 \\
\hline
\delprob{} & 1/2 & 1/1 & 1/1024 & 1/1024 & 1/512 & 1/512 \\
\hline
\speed{} & 0.98 & 0.91 & 1.32 & 1.67 & 1.56 & 1.29 \\
\hline
\workinc{} & 1.24 & 1.04 & 1.07 & 1.31 & 1.07 & 1.07 \\
\hline
& \normalsize{\textbf{BFS TWITTER}} & \normalsize{\textbf{BFS WEB}} & \normalsize{\textbf{A* USA}} & \normalsize{\textbf{A* WEST}} & \normalsize{\textbf{MST USA}} & \normalsize{\textbf{MST WEST}} \\
\hline
\insprob{} & 1/16 & 1/64 & 1/1 & 1/1 & 1/1024 & 1/2 \\
\hline
\delprob{} & 1/1024 & 1/128 & 1/2 & 1/1 & 1/16 & 1/16 \\
\hline
\speed{} & 1.19 & 2.09 & 0.91 & 0.92 & 1.19 & 1.01 \\
\hline
\workinc{} & 1.03 & 1.01 & 1.19 & 1.03 & 1.00 & 1.00 \\
\hline
\end{tabular}
\vspace{0.3em}
\caption{The optimal parameters for Multi-Queue with the \emph{temporal locality} optimization for both \ins{} and \del{} obtained on the \textbf{AMD} platform. Based on Figure~\ref{fig:mq_insertprob_deleteprob}. For each benchmark, the best parameters are presented with the speedup and work increase. }
\label{table:mqpp_amd}
\end{table}

\newpage
\subsection{Classic Multi-Queue Optimizations on Intel: insert=Temporal Locality, delete=Temporal Locality }

\begin{figure*}[h]
    \centering
    \includegraphics[width=0.95\textwidth]{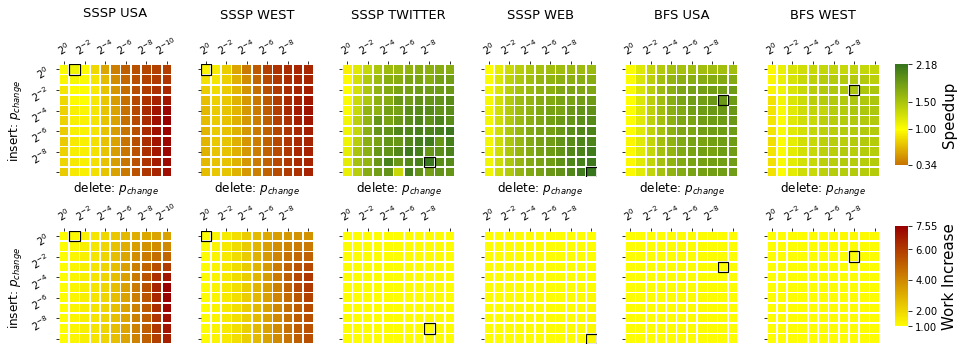}
    \includegraphics[width=0.95\textwidth]{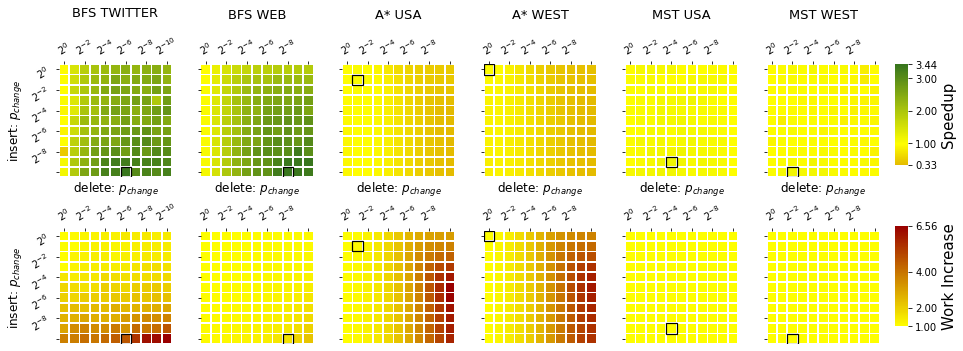}
    \vspace{-1em}
    \caption{Ablation of queue change probabilities (Temporal Locality) for both \ins{} and \del{}.  Experiments execute on $128$ threads on the \textbf{Intel} platform. The baseline is the classic Multi-Queue on $128$ threads with $C$ = 4. The fastest configuration for each benchmark is highlighted with a black border and listed in Table~\ref{table:mqpp_intel}.}
    \label{fig:mq_insertprob_deleteprob2}
\end{figure*}

\begin{table}[h]
\small
\begin{tabular}{ |c|c|c|c|c|c|c| }
\hline
 & \normalsize{\textbf{SSSP USA}}&\normalsize{\textbf{SSSP WEST}} &\normalsize{\textbf{SSSP TWITTER}}&\normalsize{\textbf{SSSP WEB}} 
 & \normalsize{\textbf{BFS USA}} & \normalsize{\textbf{BFS WEST}} \\
\hline
\insprob{} & 1/1 & 1/1 & 1/512 & 1/1024 & 1/8 & 1/4 \\
\hline
\delprob{} & 1/2 & 1/1 & 1/256 & 1/1024 & 1/512 & 1/256 \\
\hline
\speed{} & 1.05 & 0.97 & 1.73 & 2.18 & 1.88 & 1.50 \\
\hline
\workinc{} & 1.15 & 1.04 & 1.01 & 1.15 & 1.06 & 1.07 \\
\hline 
& \normalsize{\textbf{BFS TWITTER}} & \normalsize{\textbf{BFS WEB}} & \normalsize{\textbf{A* USA}} & \normalsize{\textbf{A* WEST}} & \normalsize{\textbf{MST USA}} & \normalsize{\textbf{MST WEST}} \\
\hline
\insprob{} & 1/1024 & 1/1024 & 1/2 & 1/1 & 1/512 & 1/1024 \\
\hline
\delprob{} & 1/64 & 1/256 & 1/2 & 1/1 & 1/16 & 1/4 \\
\hline
\speed{} & 1.60 & 3.44 & 1.06 & 0.98 & 1.20 & 1.16 \\
\hline
\workinc{} & 1.09 & 1.02 & 1.11 & 1.02 & 1.00 & 1.01 \\
\hline
\end{tabular}
\vspace{0.3em}
\caption{The optimal parameters for Multi-Queue with the \emph{temporal locality} optimization for both \ins{} and \del{} obtained on the \textbf{Intel} platform.
Based on Figure~\ref{fig:mq_insertprob_deleteprob2}. 
For each benchmark, the best parameters are presented with the speedup and work increase. }
\label{table:mqpp_intel}
\end{table}

\newpage
\subsection{Classic Multi-Queue Optimizations on AMD: insert=Temporal Locality, delete=Task Batching }

\begin{figure*}[h]
    \centering
    \includegraphics[width=0.95\textwidth]{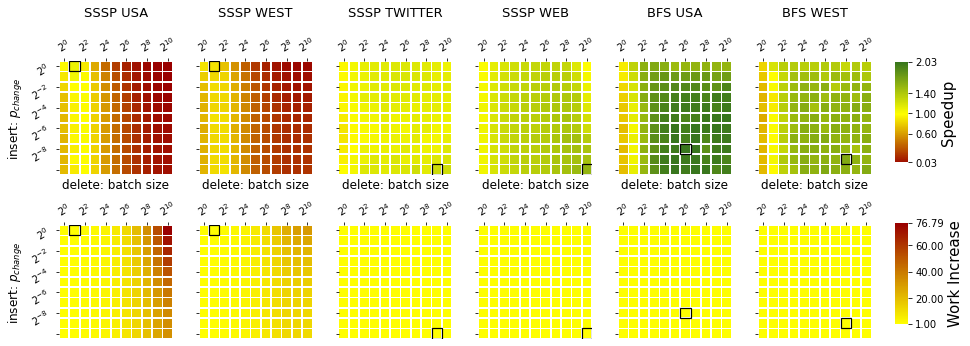}
    \includegraphics[width=0.95\textwidth]{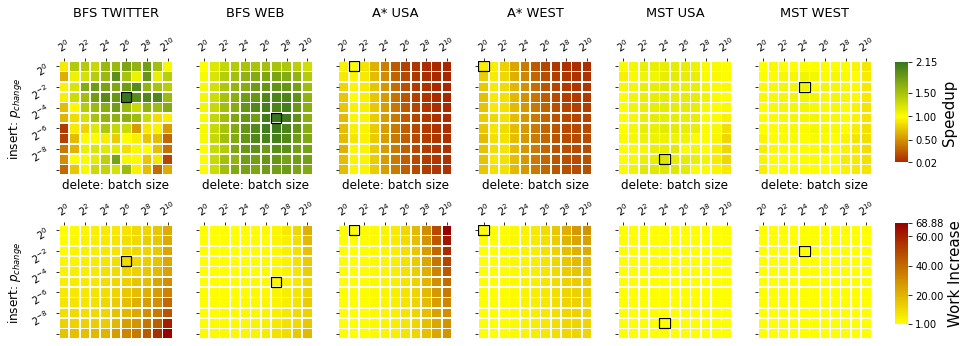}
    \vspace{-1em}
    \caption{Ablation of queue change probability (Temporal Locality) for \ins{} and batch size for \del{}. Experiments execute on $256$ threads on the \textbf{AMD} platform. The baseline is the classic Multi-Queue on $256$ threads with $C$ = 4. The fastest configuration for each benchmark is highlighted with a black border and listed in Table~\ref{table:mqpl_amd}.}
    \label{fig:mq_insertprob_deletebatch}
\end{figure*}

\begin{table}[h]
\small
\begin{tabular}{ |c|c|c|c|c|c|c| }
\hline
& \normalsize{\textbf{SSSP USA}} & \normalsize{\textbf{SSSP WEST}} & \normalsize{\textbf{SSSP TWITTER}} & \normalsize{\textbf{SSSP WEB}} & \normalsize{\textbf{BFS USA}} & \normalsize{\textbf{BFS WEST}} \\
\hline
\insprob{} & 1/1 & 1/1 & 1/1024 & 1/1024 & 1/256 & 1/512 \\
\hline
\delbatch{} 
& 2 & 2 & 512 & 1024 & 64 & 256 \\
\hline
\speed{} & 1.09 & 0.90 & 1.25 & 1.53 & 2.03 & 1.65 \\
\hline
\workinc{} & 1.28 & 1.35 & 1.10 & 1.40 & 1.08 & 1.08 \\
\hline
& \normalsize{\textbf{BFS TWITTER}} & \normalsize{\textbf{BFS WEB}} & \normalsize{\textbf{A* USA}} & \normalsize{\textbf{A* WEST}} & \normalsize{\textbf{MST USA}} & \normalsize{\textbf{MST WEST}} \\
\hline
\insprob{} & 1/8 & 1/32 & 1/1 & 1/1 & 1/512 & 1/4 \\
\hline
\delbatch{} 
& 64 & 128 & 2 & 1 & 16 & 16 \\
\hline
\speed{} & 1.15 & 2.15 & 1.00 & 0.93 & 1.20 & 1.09 \\
\hline
\workinc{} & 1.05 & 1.01 & 1.22 & 1.03 & 1.00 & 1.00 \\
\hline
\end{tabular}
\vspace{0.3em}
\caption{The optimal parameters for Multi-Queue with the \emph{temporal locality} optimization for \ins{} and \emph{task batching} for \del{}, obtained on the \textbf{AMD} platform. Based on Figure~\ref{fig:mq_insertprob_deletebatch}. For each benchmark, the best parameters are presented with the speedup and work increase. }
\label{table:mqpl_amd}
\end{table}

\newpage
\subsection{Classic Multi-Queue Optimizations on Intel: insert=Temporal Locality, delete=Batching }
\begin{figure*}[h]
    \centering
    \includegraphics[width=0.95\textwidth]{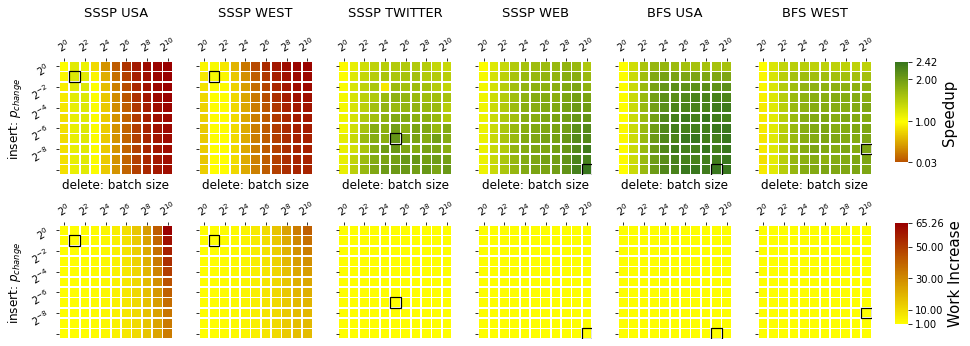}
    \includegraphics[width=0.95\textwidth]{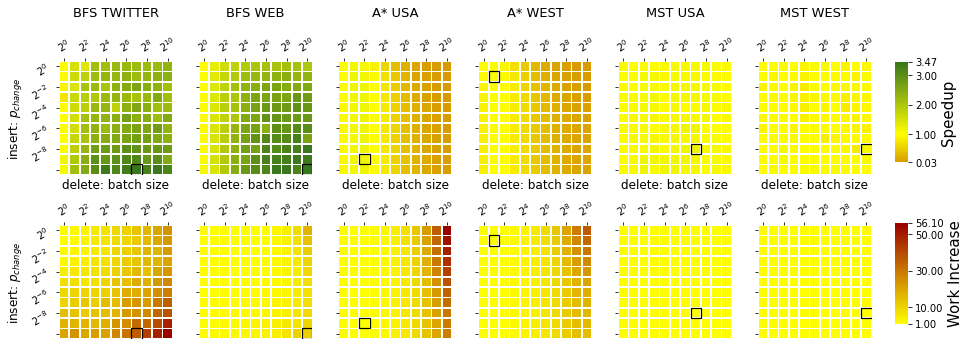}
    \vspace{-1em}
    \caption{Ablation of queue change probability (Temporal Locality) for \ins{} and batch size for \del{}. Experiments execute on $128$ threads on the \textbf{Intel} platform. The baseline is the classic Multi-Queue on $128$ threads with $C$ = 4. The fastest configuration for each benchmark is highlighted with a black border and listed in Table~\ref{table:mqpl_intel}.}
    \label{fig:mq_insertprob_deletebatch2}
\end{figure*}

\begin{table}[h]
\small
\begin{tabular}{ |c|c|c|c|c|c|c| }
\hline
 & \normalsize{\textbf{SSSP USA}} & \normalsize{\textbf{SSSP WEST}} & \normalsize{\textbf{SSSP TWITTER}} & \normalsize{\textbf{SSSP WEB}} & \normalsize{\textbf{BFS USA}} & \normalsize{\textbf{BFS WEST}} \\
\hline
\insprob{} & 1/2 & 1/2 & 1/128 & 1/1024 & 1/1024 & 1/256 \\
\hline
\delbatch{} & 2 & 2 & 32 & 1024 & 512 & 1024 \\
\hline
\speed{} & 1.16 & 1.03 & 1.76 & 2.15 & 2.42 & 1.93 \\
\hline
\workinc{} & 1.19 & 1.30 & 1.02 & 1.22 & 1.07 & 1.08 \\
\hline & \normalsize{\textbf{BFS TWITTER}} & \normalsize{\textbf{BFS WEB}} & \normalsize{\textbf{A* USA}} & \normalsize{\textbf{A* WEST}} & \normalsize{\textbf{MST USA}} & \normalsize{\textbf{MST WEST}} \\
\hline
\insprob{} & 1/1024 & 1/1024 & 1/512 & 1/2 & 1/256 & 1/256 \\
\hline
\delbatch{} & 128 & 1024 & 4 & 2 & 128 & 1024 \\
\hline
\speed{} & 1.64 & 3.47 & 1.18 & 1.04 & 1.20 & 1.12 \\
\hline
\workinc{} & 1.12 & 1.04 & 1.34 & 1.21 & 1.00 & 1.00 \\
\hline
\end{tabular}
\vspace{0.3em}
\caption{The optimal parameters for Multi-Queue with the \emph{temporal locality} optimization for \ins{} and \emph{task batching} for \del{}, obtained on \textbf{Intel}. Based on Figure~\ref{fig:mq_insertprob_deletebatch2}. For each benchmark, the best parameters are presented with the speedup and work increase.}
\label{table:mqpl_intel}
\end{table}

\newpage
\subsection{Classic Multi-Queue Optimizations on AMD: insert=Batching, delete=Temporal Locality }

\begin{figure*}[h]
    \centering
    \includegraphics[width=0.95\textwidth]{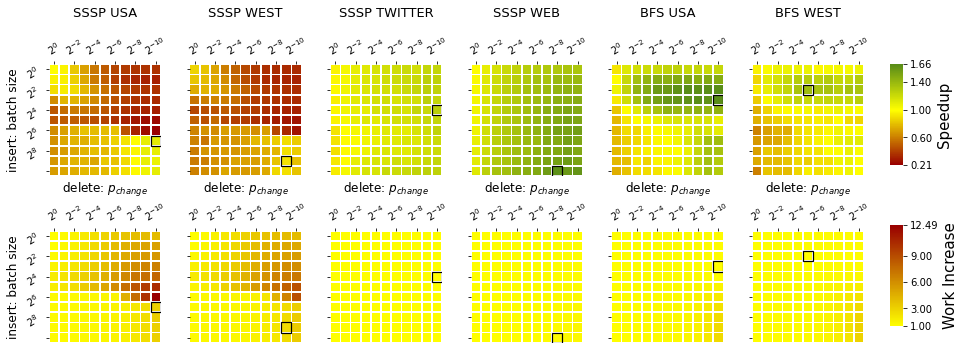}
    \includegraphics[width=0.95\textwidth]{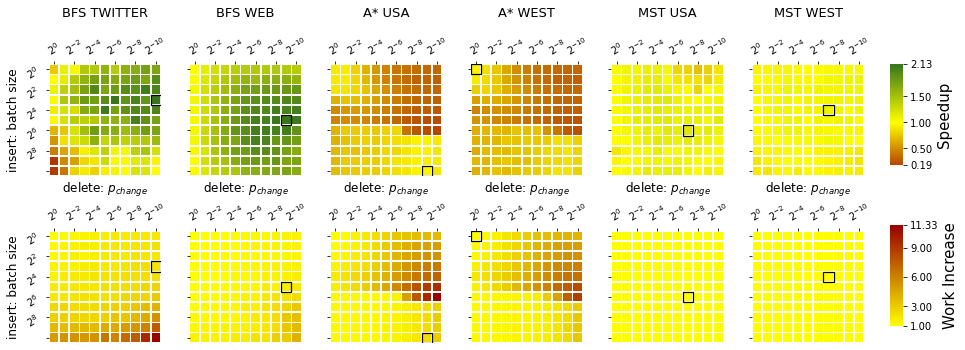}
    \vspace{-1em}
    \caption{Ablation of batch size (B) for \ins{} and queue change probability (TL) for \del{}. Experiments execute on $256$ threads on the \textbf{AMD} platform. The baseline is the classic Multi-Queue on $256$ threads with $C$ = 4. The fastest configuration for each benchmark is highlighted with a black border. Best viewed in color.}
    \label{fig:mq_insertbatch_deleteprob}
\end{figure*}

\begin{table}[h]
\small
\begin{tabular}{ |c|c|c|c|c|c|c| }
\hline
 & \normalsize{\textbf{SSSP USA}} & \normalsize{\textbf{SSSP WEST}} & \normalsize{\textbf{SSSP TWITTER}} & \normalsize{\textbf{SSSP WEB}} & \normalsize{\textbf{BFS USA}} & \normalsize{\textbf{BFS WEST}} \\
\hline
\insbatch{} & 128 & 512 & 16 & 1024 & 8 & 4 \\
\hline
\delprob{} & 1/1024 & 1/512 & 1/1024 & 1/256 & 1/1024 & 1/32 \\
\hline
\speed{} & 1.11 & 0.91 & 1.27 & 1.64 & 1.66 & 1.38 \\
\hline
\workinc{} & 3.06 & 2.60 & 1.02 & 1.13 & 1.08 & 1.07 \\
\hline
 & \normalsize{\textbf{BFS TWITTER}} & \normalsize{\textbf{BFS WEB}} & \normalsize{\textbf{A* USA}} & \normalsize{\textbf{A* WEST}} & \normalsize{\textbf{MST USA}} & \normalsize{\textbf{MST WEST}} \\
\hline
\insbatch{} & 8 & 32 & 1024 & 1 & 64 & 16 \\
\hline
\delprob{} & 1/1024 & 1/512 & 1/512 & 1/1 & 1/128 & 1/128 \\
\hline
\speed{} & 1.18 & 2.13 & 0.95 & 0.94 & 1.19 & 1.10 \\
\hline
\workinc{} & 1.02 & 1.02 & 2.38 & 1.00 & 1.00 & 1.00 \\
\hline
\end{tabular}
\vspace{0.3em}
\caption{The optimal parameters for Multi-Queue with the \emph{task batching} optimization for \ins{} and \emph{temporal locality}for \del{}, obtained on the \textbf{AMD} platform. Based on Figure~\ref{fig:mq_insertbatch_deleteprob}. For each benchmark, the best parameters are presented with the speedup and work increase.}
\label{table:mqlp_amd}
\end{table}

\newpage
\subsection{Classic Multi-Queue Optimizations on Intel: insert=Batching, delete=Temporal Locality }

\begin{figure*}[h]
    \centering
    \includegraphics[width=0.95\textwidth]{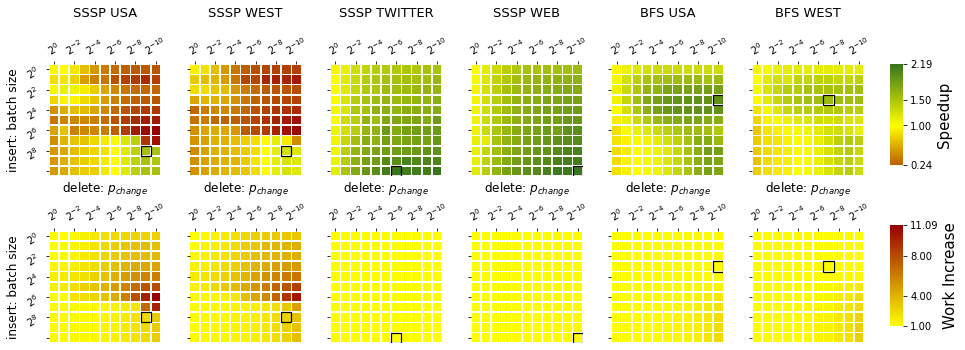}
    \includegraphics[width=0.95\textwidth]{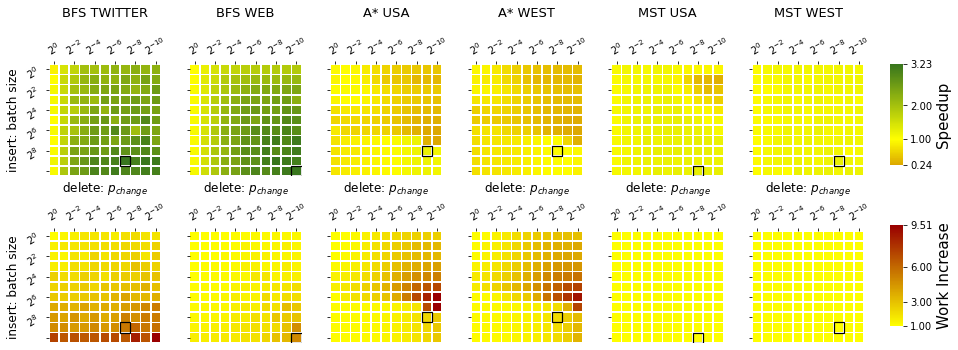}
    \vspace{-1em}
    \caption{Intel. Ablation of batch size (B) for \ins{} and queue change probability (TL) for \del{}. Experiments execute on $128$ threads. The baseline is the classic Multi-Queue on $128$ threads with $C$ = 4. The fastest configuration for each benchmark is highlighted with a black border. Best viewed in color.}
    \label{fig:mq_insertbatch_deleteprob2}
\end{figure*}

\begin{table}[h]
\small
\begin{tabular}{ |c|c|c|c|c|c|c| }
\hline
 & \normalsize{\textbf{SSSP USA}} & \normalsize{\textbf{SSSP WEST}} & \normalsize{\textbf{SSSP TWITTER}} & \normalsize{\textbf{SSSP WEB}} & \normalsize{\textbf{BFS USA}} & \normalsize{\textbf{BFS WEST}} \\
\hline
\insbatch{} & 256 & 256 & 1024 & 1024 & 8 & 8 \\
\hline
\delprob{} & 1/512 & 1/512 & 1/64 & 1/1024 & 1/1024 & 1/128 \\
\hline
\speed{} & 1.34 & 1.16 & 2.00 & 2.19 & 1.95 & 1.58 \\
\hline
\workinc{} & 2.47 & 2.80 & 1.02 & 1.21 & 1.08 & 1.09 \\
\hline
 & \normalsize{\textbf{BFS TWITTER}} & \normalsize{\textbf{BFS WEB}} & \normalsize{\textbf{A* USA}} & \normalsize{\textbf{A* WEST}} & \normalsize{\textbf{MST USA}} & \normalsize{\textbf{MST WEST}} \\
\hline
\insbatch{} & 512 & 1024 & 256 & 256 & 1024 & 512 \\
\hline
\delprob{} & 1/128 & 1/1024 & 1/512 & 1/256 & 1/256 & 1/256 \\
\hline
\speed{} & 1.60 & 3.23 & 1.22 & 1.03 & 1.27 & 1.24 \\
\hline
\workinc{} & 1.05 & 1.05 & 2.28 & 2.13 & 1.00 & 1.00 \\
\hline
\end{tabular}
\vspace{0.3em}
\caption{The optimal parameters for Multi-Queue with the \emph{task batching} optimization for \ins{} and \\emph{temporal locality}for \del{}, obtained on the \textbf{Intel} platform. Based on Figure~\ref{fig:mq_insertbatch_deleteprob2}. For each benchmark, the best parameters are presented with the speedup and work increase.}
\label{table:mqlp_intel}
\end{table}

\newpage
\subsection{Classic Multi-Queue Optimizations on AMD: insert=Batching, delete=Batching }

\begin{figure*}[h]
    \centering
    \includegraphics[width=0.95\textwidth]{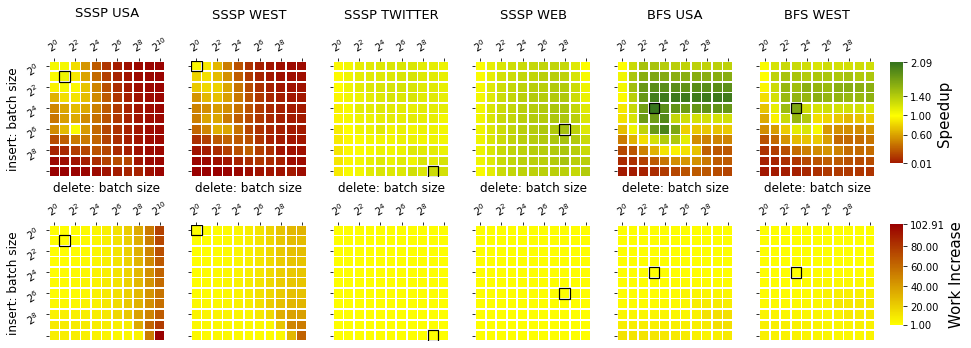}
    \includegraphics[width=0.95\textwidth]{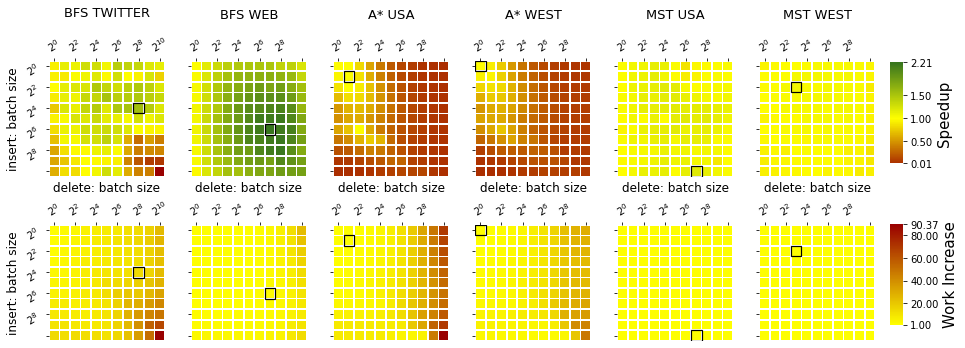}
    \vspace{-1em}
    \caption{Ablation of batch sizes (B) for \ins{} and \del{}. Experiments execute on $256$ threads on the \textbf{AMD} platform. The baseline is the classic Multi-Queue on $256$ threads with $C$ = 4. The fastest configuration for each benchmark is highlighted with a black border. Best viewed in color.}
    \label{fig:mq_insertbatch_deletebatch}
\end{figure*}

\begin{table}[h]
\small
\begin{tabular}{ |c|c|c|c|c|c|c| }
\hline
 & \normalsize{\textbf{SSSP USA}} & \normalsize{\textbf{SSSP WEST}} & \normalsize{\textbf{SSSP TWITTER}} & \normalsize{\textbf{SSSP WEB}} & \normalsize{\textbf{BFS USA}} & \normalsize{\textbf{BFS WEST}} \\
\hline
\insbatch{} & 2 & 1 & 1024 & 64 & 16 & 16 \\
\hline
\delbatch{} & 2 & 1 & 512 & 256 & 8 & 8 \\
\hline
\speed{} & 1.08 & 0.97 & 1.25 & 1.44 & 2.09 & 1.71 \\
\hline
\workinc{} & 1.33 & 1.00 & 1.12 & 1.15 & 1.11 & 1.29 \\
\hline
 & \normalsize{\textbf{BFS TWITTER}} & \normalsize{\textbf{BFS WEB}} & \normalsize{\textbf{A* USA}} & \normalsize{\textbf{A* WEST}} & \normalsize{\textbf{MST USA}} & \normalsize{\textbf{MST WEST}} \\
\hline
\insbatch{} & 16 & 64 & 2 & 1 & 1024 & 4 \\
\hline
\delbatch{} & 256 & 128 & 2 & 1 & 128 & 8 \\
\hline
\speed{} & 1.15 & 2.21 & 1.02 & 1.00 & 1.19 & 1.07 \\
\hline
\workinc{} & 1.06 & 1.01 & 1.27 & 1.00 & 1.00 & 1.00 \\
\hline
\end{tabular}
\vspace{0.3em}
\caption{The optimal parameters for Multi-Queue with the \emph{task batching} optimization for both \ins{} and \del{}, obtained on the \textbf{AMD} platform. Based on Figure~\ref{fig:mq_insertbatch_deletebatch}. For each benchmark, the best parameters are presented with the speedup and work increase.}
\label{table:mqll_amd}
\end{table}

\newpage
\subsection{Classic Multi-Queue Optimizations on Intel: insert=Batching, delete=Batching }

\begin{figure*}[h]
    \centering
    \includegraphics[width=0.95\textwidth]{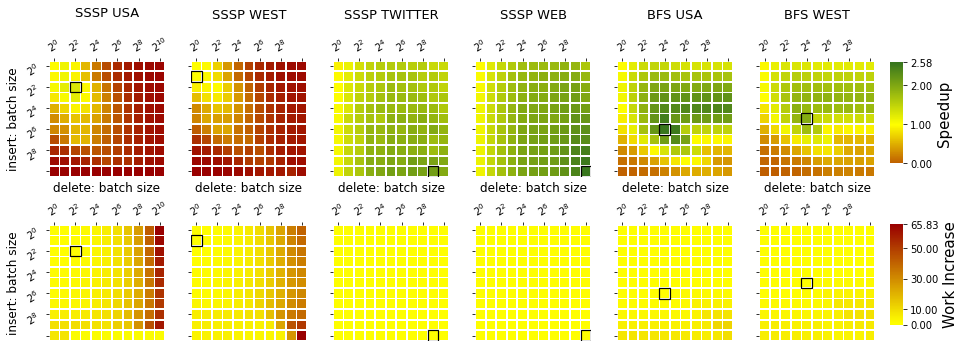}
    \includegraphics[width=0.95\textwidth]{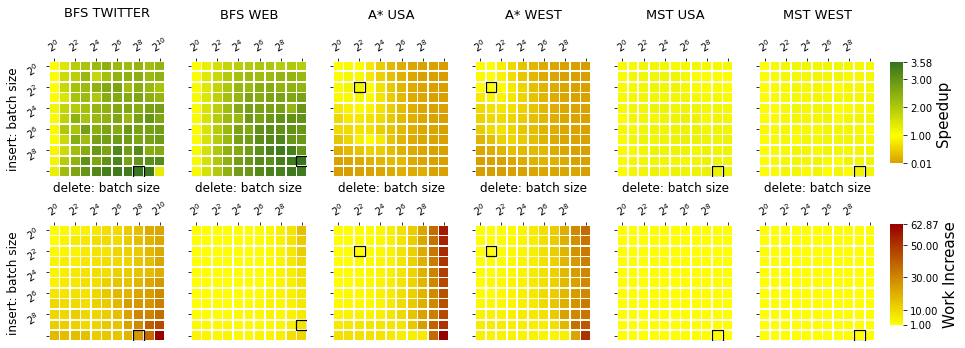}
    \vspace{-1em}
    \caption{Ablation of batch sizes (B) for \ins{} and \del{}. Experiments execute on $128$ threads on \textbf{Intel}. The baseline is the classic Multi-Queue on $128$ threads with $C$ = 4. The fastest configuration for each benchmark is highlighted with a black border. Best viewed in color.}
    \label{fig:mq_insertbatch_deletebatch2}
\end{figure*}

\begin{table}[h]
\small
\begin{tabular}{ |c|c|c|c|c|c|c| }
\hline
 & \normalsize{\textbf{SSSP USA}} & \normalsize{\textbf{SSSP WEST}} & \normalsize{\textbf{SSSP TWITTER}} & \normalsize{\textbf{SSSP WEB}} & \normalsize{\textbf{BFS USA}} & \normalsize{\textbf{BFS WEST}} \\
\hline
\insbatch{} & 4 & 2 & 1024 & 1024 & 64 & 32 \\
\hline
\delbatch{} & 4 & 1 & 512 & 1024 & 16 & 16 \\
\hline
\speed{} & 1.19 & 1.00 & 1.71 & 2.07 & 2.58 & 1.96 \\
\hline
\workinc{} & 1.55 & 1.08 & 1.09 & 1.23 & 1.22 & 1.30 \\
\hline
 & \normalsize{\textbf{BFS TWITTER}} & \normalsize{\textbf{BFS WEB}} & \normalsize{\textbf{A* USA}} & \normalsize{\textbf{A* WEST}} & \normalsize{\textbf{MST USA}} & \normalsize{\textbf{MST WEST}} \\
\hline
\insbatch{} & 1024 & 512 & 4 & 4 & 1024 & 1024 \\
\hline
\delbatch{} & 256 & 1024 & 4 & 2 & 512 & 512 \\
\hline
\speed{} & 1.63 & 3.58 & 1.19 & 1.00 & 1.26 & 1.23 \\
\hline
\workinc{} & 1.08 & 1.02 & 1.40 & 1.33 & 1.00 & 1.00 \\
\hline
\end{tabular}
\vspace{0.3em}
\caption{The optimal parameters for Multi-Queue with the \emph{task batching} optimization for both \ins{} and \del{}, obtained on \textbf{Intel}. Based on Figure~\ref{fig:mq_insertbatch_deletebatch2}. For each benchmark, the best parameters are presented with the speedup and work increase.}
\label{table:mqll_intel}
\end{table}

\newpage
\subsection{Comparison of Task Batching and Temporal Locality Optimizations}

\begin{figure*}[h]
    \centering
    \includegraphics[width=0.95\textwidth]{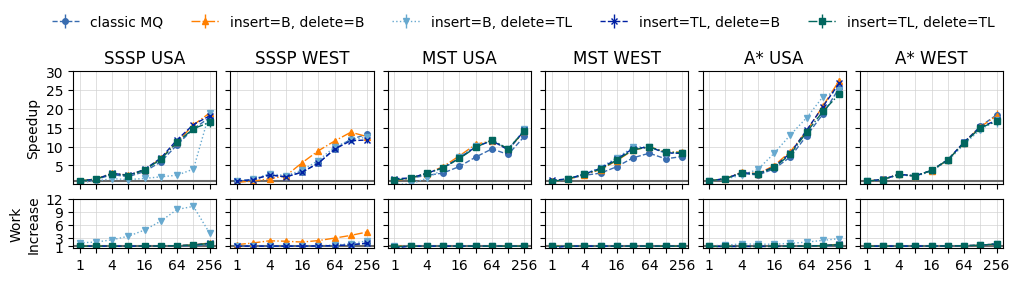}
    \includegraphics[width=0.95\textwidth]{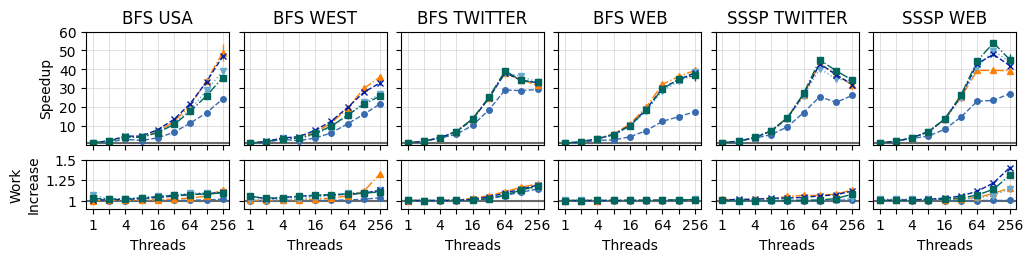}
    \caption{A comparison of task batching (B) and temporal locality (TL) optimizations on the classic Multi-Queue on the \textbf{AMD} platform.}
    \label{fig:mq_combinations} 
\end{figure*}

\begin{figure*}[h]
    \centering
    \includegraphics[width=0.95\textwidth]{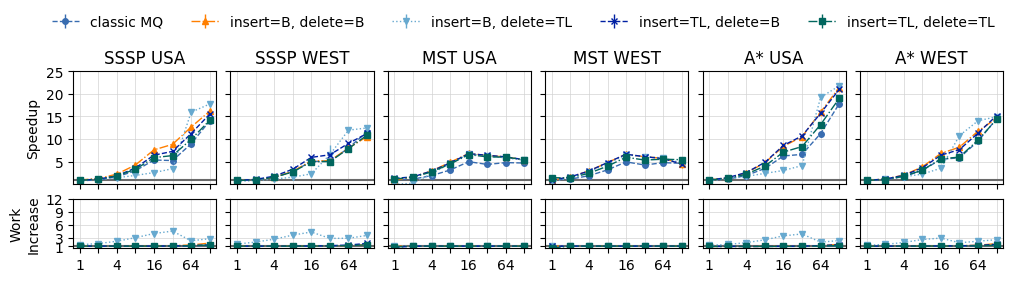}
    \includegraphics[width=0.95\textwidth]{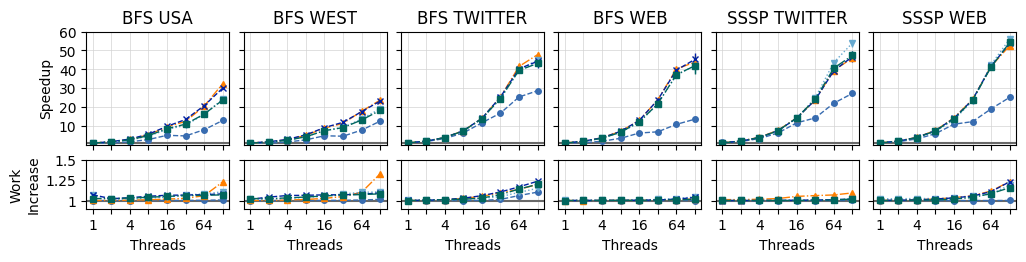}
    \caption{Intel. A comparison of task batching (B) and temporal locality (TL) optimizations on the classic Multi-Queue.}
    \label{fig:mq_combinations2} 
\end{figure*}


\newpage
\section{Stealing Multi-Queue Implementation Details}\label{appendix:smq:tuning}

This section presents the experimental data for the proposed Stealing Multi-Queue (SMQ) algorithm. Figures~\ref{fig:smq_hm_amd}--\ref{fig:smq_hm_intel} and Tables~\ref{table:smq_hm_amd}--\ref{table:smq_hm_intel} show results for the version with sequential d-ary heaps and stealing buffers, while Figures~\ref{fig:sl_smq_hm_amd}--\ref{fig:sl_smq_hm_intel} and Tables~\ref{table:sl_smq_hm_amd}--\ref{table:sl_smq_hm_intel} show results for the Skip-List based version.

\subsection{SMQ via D-Ary Heaps on AMD}
\vspace{-1em}
\begin{figure*}[h]
    \centering
    \includegraphics[width=0.93\textwidth]{plots/heatmaps/AMD/smq_amd_heatmaps1.png}
    
    \includegraphics[width=0.93\textwidth]{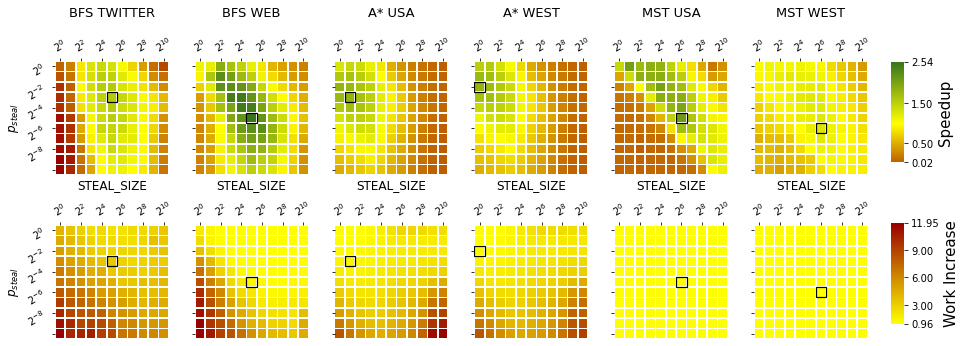}
    \vspace{-1em}
    \caption{Ablation of stealing probability $p_{\idlow{steal}}$ and steal buffer size on the \textbf{AMD} architecture, for \smq{} implemented using $d$-ary heaps. Experiments execute on $256$ threads. The baseline is the classic Multi-Queue on $256$ threads with $C$ = 4. The fastest configuration for each benchmark is highlighted with a black border and listed in Table~\ref{table:smq_hm_amd}. Best viewed in color.}
    \label{fig:smq_hm_amd}
\end{figure*}

\begin{table}[h]
\small
\begin{tabular}{ |c|c|c|c|c|c|c| }
\hline
 & \large{\textbf{SSSP USA}} & \large{\textbf{SSSP W}} & \large{\textbf{SSSP TWITTER}} & \large{\textbf{SSSP WEB}} & \large{\textbf{BFS USA}} & \large{\textbf{BFS W}} \\
\hline
\stealprob{} & 1/4 & 1/4 & 1/16 & 1/16 & 1/4 & 1/4 \\
\hline
\chunksize{} & 1 & 1 & 64 & 8 & 1 & 1 \\
\hline
\speed{} & 2.28 & 2.00 & 2.01 & 2.92 & 2.87 & 2.33 \\
\hline
\workinc{} & 1.18 & 1.20 & 1.03 & 1.11 & 1.05 & 1.12 \\
\hline
& \large{\textbf{BFS TWITTER}} & \large{\textbf{BFS WEB}} & \large{\textbf{A* USA}} & \large{\textbf{A* W}} & \large{\textbf{MST USA}} & \large{\textbf{MST W}} \\
\hline
\stealprob{} & 1/8 & 1/32 & 1/8 & 1/4 & 1/32 & 1/64 \\
\hline
\chunksize{} &  32 & 32 & 2 & 1 & 64 & 64 \\
\hline
\speed{} & 1.30 & 2.54 & 1.83 & 1.82 & 2.13 & 1.28 \\
\hline
\workinc{} & 1.33 & 1.09 & 1.35 & 1.23 & 1.00 & 1.00 \\
\hline
\end{tabular}
\vspace{0.3em}
\caption{The optimal parameters for Stealing Multi-Queue via $d$-ary heaps on $256$ threads, obtained on the \textbf{AMD} platform. Based on Figure~\ref{fig:smq_hm_amd}. For each benchmark, the best parameters are presented with the speedup and work increase.}
\label{table:smq_hm_amd}
\end{table}

\newpage
\subsection{SMQ via D-Ary Heaps on Intel}

\begin{figure*}[h]
    \centering
    \includegraphics[width=0.95\textwidth]{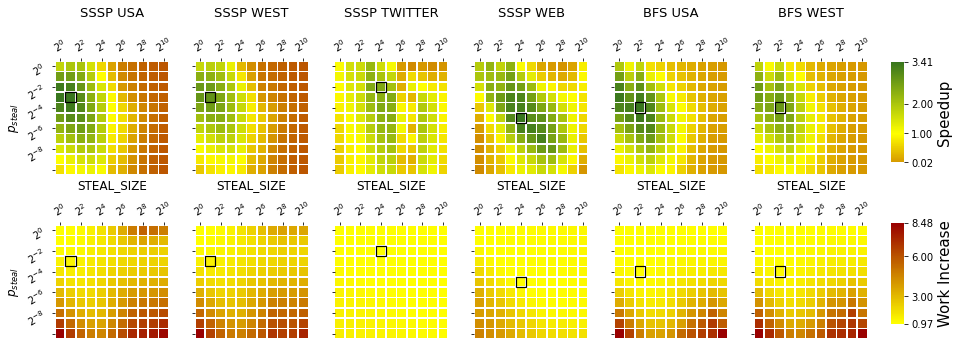}
    \includegraphics[width=0.95\textwidth]{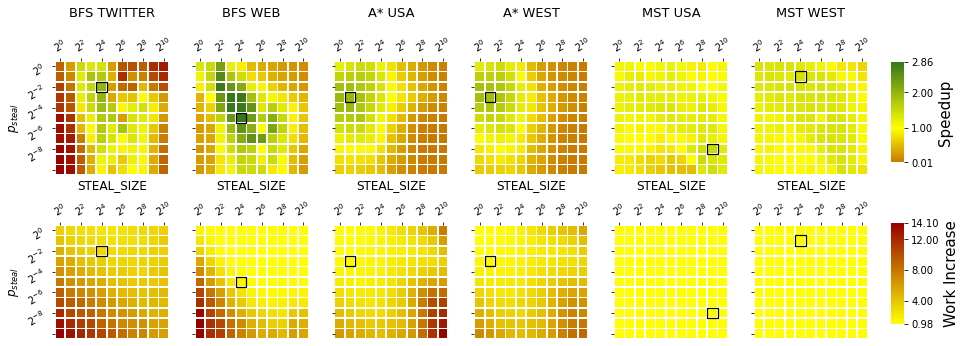}
    \vspace{-1em}
    \caption{Ablation of stealing probability $p_{\idlow{steal}}$ and steal buffer size on \textbf{Intel}, for \smq{} implemented using $d$-ary heaps. Experiments execute on $128$ threads. The baseline is the classic Multi-Queue on $128$ threads with $C$ = 4. The fastest configuration for each benchmark is highlighted with a black border and listed in Table~\ref{table:smq_hm_intel}. Best viewed in color.}
    \label{fig:smq_hm_intel}
\end{figure*}

\begin{table}[h]
\small
\begin{tabular}{ |c|c|c|c|c|c|c| }
\hline
 & \large{\textbf{SSSP USA}} & \large{\textbf{SSSP W}} & \large{\textbf{SSSP TWITTER}} & \large{\textbf{SSSP WEB}} & \large{\textbf{BFS USA}} & \large{\textbf{BFS W}} \\
\hline
\stealprob{} & 1/8 & 1/8 & 1/4 & 1/32 & 1/16 & 1/16 \\
\hline
\chunksize{} & 2 & 2 & 16 & 16 & 4 & 4 \\
\hline
\speed{} & 2.47 & 2.15 & 1.92 & 2.46 & 3.41 & 2.82 \\
\hline
\workinc{} & 1.23 & 1.38 & 1.01 & 1.06 & 1.06 & 1.17 \\
\hline
 & \large{\textbf{BFS TWITTER}} & \large{\textbf{BFS WEB}} & \large{\textbf{A* USA}} & \large{\textbf{A* W}} & \large{\textbf{MST USA}} & \large{\textbf{MST W}} \\
\hline
\stealprob{} & 1/4 & 1/32 & 1/8 & 1/8 & 1/256 & 1/2 \\
\hline
\chunksize{} & 16 & 16 & 2 & 2 & 512 & 16 \\
\hline
\speed{} & 1.44 & 2.86 & 1.95 & 1.85 & 1.48 & 1.39 \\
\hline
\workinc{} & 1.20 & 1.10 & 1.24 & 1.33 & 1.00 & 1.00 \\
\hline
\end{tabular}
\vspace{0.3em}
\caption{The optimal parameters of SMQ via $d$-ary heaps on \textbf{Intel}, based on Figure~\ref{fig:smq_hm_intel}. For each benchmark, the best $p_{\idlow{steal}}$ and steal buffer size combination is presented, with its speedup and work increase.  }
\label{table:smq_hm_intel}
\end{table}

\newpage
\subsection{SMQ via Skip Lists on AMD}

\begin{figure*}[h]
    \centering
    \includegraphics[width=0.95\textwidth]{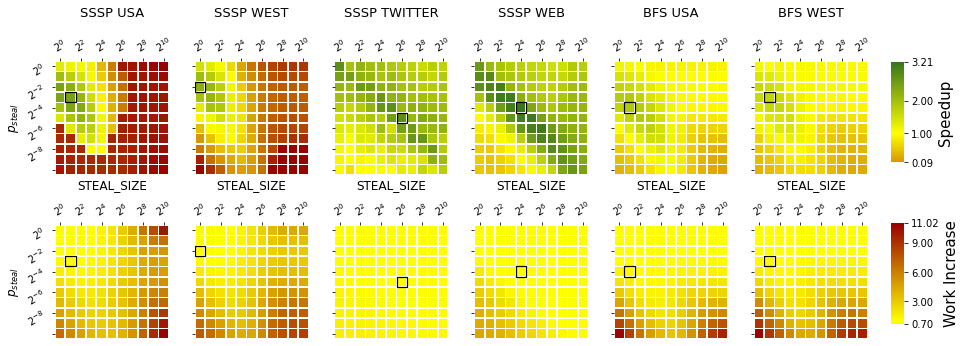}
    \includegraphics[width=0.95\textwidth]{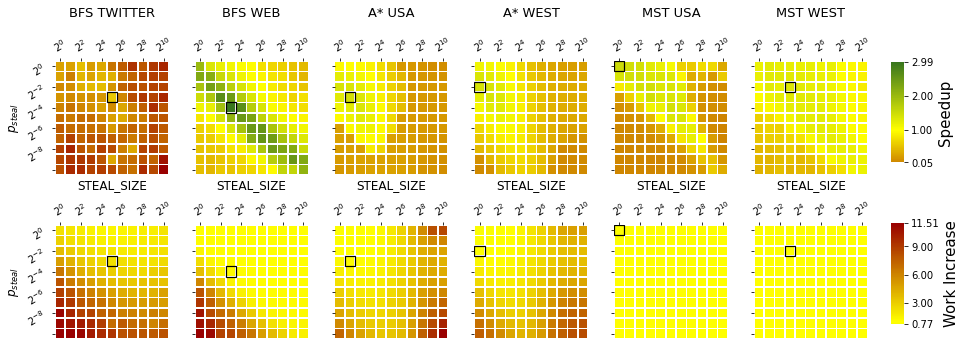}
    \vspace{-1em}
    \caption{Ablation of stealing probability $p_{\idlow{steal}}$ and steal buffer size on the \textbf{AMD} architecture, for \smq{} implemented using skip-lists. Experiments execute on $256$ threads. The baseline is the classic Multi-Queue on $256$ threads with $C$ = 4. The fastest configuration for each benchmark is highlighted with a black border and listed in Table~\ref{table:sl_smq_hm_amd}. Best viewed in color.}
    \label{fig:sl_smq_hm_amd}
\end{figure*}

\begin{table}[h]
\small
\begin{tabular}{ |c|c|c|c|c|c|c| }
\hline
 & \large{\textbf{SSSP USA}} & \large{\textbf{SSSP W}} & \large{\textbf{SSSP TWITTER}} & \large{\textbf{SSSP WEB}} & \large{\textbf{BFS USA}} & \large{\textbf{BFS W}} \\
\hline
\stealprob{} & 1/8 & 1/4 & 1/32 & 1/16 & 1/16 & 1/8 \\
\hline
\chunksize{} & 2 & 1 & 64 & 16 & 2 & 2 \\
\hline
\speed{} & 1.59 & 1.51 & 1.74 & 3.21 & 1.87 & 1.68 \\
\hline
\workinc{} & 1.13 & 1.00 & 1.03 & 1.01 & 1.21 & 1.13 \\
\hline
& \large{\textbf{BFS TWITTER}} & \large{\textbf{BFS WEB}} & \large{\textbf{A* USA}} & \large{\textbf{A* W}} & \large{\textbf{MST USA}} & \large{\textbf{MST W}} \\
\hline
\stealprob{} & 1/8 & 1/16 & 1/8 & 1/4 & 1/1 & 1/4 \\
\hline
\chunksize{} & 32 & 8 & 2 & 1 & 1 & 8 \\
\hline
\speed{} & 0.88 & 2.99 & 1.44 & 1.42 & 1.52 & 1.34 \\
\hline
\workinc{} & 1.14 & 1.04 & 1.14 & 1.04 & 1.01 & 1.00 \\
\hline
\end{tabular}
\vspace{0.3em}
\caption{The optimal parameters of SMQ via skip lists on the \textbf{AMD} architecture, based on Figure~\ref{fig:smq_hm_amd}. For each benchmark, the best $p_{\idlow{steal}}$ and steal buffer size combination is presented, with its speedup and work increase.}
\label{table:sl_smq_hm_amd}
\end{table}

\newpage
\subsection{SMQ via Skip Lists on Intel}

\begin{figure*}[h]
    \centering
    \includegraphics[width=0.95\textwidth]{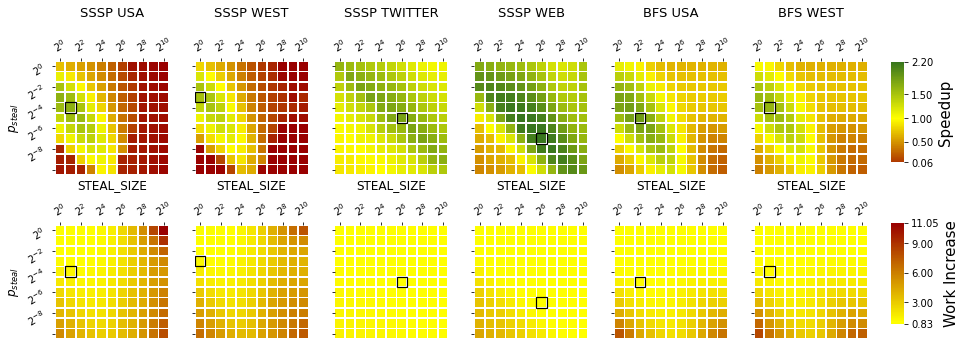}
    \includegraphics[width=0.95\textwidth]{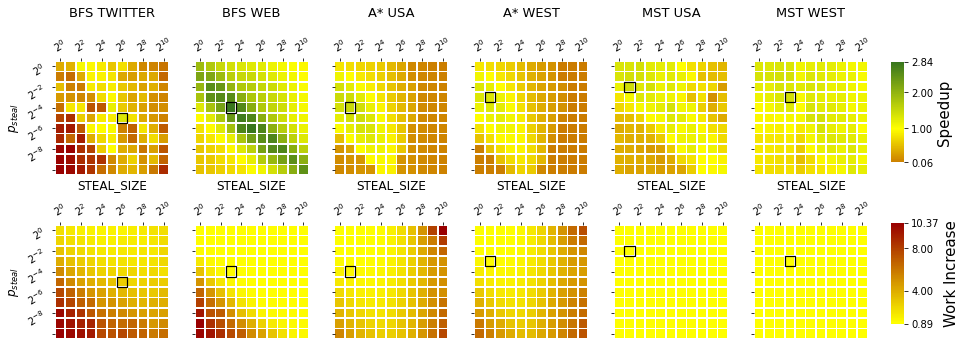}
    \vspace{-1em}
    \caption{Ablation of stealing probability $p_{\idlow{steal}}$ and steal buffer size on \textbf{Intel}, for \smq{} implemented using skip-lists. Experiments execute on $128$ threads. The baseline is the classic Multi-Queue on $128$ threads with $C$ = 4. The fastest configuration for each benchmark is highlighted with a black border and listed in Table~\ref{table:sl_smq_hm_intel}. Best viewed in color.}
    \label{fig:sl_smq_hm_intel}
\end{figure*}

\begin{table}[h]
\small
\begin{tabular}{ |c|c|c|c|c|c|c| }
\hline
 & \large{\textbf{SSSP USA}} & \large{\textbf{SSSP W}} & \large{\textbf{SSSP TWITTER}} & \large{\textbf{SSSP WEB}} & \large{\textbf{BFS USA}} & \large{\textbf{BFS W}} \\
\hline
\stealprob{} & 1/16 & 1/8 & 1/32 & 1/128 & 1/32 & 1/16 \\
\hline
\chunksize{} & 2 & 1 & 64 & 64 & 4 & 2 \\
\hline
\speed{} & 1.53 & 1.42 & 1.64 & 2.20 & 1.85 & 1.64 \\
\hline
\workinc{} & 1.28 & 1.30 & 1.01 & 1.02 & 1.12 & 1.18 \\
\hline
 & \large{\textbf{BFS TWITTER}} & \large{\textbf{BFS WEB}} & \large{\textbf{A* USA}} & \large{\textbf{A* W}} & \large{\textbf{MST USA}} & \large{\textbf{MST W}} \\
\hline
\stealprob{} & 1/32 & 1/16 & 1/16 & 1/8 & 1/4 & 1/8 \\
\hline
\chunksize{} & 64 & 8 & 2 & 2 & 2 & 8 \\
\hline
\speed{} & 1.10 & 2.84 & 1.49 & 1.33 & 1.43 & 1.42 \\
\hline
\workinc{} & 1.21 & 1.02 & 1.25 & 1.14 & 1.01 & 1.00 \\
\hline
\end{tabular}
\vspace{0.3em}
\caption{The optimal parameters of SMQ via skip lists on the \textbf{Intel} architecture, based on Figure~\ref{fig:sl_smq_hm_intel}. For each benchmark, the best $p_{\idlow{steal}}$ and steal buffer size combination is presented, with its speedup and work increase.}
\label{table:sl_smq_hm_intel}
\end{table}


\clearpage
\section{NUMA Awareness}\label{appendix:numa:tuning}
In this section we examine how the proposed in Section~\ref{sec:smq_impl} optimization for NUMA architectures improves both the Multi-Queue variants (Tables~\ref{table:mqpp_numa_amd}--\ref{table:mqll_numa_intel}) and the proposed Stealing Multi-Queue (SMQ) via both d-ary heaps (Tables~\ref{table:smq_numa_amd}--\ref{table:smq_numa_intel}) and Skip-Lists (Tables~\ref{table:slsmq_numa_amd}--\ref{table:slsmq_numa_intel}).

\subsection{MQ Optimized NUMA: insert=Temporal Locality, delete=Temporal Locality }
\begin{table}[h]
\small
\begin{center}
\begin{tabular}{ |c|c|c|c|c|c|c|c|c|c|c|c| }
\hline
 & \large{\textbf{1}} & \large{\textbf{2}} & \large{\textbf{4}} & \large{\textbf{8}} & \large{\textbf{16}} & \large{\textbf{32}} & \large{\textbf{64}} & \large{\textbf{128}} & \large{\textbf{256}} & \large{\textbf{512}} & \large{\textbf{1024}} \\
\hline
\large{\textbf{BFS USA}} & 1.44 & 1.56 & 1.65 & 1.70 & 1.73 & \color{Numa}{\textbf{1.75}} & \color{Numa}{\textbf{1.75}} & 1.73 & 1.64 & 1.54 & 1.36 \\
\hline
\large{\textbf{BFS WEST}} & 1.21 & 1.30 & 1.37 & 1.41 & 1.40 & \color{Numa}{\textbf{1.48}} & \color{Numa}{\textbf{1.48}} & 1.47 & 1.41 & 1.32 & 1.20 \\
\hline
\large{\textbf{BFS TWITTER}} & \color{Numa}{\textbf{1.14}} & 1.10 & 1.12 & 1.11 & 1.04 & 0.97 & 0.91 & 0.88 & 0.82 & 0.80 & 0.78 \\
\hline
\large{\textbf{BFS WEB}} & 2.22 & 2.32 & \color{Numa}{\textbf{2.38}} & 2.03 & 1.97 & 1.67 & 1.52 & 1.55 & 1.43 & 1.44 & 1.45 \\
\hline
\large{\textbf{SSSP USA}} & 1.00 & 1.03 & 1.10 & 1.12 & 1.16 & 1.19 & 1.19 & \color{Numa}{\textbf{1.20}} & \color{Numa}{\textbf{1.20}} & \color{Numa}{\textbf{1.20}} & 1.15 \\
\hline
\large{\textbf{SSSP WEST}} & 0.95 & 0.98 & 0.99 & 1.05 & 1.08 & 1.09 & 1.12 & \color{Numa}{\textbf{1.14}} & 1.09 & 1.06 & 1.08 \\
\hline
\large{\textbf{SSSP TWITTER}} & 1.30 & 1.27 & \color{Numa}{\textbf{1.35}} & 1.28 & 1.23 & 1.27 & 1.27 & 1.27 & 1.33 & 1.23 & 1.33 \\
\hline
\large{\textbf{SSSP WEB}} & 1.68 & 1.77 & 1.78 & \color{Numa}{\textbf{1.86}} & 1.78 & 1.67 & 1.60 & 1.55 & 1.44 & 1.55 & 1.58 \\
\hline
\large{\textbf{MST USA}} & 0.43 & 1.06 & \color{Numa}{\textbf{1.15}} & 1.00 & 0.98 & 0.79 & 0.59 & 0.69 & 0.40 & 0.47 & 0.73 \\
\hline
\large{\textbf{MST WEST}} & \color{Numa}{\textbf{1.16}} & 1.12 & 1.06 & 1.12 & 1.12 & 1.13 & 1.13 & 1.14 & 1.14 & 1.15 & 1.09 \\
\hline
\large{\textbf{A* USA}} & 0.95 & 1.00 & 1.05 & 1.09 & 1.12 & 1.12 & 1.13 & \color{Numa}{\textbf{1.15}} & 1.13 & 1.11 & 1.11 \\
\hline
\large{\textbf{A* WEST}} & 0.94 & 1.02 & 1.06 & 1.08 & 1.11 & 1.13 & \color{Numa}{\textbf{1.17}} & 1.16 & \color{Numa}{\textbf{1.17}} & 1.14 & 1.11 \\
\hline
\end{tabular}
\end{center}
\vspace{0.3em}
\caption{Speedups for Multi-Queue variants with the \emph{temporal locality} optimization for both \ins{} and \del{}, and various weights $K$ for non-local NUMA node accesses, obtained on the \textbf{AMD} platform on $256$ threads; the best speedups are highlighted with {\color{Numa}{\textbf{\numacol{}}}}. The baseline is the classic Multi-Queue on $256$ threads with $C$ = 4. With $K = 1$ the algorithm is the same as without the NUMA-specific optimization. }
\label{table:mqpp_numa_amd}
\end{table}
\begin{table}[h]
\small
\begin{center}
\begin{tabular}{ |c|c|c|c|c|c|c|c|c|c|c|c| }
\hline
 & \large{\textbf{1}} & \large{\textbf{2}} & \large{\textbf{4}} & \large{\textbf{8}} & \large{\textbf{16}} & \large{\textbf{32}} & \large{\textbf{64}} & \large{\textbf{128}} & \large{\textbf{256}} & \large{\textbf{512}} & \large{\textbf{1024}} \\
\hline
\large{\textbf{BFS USA}} & 1.58 & 1.79 & 2.09 & 2.37 & 2.64 & 2.84 & \color{Numa}{\textbf{2.97}} & 2.86 & 2.72 & 2.49 & 1.99 \\
\hline
\large{\textbf{BFS WEST}} & 1.28 & 1.45 & 1.68 & 1.94 & 2.16 & 2.37 & \color{Numa}{\textbf{2.43}} & 2.40 & 2.17 & 2.02 & 1.68 \\
\hline
\large{\textbf{BFS TWITTER}} & \color{Numa}{\textbf{1.58}} & \color{Numa}{\textbf{1.58}} & 1.39 & 1.36 & 1.21 & 1.05 & 0.99 & 0.78 & 0.73 & 0.64 & 0.68 \\
\hline
\large{\textbf{BFS WEB}} & 3.11 & 3.32 & \color{Numa}{\textbf{3.69}} & 3.23 & 2.62 & 2.14 & 1.81 & 1.56 & 1.53 & 1.31 & 1.26 \\
\hline
\large{\textbf{SSSP USA}} & 0.92 & 0.99 & 1.12 & 1.28 & 1.44 & 1.58 & 1.68 & 1.74 & \color{Numa}{\textbf{1.78}} & 1.77 & 1.72 \\
\hline
\large{\textbf{SSSP WEST}} & 0.88 & 0.92 & 1.02 & 1.17 & 1.32 & 1.46 & 1.56 & 1.61 & \color{Numa}{\textbf{1.65}} & 1.62 & 1.54 \\
\hline
\large{\textbf{SSSP TWITTER}} & 1.77 & 1.90 & \color{Numa}{\textbf{1.96}} & \color{Numa}{\textbf{1.96}} & 1.63 & 1.74 & 1.74 & 1.72 & 1.74 & 1.67 & 1.66 \\
\hline
\large{\textbf{SSSP WEB}} & 2.11 & 2.30 & \color{Numa}{\textbf{2.34}} & 2.25 & 1.90 & 1.62 & 1.38 & 1.18 & 1.09 & 1.00 & 0.95 \\
\hline
\large{\textbf{MST USA}} & 1.02 & 1.11 & 1.22 & \color{Numa}{\textbf{1.28}} & \color{Numa}{\textbf{1.28}} & \color{Numa}{\textbf{1.28}} & 1.22 & 1.20 & 1.22 & 1.22 & 1.18 \\
\hline
\large{\textbf{MST WEST}} & 1.16 & 1.17 & 1.20 & 1.10 & 1.12 & 1.24 & \color{Numa}{\textbf{1.31}} & 1.28 & 1.28 & 1.21 & 1.13 \\
\hline
\large{\textbf{A* USA}} & 0.99 & 1.06 & 1.19 & 1.34 & 1.49 & 1.60 & 1.69 & 1.75 & \color{Numa}{\textbf{1.76}} & 1.75 & 1.69 \\
\hline
\large{\textbf{A* WEST}} & 0.93 & 0.97 & 1.07 & 1.22 & 1.36 & 1.48 & 1.58 & 1.63 & \color{Numa}{\textbf{1.65}} & 1.64 & 1.58 \\
\hline
\end{tabular}
\end{center}
\vspace{0.3em}
\caption{Speedups for Multi-Queue variants with the \emph{temporal locality} optimization for both \ins{} and \del{},, and various weights $K$ for non-local NUMA node accesses, obtained on the \textbf{Intel} platform on $128$ threads; the best speedups are highlighted with {\color{Numa}{\textbf{\numacol{}}}}. The baseline is the classic Multi-Queue on $128$ threads with $C$ = 4. With $K = 1$ the algorithm is the same as without the NUMA-specific optimization.}
\label{table:mqpp_numa_intel}
\end{table} 

\clearpage
\subsection{MQ Optimized NUMA: insert=Temporal Locality, delete=Task Batching }
\begin{table}[h]
\small
\begin{center}
\begin{tabular}{ |c|c|c|c|c|c|c|c|c|c|c|c| }
\hline
 & \large{\textbf{1}} & \large{\textbf{2}} & \large{\textbf{4}} & \large{\textbf{8}} & \large{\textbf{16}} & \large{\textbf{32}} & \large{\textbf{64}} & \large{\textbf{128}} & \large{\textbf{256}} & \large{\textbf{512}} & \large{\textbf{1024}} \\
\hline
\large{\textbf{BFS USA}} & 1.99 & 2.06 & 2.04 & 2.14 & 2.10 & 2.16 & 2.15 & 2.16 & \color{Numa}{\textbf{2.18}} & 2.10 & 2.03 \\
\hline
\large{\textbf{BFS WEST}} & 1.56 & 1.64 & 1.67 & 1.71 & 1.70 & 1.72 & 1.73 & \color{Numa}{\textbf{1.76}} & 1.72 & 1.68 & 1.60 \\
\hline
\large{\textbf{BFS TWITTER}} & 1.15 & 1.15 & \color{Numa}{\textbf{1.17}} & 1.16 & 1.13 & 1.11 & 1.10 & 1.07 & 1.05 & 1.00 & 0.95 \\
\hline
\large{\textbf{BFS WEB}} & 2.21 & 2.23 & 2.27 & 2.33 & \color{Numa}{\textbf{2.38}} & 2.37 & 2.11 & 2.01 & 1.77 & 1.78 & 1.62 \\
\hline
\large{\textbf{SSSP USA}} & 1.05 & 1.13 & 1.18 & 1.20 & 1.26 & 1.28 & 1.29 & 1.29 & \color{Numa}{\textbf{1.31}} & 1.30 & 1.27 \\
\hline
\large{\textbf{SSSP WEST}} & 0.88 & 0.95 & 1.00 & 1.04 & 1.05 & 1.07 & 1.07 & 1.08 & \color{Numa}{\textbf{1.10}} & 1.09 & 1.07 \\
\hline
\large{\textbf{SSSP TWITTER}} & 1.16 & 1.26 & 1.22 & 1.32 & 1.29 & 1.29 & 1.25 & \color{Numa}{\textbf{1.33}} & 1.28 & \color{Numa}{\textbf{1.33}} & 1.26 \\
\hline
\large{\textbf{SSSP WEB}} & 1.58 & 1.56 & 1.68 & 1.71 & 1.72 & \color{Numa}{\textbf{1.81}} & 1.70 & 1.55 & 1.62 & 1.62 & 1.54 \\
\hline
\large{\textbf{MST USA}} & 1.13 & 1.15 & \color{Numa}{\textbf{1.20}} & 1.12 & 1.18 & 1.17 & 1.12 & 0.89 & 0.69 & 0.61 & 0.43 \\
\hline
\large{\textbf{MST WEST}} & \color{Numa}{\textbf{1.17}} & 1.14 & 1.13 & 1.13 & 1.13 & 1.13 & 1.12 & 1.14 & 1.13 & 1.14 & 1.13 \\
\hline
\large{\textbf{A* USA}} & 0.97 & 1.07 & 1.10 & 1.15 & 1.18 & 1.17 & 1.21 & \color{Numa}{\textbf{1.22}} & \color{Numa}{\textbf{1.22}} & 1.19 & 1.20 \\
\hline
\large{\textbf{A* WEST}} & 0.98 & 0.97 & 1.06 & 1.05 & 1.13 & 1.16 & 1.16 & \color{Numa}{\textbf{1.19}} & 1.17 & 1.11 & 1.12 \\
\hline
\end{tabular}
\end{center}
\vspace{0.3em}
\caption{Speedups for Multi-Queue variants with the \emph{temporal locality} optimization for \ins{}, \emph{task batching} for \del{}, and various weights $K$ for non-local NUMA node accesses, obtained on the \textbf{AMD} platform on $256$ threads; the best speedups are highlighted with {\color{Numa}{\textbf{\numacol{}}}}. The baseline is the classic Multi-Queue on $256$ threads with $C$ = 4. With $K = 1$ the algorithm is the same as without the NUMA-specific optimization.}
\label{table:mqpl_numa_amd}
\end{table}

\begin{table}[h]
\small
\begin{center}
\begin{tabular}{ |c|c|c|c|c|c|c|c|c|c|c|c| }
\hline
 & \large{\textbf{1}} & \large{\textbf{2}} & \large{\textbf{4}} & \large{\textbf{8}} & \large{\textbf{16}} & \large{\textbf{32}} & \large{\textbf{64}} & \large{\textbf{128}} & \large{\textbf{256}} & \large{\textbf{512}} & \large{\textbf{1024}} \\
\hline
\large{\textbf{BFS USA}} & 2.22 & 2.37 & 2.58 & 2.75 & 2.93 & 3.08 & 3.22 & \color{Numa}{\textbf{3.37}} & 3.35 & 3.28 & 3.15 \\
\hline
\large{\textbf{BFS WEST}} & 1.79 & 1.93 & 2.11 & 2.32 & 2.47 & 2.64 & 2.77 & 2.84 & \color{Numa}{\textbf{2.86}} & 2.78 & 2.67 \\
\hline
\large{\textbf{BFS TWITTER}} & \color{Numa}{\textbf{1.69}} & 1.55 & \color{Numa}{\textbf{1.69}} & 1.62 & 1.55 & 1.54 & 1.50 & 1.42 & 1.35 & 1.32 & 1.28 \\
\hline
\large{\textbf{BFS WEB}} & 3.47 & 3.60 & 3.65 & \color{Numa}{\textbf{3.68}} & 3.61 & 3.44 & 3.62 & 3.13 & 2.64 & 2.69 & 2.66 \\
\hline
\large{\textbf{SSSP USA}} & 1.07 & 1.14 & 1.25 & 1.39 & 1.52 & 1.66 & 1.76 & 1.84 & 1.89 & \color{Numa}{\textbf{1.90}} & 1.89 \\
\hline
\large{\textbf{SSSP WEST}} & 0.95 & 1.01 & 1.11 & 1.23 & 1.36 & 1.47 & 1.55 & 1.62 & \color{Numa}{\textbf{1.66}} & \color{Numa}{\textbf{1.66}} & 1.63 \\
\hline
\large{\textbf{SSSP TWITTER}} & 1.56 & 1.61 & 1.80 & 1.82 & 1.93 & 1.87 & \color{Numa}{\textbf{1.96}} & 1.83 & 1.74 & 1.69 & 1.71 \\
\hline
\large{\textbf{SSSP WEB}} & 2.19 & 2.21 & 2.32 & 2.30 & 2.45 & \color{Numa}{\textbf{2.53}} & 2.48 & 2.10 & 1.91 & 1.78 & 1.67 \\
\hline
\large{\textbf{MST USA}} & 1.15 & 1.21 & 1.22 & 1.25 & 1.23 & 1.27 & \color{Numa}{\textbf{1.30}} & 1.23 & 1.20 & 1.08 & 1.09 \\
\hline
\large{\textbf{MST WEST}} & 1.20 & 1.19 & 1.26 & 1.27 & 1.24 & 1.30 & 1.28 & \color{Numa}{\textbf{1.31}} & 1.30 & 1.10 & 1.27 \\
\hline
\large{\textbf{A* USA}} & 1.18 & 1.24 & 1.32 & 1.43 & 1.55 & 1.62 & 1.70 & 1.74 & \color{Numa}{\textbf{1.79}} & \color{Numa}{\textbf{1.79}} & 1.78 \\
\hline
\large{\textbf{A* WEST}} & 1.02 & 1.07 & 1.16 & 1.28 & 1.41 & 1.50 & 1.59 & 1.64 & \color{Numa}{\textbf{1.68}} & 1.67 & 1.65 \\
\hline
\end{tabular}
\end{center}
\vspace{0.3em}
\caption{Speedups for Multi-Queue variants with the \emph{temporal locality} optimization for \ins{}, \emph{task batching} for \del{}, and various weights $K$ for non-local NUMA node accesses, obtained on the \textbf{Intel} platform on $128$ threads; the best speedups are highlighted with {\color{Numa}{\textbf{\numacol{}}}}. The baseline is the classic Multi-Queue on $128$ threads with $C$ = 4. With $K = 1$ the algorithm is the same as without the NUMA-specific optimization. }
\label{table:mqpl_numa_intel}
\end{table}

\clearpage
\subsection{MQ Optimized NUMA: insert=Task Batching, delete=Temporal Locality }
\begin{table}[h]
\small
\begin{center}
\begin{tabular}{ |c|c|c|c|c|c|c|c|c|c|c|c| }
\hline
 & \large{\textbf{1}} & \large{\textbf{2}} & \large{\textbf{4}} & \large{\textbf{8}} & \large{\textbf{16}} & \large{\textbf{32}} & \large{\textbf{64}} & \large{\textbf{128}} & \large{\textbf{256}} & \large{\textbf{512}} & \large{\textbf{1024}} \\
\hline
\large{\textbf{BFS USA}} & 1.52 & 1.58 & 1.70 & 1.78 & 1.74 & 1.84 & 1.85 & \color{Numa}{\textbf{1.86}} & 1.81 & 1.81 & 1.77 \\
\hline
\large{\textbf{BFS WEST}} & 1.23 & 1.28 & 1.33 & 1.43 & 1.40 & 1.46 & \color{Numa}{\textbf{1.51}} & 1.49 & 1.47 & 1.47 & 1.39 \\
\hline
\large{\textbf{BFS TWITTER}} & 1.13 & \color{Numa}{\textbf{1.18}} & \color{Numa}{\textbf{1.18}} & 1.11 & 1.11 & 1.06 & 1.00 & 0.93 & 0.86 & 0.83 & 0.80 \\
\hline
\large{\textbf{BFS WEB}} & 2.03 & 2.09 & 2.25 & \color{Numa}{\textbf{2.36}} & 2.08 & 1.85 & 1.76 & 1.57 & 1.53 & 1.43 & 1.44 \\
\hline
\large{\textbf{SSSP USA}} & 0.99 & \color{Numa}{\textbf{1.00}} & 0.97 & 0.90 & 0.95 & 0.95 & 0.97 & \color{Numa}{\textbf{1.00}} & 0.85 & 0.94 & 0.97 \\
\hline
\large{\textbf{SSSP WEST}} & 0.81 & 0.82 & 0.79 & 0.74 & 0.81 & \color{Numa}{\textbf{0.85}} & 0.81 & 0.80 & 0.84 & 0.81 & 0.78 \\
\hline
\large{\textbf{SSSP TWITTER}} & 1.18 & 1.27 & 1.33 & 1.31 & 1.26 & 1.25 & 1.30 & 1.28 & 1.25 & 1.26 & \color{Numa}{\textbf{1.37}} \\
\hline
\large{\textbf{SSSP WEB}} & 1.53 & 1.59 & 1.72 & 1.80 & \color{Numa}{\textbf{1.87}} & 1.85 & 1.79 & 1.50 & 1.52 & 1.45 & 1.45 \\
\hline
\large{\textbf{MST USA}} & 1.16 & 1.14 & 1.14 & 1.05 & \color{Numa}{\textbf{1.18}} & 1.09 & 0.95 & 0.98 & 0.88 & 0.89 & 0.91 \\
\hline
\large{\textbf{MST WEST}} & 1.07 & 1.06 & 1.08 & 1.09 & 1.11 & 1.13 & 1.14 & 1.18 & 1.14 & \color{Numa}{\textbf{1.19}} & 1.18 \\
\hline
\large{\textbf{A* USA}} & 0.89 & 0.88 & 0.90 & 0.88 & 0.89 & 0.92 & 0.91 & \color{Numa}{\textbf{0.94}} & 0.91 & 0.91 & 0.90 \\
\hline
\large{\textbf{A* WEST}} & 0.97 & 0.99 & 1.04 & 1.08 & 1.10 & 1.06 & \color{Numa}{\textbf{1.15}} & 1.11 & 1.14 & 1.13 & 1.11 \\
\hline
\end{tabular}
\end{center}
\vspace{0.3em}
\caption{Speedups for Multi-Queue variants with the \emph{task batching} optimization for \ins{}, \emph{temporal locality} for \del{}, and various weights $K$ for non-local NUMA node accesses, obtained on the \textbf{AMD} platform on $256$ threads; the best speedups are highlighted with {\color{Numa}{\textbf{\numacol{}}}}. The baseline is the classic Multi-Queue on $256$ threads with $C$ = 4. With $K = 1$ the algorithm is the same as without the NUMA-specific optimization.}
\label{table:mqlp_numa_amd}
\end{table}

\begin{table}[h]
\small
\begin{center}
\begin{tabular}{ |c|c|c|c|c|c|c|c|c|c|c|c| }
\hline
 & \large{\textbf{1}} & \large{\textbf{2}} & \large{\textbf{4}} & \large{\textbf{8}} & \large{\textbf{16}} & \large{\textbf{32}} & \large{\textbf{64}} & \large{\textbf{128}} & \large{\textbf{256}} & \large{\textbf{512}} & \large{\textbf{1024}} \\
\hline
\large{\textbf{BFS USA}} & 1.68 & 1.81 & 2.01 & 2.26 & 2.48 & 2.66 & 2.82 & 2.92 & \color{Numa}{\textbf{2.95}} & 2.87 & 2.73 \\
\hline
\large{\textbf{BFS WEST}} & 1.32 & 1.45 & 1.63 & 1.87 & 2.08 & 2.26 & 2.38 & \color{Numa}{\textbf{2.47}} & 2.42 & 2.36 & 2.09 \\
\hline
\large{\textbf{BFS TWITTER}} & 1.54 & 1.52 & \color{Numa}{\textbf{1.56}} & 1.45 & 1.30 & 1.20 & 1.10 & 0.96 & 0.87 & 0.74 & 0.67 \\
\hline
\large{\textbf{BFS WEB}} & 2.93 & 3.07 & 3.30 & \color{Numa}{\textbf{3.46}} & 3.43 & 2.81 & 2.38 & 1.91 & 1.68 & 1.42 & 1.25 \\
\hline
\large{\textbf{SSSP USA}} & 0.68 & 0.74 & 0.84 & 0.98 & 1.09 & 1.19 & 1.25 & 1.29 & \color{Numa}{\textbf{1.30}} & 1.27 & 1.24 \\
\hline
\large{\textbf{SSSP WEST}} & 0.65 & 0.66 & 0.78 & 0.92 & 1.03 & 1.12 & 1.19 & 1.22 & \color{Numa}{\textbf{1.23}} & 1.21 & 1.15 \\
\hline
\large{\textbf{SSSP TWITTER}} & 1.82 & 1.98 & 2.06 & \color{Numa}{\textbf{2.14}} & 2.04 & 2.00 & 1.81 & 1.80 & 1.83 & 1.77 & 1.83 \\
\hline
\large{\textbf{SSSP WEB}} & 2.04 & 2.18 & 2.38 & 2.46 & \color{Numa}{\textbf{2.61}} & 2.39 & 2.10 & 1.70 & 1.44 & 1.19 & 1.07 \\
\hline
\large{\textbf{MST USA}} & \color{Numa}{\textbf{1.18}} & 1.15 & 1.09 & 1.17 & 1.17 & 1.07 & 0.81 & 0.60 & 0.34 & 0.11 & 0.34 \\
\hline
\large{\textbf{MST WEST}} & 0.90 & 0.98 & 0.91 & 0.89 & 1.06 & 1.03 & 0.93 & 0.95 & \color{Numa}{\textbf{1.15}} & 1.01 & 0.92 \\
\hline
\large{\textbf{A* USA}} & 0.77 & 0.82 & 0.93 & 1.05 & 1.17 & 1.25 & 1.32 & \color{Numa}{\textbf{1.36}} & \color{Numa}{\textbf{1.36}} & 1.34 & 1.29 \\
\hline
\large{\textbf{A* WEST}} & 0.71 & 0.75 & 0.85 & 0.97 & 1.08 & 1.16 & 1.22 & \color{Numa}{\textbf{1.26}} & 1.25 & 1.23 & 1.15 \\
\hline
\end{tabular}
\end{center}
\vspace{0.3em}
\caption{Speedups for Multi-Queue variants with the \emph{task batching} optimization for \ins{}, \emph{temporal locality} for \del{}, and various weights $K$ for non-local NUMA node accesses, obtained on the \textbf{Intel} platform on $128$ threads; the best speedups are highlighted with {\color{Numa}{\textbf{\numacol{}}}}. The baseline is the classic Multi-Queue on $128$ threads with $C$ = 4. With $K = 1$ the algorithm is the same as without the NUMA-specific optimization.}
\label{table:mqlp_numa_intel}
\end{table}

\clearpage
\subsection{MQ Optimized NUMA: insert=Task Batching, delete=Task Batching }
\begin{table}[h]
\small
\begin{center}
\begin{tabular}{ |c|c|c|c|c|c|c|c|c|c|c|c| }
\hline
 & \large{\textbf{1}} & \large{\textbf{2}} & \large{\textbf{4}} & \large{\textbf{8}} & \large{\textbf{16}} & \large{\textbf{32}} & \large{\textbf{64}} & \large{\textbf{128}} & \large{\textbf{256}} & \large{\textbf{512}} & \large{\textbf{1024}} \\
\hline
\large{\textbf{BFS USA}} & 2.09 & 2.10 & 2.15 & 2.18 & 2.23 & 2.26 & \color{Numa}{\textbf{2.27}} & 2.22 & 2.25 & 2.22 & 2.10 \\
\hline
\large{\textbf{BFS WEST}} & 1.62 & 1.64 & 1.66 & 1.40 & 1.69 & 1.70 & 1.71 & 1.65 & \color{Numa}{\textbf{1.75}} & 1.71 & 1.57 \\
\hline
\large{\textbf{BFS TWITTER}} & 1.15 & 1.17 & \color{Numa}{\textbf{1.18}} & 1.16 & 1.14 & 1.12 & 1.06 & 1.06 & 1.07 & 1.05 & 1.00 \\
\hline
\large{\textbf{BFS WEB}} & 2.27 & 2.23 & 2.37 & 2.44 & 2.49 & \color{Numa}{\textbf{2.52}} & \color{Numa}{\textbf{2.52}} & 2.30 & 2.16 & 1.98 & 1.78 \\
\hline
\large{\textbf{SSSP USA}} & 1.10 & 1.13 & 1.16 & 1.20 & 1.22 & 1.22 & 1.25 & 1.24 & 1.22 & \color{Numa}{\textbf{1.26}} & 1.24 \\
\hline
\large{\textbf{SSSP WEST}} & 0.96 & 0.98 & 1.02 & 1.09 & 1.11 & 1.16 & \color{Numa}{\textbf{1.18}} & 1.14 & 1.12 & 1.13 & 1.16 \\
\hline
\large{\textbf{SSSP TWITTER}} & 1.18 & 1.17 & 1.27 & 1.29 & 1.32 & 1.21 & 1.29 & 1.32 & 1.27 & 1.23 & \color{Numa}{\textbf{1.33}} \\
\hline
\large{\textbf{SSSP WEB}} & 1.42 & 1.43 & 1.48 & 1.53 & 1.59 & \color{Numa}{\textbf{1.64}} & 1.61 & 1.59 & 1.49 & 1.41 & 1.39 \\
\hline
\large{\textbf{MST USA}} & 1.15 & \color{Numa}{\textbf{1.17}} & 1.12 & 1.16 & 1.11 & 1.12 & 1.14 & 1.11 & 1.11 & 1.07 & 1.06 \\
\hline
\large{\textbf{MST WEST}} & 1.15 & 1.15 & 1.15 & 1.16 & 1.15 & 1.18 & 1.18 & 1.18 & 1.14 & \color{Numa}{\textbf{1.19}} & \color{Numa}{\textbf{1.19}} \\
\hline
\large{\textbf{A* USA}} & 1.05 & 1.09 & 1.10 & 1.13 & 1.15 & 1.19 & 1.20 & 1.20 & 1.23 & 1.19 & \color{Numa}{\textbf{1.24}} \\
\hline
\large{\textbf{A* WEST}} & 0.98 & 1.03 & 1.08 & 1.10 & 1.14 & 1.19 & 1.20 & \color{Numa}{\textbf{1.21}} & \color{Numa}{\textbf{1.21}} & \color{Numa}{\textbf{1.21}} & 1.18 \\
\hline
\end{tabular}
\end{center}
\vspace{0.3em}
\caption{Speedups for Multi-Queue variants with the \emph{task batching} optimization for both \ins{} and \del{}, and various weights $K$ for non-local NUMA node accesses, obtained on the \textbf{AMD} platform on $256$ threads; the best speedups are highlighted with {\color{Numa}{\textbf{\numacol{}}}}. The baseline is the classic Multi-Queue on $256$ threads with $C$ = 4. With $K = 1$ the algorithm is the same as without the NUMA-specific optimization.}
\label{table:mqll_numa_amd}
\end{table}

\begin{table}[h]
\small
\begin{center}
\begin{tabular}{ |c|c|c|c|c|c|c|c|c|c|c|c| }
\hline
 & \large{\textbf{1}} & \large{\textbf{2}} & \large{\textbf{4}} & \large{\textbf{8}} & \large{\textbf{16}} & \large{\textbf{32}} & \large{\textbf{64}} & \large{\textbf{128}} & \large{\textbf{256}} & \large{\textbf{512}} & \large{\textbf{1024}} \\
\hline
\large{\textbf{BFS USA}} & 2.47 & 2.50 & 2.63 & 2.79 & 2.99 & 3.16 & 3.33 & \color{Numa}{\textbf{3.44}} & 3.42 & 3.35 & 3.17 \\
\hline
\large{\textbf{BFS WEST}} & 1.87 & 1.94 & 2.04 & 2.17 & 2.38 & 2.51 & 2.64 & \color{Numa}{\textbf{2.73}} & 2.56 & 2.72 & 2.58 \\
\hline
\large{\textbf{BFS TWITTER}} & \color{Numa}{\textbf{1.64}} & 1.59 & 1.58 & 1.62 & 1.61 & 1.57 & 1.50 & 1.51 & 1.49 & 1.41 & 1.43 \\
\hline
\large{\textbf{BFS WEB}} & 3.49 & 3.44 & 3.47 & 3.61 & 3.41 & \color{Numa}{\textbf{3.72}} & 3.44 & 3.45 & 3.55 & 3.33 & 3.33 \\
\hline
\large{\textbf{SSSP USA}} & 1.13 & 1.16 & 1.23 & 1.33 & 1.44 & 1.54 & 1.64 & 1.72 & 1.79 & 1.85 & \color{Numa}{\textbf{1.88}} \\
\hline
\large{\textbf{SSSP WEST}} & 0.88 & 0.88 & 0.98 & 1.05 & 1.22 & 1.26 & 1.34 & 1.37 & 1.32 & \color{Numa}{\textbf{1.38}} & 1.35 \\
\hline
\large{\textbf{SSSP TWITTER}} & 1.69 & 1.68 & 1.67 & 1.80 & 1.81 & \color{Numa}{\textbf{1.91}} & 1.88 & 1.81 & 1.72 & 1.69 & 1.65 \\
\hline
\large{\textbf{SSSP WEB}} & 2.08 & 2.10 & 2.14 & 2.22 & 2.34 & 2.36 & \color{Numa}{\textbf{2.42}} & 2.26 & 1.91 & 1.84 & 1.62 \\
\hline
\large{\textbf{MST USA}} & 1.00 & 1.05 & 1.07 & 1.13 & 1.09 & \color{Numa}{\textbf{1.15}} & 1.10 & 1.09 & 1.07 & 1.00 & 0.78 \\
\hline
\large{\textbf{MST WEST}} & 1.10 & 1.03 & 0.95 & 1.03 & 1.13 & 1.08 & \color{Numa}{\textbf{1.19}} & 1.13 & 1.04 & 1.13 & 1.09 \\
\hline
\large{\textbf{A* USA}} & 1.19 & 1.21 & 1.28 & 1.37 & 1.47 & 1.56 & 1.64 & 1.71 & 1.79 & 1.82 & \color{Numa}{\textbf{1.86}} \\
\hline
\large{\textbf{A* WEST}} & 0.97 & 1.00 & 1.09 & 1.21 & 1.23 & 1.31 & 1.41 & 1.45 & 1.45 & \color{Numa}{\textbf{1.55}} & 1.52 \\
\hline
\end{tabular}
\end{center}
\vspace{0.3em}
\caption{Speedups for Multi-Queue variants with the \emph{task batching} optimization for both \ins{} and \del{}, and various weights $K$ for non-local NUMA node accesses, obtained on the \textbf{Intel} platform on $128$ threads; the best speedups are highlighted with {\color{Numa}{\textbf{\numacol{}}}}. The baseline is the classic Multi-Queue on $128$ threads with $C$ = 4. With $K = 1$ the algorithm is the same as without the NUMA-specific optimization. }
\label{table:mqll_numa_intel}
\end{table}

\clearpage
\subsection{SMQ via D-Ary Heaps NUMA }
\begin{table}[h]
\small
\begin{center}
\begin{tabular}{ |c|c|c|c|c|c|c|c|c|c|c|c| }
\hline
 & \large{\textbf{1}} & \large{\textbf{2}} & \large{\textbf{4}} & \large{\textbf{8}} & \large{\textbf{16}} & \large{\textbf{32}} & \large{\textbf{64}} & \large{\textbf{128}} & \large{\textbf{256}} & \large{\textbf{512}} & \large{\textbf{1024}} \\
\hline
\large{\textbf{BFS USA}} & 2.85 & 2.87 & 2.88 & \color{Numa}{\textbf{2.89}} & 2.88 & 2.78 & 2.86 & 2.87 & 2.88 & 2.86 & 2.88 \\
\hline
\large{\textbf{BFS WEST}} & 2.35 & 2.38 & 2.39 & \color{Numa}{\textbf{2.42}} & 2.37 & 2.33 & 2.35 & 2.38 & 2.39 & 2.37 & 2.41 \\
\hline
\large{\textbf{BFS TWITTER}} & 1.28 & 1.28 & 1.29 & 1.28 & 1.28 & \color{Numa}{\textbf{1.30}} & \color{Numa}{\textbf{1.30}} & \color{Numa}{\textbf{1.30}} & \color{Numa}{\textbf{1.30}} & \color{Numa}{\textbf{1.30}} & 1.28 \\
\hline
\large{\textbf{BFS WEB}} & 2.55 & 2.36 & 2.59 & 2.47 & 2.62 & 2.56 & \color{Numa}{\textbf{2.63}} & \color{Numa}{\textbf{2.63}} & 2.54 & 2.60 & 2.60 \\
\hline
\large{\textbf{SSSP USA}} & 2.35 & \color{Numa}{\textbf{2.44}} & 2.43 & \color{Numa}{\textbf{2.44}} & \color{Numa}{\textbf{2.44}} & 2.41 & 2.43 & 2.43 & 2.40 & 2.43 & 2.43 \\
\hline
\large{\textbf{SSSP WEST}} & 2.12 & 2.10 & 2.10 & 2.11 & 2.12 & 2.10 & 2.11 & 2.10 & 2.11 & \color{Numa}{\textbf{2.13}} & 2.11 \\
\hline
\large{\textbf{SSSP TWITTER}} & 2.04 & 2.01 & 2.05 & 1.94 & 2.02 & 2.01 & 1.99 & 2.05 & 2.00 & \color{Numa}{\textbf{2.06}} & 2.02 \\
\hline
\large{\textbf{SSSP WEB}} & 2.63 & 2.77 & 2.67 & 2.74 & 2.60 & 2.76 & 2.86 & 2.62 & 2.80 & 2.78 & \color{Numa}{\textbf{2.88}} \\
\hline
\large{\textbf{MST USA}} & 1.99 & 1.90 & 1.94 & 1.82 & 1.85 & 1.96 & 1.89 & 1.83 & 1.95 & \color{Numa}{\textbf{2.09}} & 1.74 \\
\hline
\large{\textbf{MST WEST}} & 1.26 & 1.22 & 1.28 & 1.27 & 1.27 & 1.27 & 1.30 & 1.29 & 1.27 & \color{Numa}{\textbf{1.31}} & 1.25 \\
\hline
\large{\textbf{A* USA}} & \color{Numa}{\textbf{1.92}} & 1.89 & 1.90 & 1.87 & 1.90 & 1.89 & 1.90 & 1.90 & 1.91 & 1.91 & 1.91 \\
\hline
\large{\textbf{A* WEST}} & 1.92 & \color{Numa}{\textbf{1.94}} & 1.93 & 1.87 & 1.92 & 1.91 & 1.90 & 1.84 & 1.91 & 1.92 & 1.90 \\
\hline
\end{tabular}
\end{center}
\vspace{0.3em}
\caption{Speedups for the Stealing Multi-Queue implementation via $d$-ary heaps with various weights $K$ for non-local NUMA node accesses, obtained on the \textbf{AMD} platform on $256$ threads; the best speedups are highlighted with {\color{Numa}{\textbf{\numacol{}}}}. The baseline is the classic Multi-Queue on $256$ threads with $C$ = 4. With $K = 1$ the algorithm is the same as without the NUMA-specific optimization. }
\label{table:smq_numa_amd}
\end{table}

\begin{table}[h]
\small
\begin{center}
\begin{tabular}{ |c|c|c|c|c|c|c|c|c|c|c|c| }
\hline
 & \large{\textbf{1}} & \large{\textbf{2}} & \large{\textbf{4}} & \large{\textbf{8}} & \large{\textbf{16}} & \large{\textbf{32}} & \large{\textbf{64}} & \large{\textbf{128}} & \large{\textbf{256}} & \large{\textbf{512}} & \large{\textbf{1024}} \\
\hline
\large{\textbf{BFS USA}} & 3.27 & \color{Numa}{\textbf{3.28}} & 3.27 & 3.19 & 3.27 & \color{Numa}{\textbf{3.28}} & 3.27 & 3.26 & 3.27 & 3.27 & 3.27 \\
\hline
\large{\textbf{BFS WEST}} & \color{Numa}{\textbf{2.72}} & \color{Numa}{\textbf{2.72}} & 2.71 & 2.71 & \color{Numa}{\textbf{2.72}} & 2.71 & 2.71 & 2.70 & 2.71 & 2.71 & 2.71 \\
\hline
\large{\textbf{BFS TWITTER}} & 1.50 & 1.48 & 1.46 & 1.50 & \color{Numa}{\textbf{1.51}} & \color{Numa}{\textbf{1.51}} & 1.46 & 1.46 & 1.49 & 1.47 & 1.49 \\
\hline
\large{\textbf{BFS WEB}} & 2.85 & 2.78 & 3.04 & 2.99 & 2.90 & 2.67 & 2.97 & 2.68 & 3.00 & 2.94 & \color{Numa}{\textbf{3.08}} \\
\hline
\large{\textbf{SSSP USA}} & \color{Numa}{\textbf{2.35}} & \color{Numa}{\textbf{2.35}} & 2.34 & \color{Numa}{\textbf{2.35}} & \color{Numa}{\textbf{2.35}} & \color{Numa}{\textbf{2.35}} & 2.34 & \color{Numa}{\textbf{2.35}} & \color{Numa}{\textbf{2.35}} & \color{Numa}{\textbf{2.35}} & 2.34 \\
\hline
\large{\textbf{SSSP WEST}} & 2.02 & \color{Numa}{\textbf{2.03}} & \color{Numa}{\textbf{2.03}} & \color{Numa}{\textbf{2.03}} & 2.01 & 2.01 & 2.01 & \color{Numa}{\textbf{2.03}} & 2.01 & 2.02 & 2.01 \\
\hline
\large{\textbf{SSSP TWITTER}} & 1.92 & 1.86 & 1.92 & \color{Numa}{\textbf{1.93}} & 1.72 & 1.89 & 1.86 & 1.90 & 1.78 & 1.91 & 1.83 \\
\hline
\large{\textbf{SSSP WEB}} & 2.43 & 2.37 & 2.33 & 2.46 & \color{Numa}{\textbf{2.47}} & 2.40 & 2.29 & 2.44 & 2.44 & 2.42 & 2.30 \\
\hline
\large{\textbf{MST USA}} & 1.24 & 1.27 & 1.20 & 1.25 & \color{Numa}{\textbf{1.31}} & 1.11 & 1.28 & 1.25 & 1.28 & 1.18 & 1.22 \\
\hline
\large{\textbf{MST WEST}} & 1.21 & 1.21 & 1.25 & 1.27 & 1.21 & 1.22 & 1.19 & 1.26 & \color{Numa}{\textbf{1.28}} & 1.27 & 1.24 \\
\hline
\large{\textbf{A* USA}} & 2.14 & 2.15 & 2.16 & \color{Numa}{\textbf{2.17}} & 2.15 & 2.16 & 2.16 & 2.16 & 2.15 & 2.16 & 2.16 \\
\hline
\large{\textbf{A* WEST}} & 1.90 & 1.91 & 1.91 & 1.91 & 1.91 & \color{Numa}{\textbf{1.92}} & 1.90 & 1.90 & 1.90 & \color{Numa}{\textbf{1.92}} & \color{Numa}{\textbf{1.92}} \\
\hline
\end{tabular}
\end{center}
\vspace{0.3em}
\caption{Speedups for the Stealing Multi-Queue implementation via $d$-ary heaps with various weights $K$ for non-local NUMA node accesses, obtained on the \textbf{Intel} platform on $128$ threads; the best speedups are highlighted with {\color{Numa}{\textbf{\numacol{}}}}. The baseline is the classic Multi-Queue on $128$ threads with $C$ = 4. With $K = 1$ the algorithm is the same as without the NUMA-specific optimization. }
\label{table:smq_numa_intel}
\end{table}

\clearpage
\subsection{SMQ via Skip Lists NUMA}
\begin{table}[h]
\small
\begin{center}
\begin{tabular}{ |c|c|c|c|c|c|c|c|c|c|c|c| }
\hline
 & \large{\textbf{1}} & \large{\textbf{2}} & \large{\textbf{4}} & \large{\textbf{8}} & \large{\textbf{16}} & \large{\textbf{32}} & \large{\textbf{64}} & \large{\textbf{128}} & \large{\textbf{256}} & \large{\textbf{512}} & \large{\textbf{1024}} \\
\hline
\large{\textbf{BFS USA}} & 1.80 & 1.72 & 1.81 & 1.82 & \color{Numa}{\textbf{1.83}} & 1.80 & 1.82 & 1.81 & 1.81 & 1.82 & 1.80 \\
\hline
\large{\textbf{BFS WEST}} & 1.55 & 1.60 & 1.42 & 1.61 & 1.60 & 1.60 & 1.58 & 1.60 & \color{Numa}{\textbf{1.62}} & 1.52 & 1.60 \\
\hline
\large{\textbf{BFS TWITTER}} & 0.94 & 0.91 & 0.94 & 0.83 & 0.93 & 0.95 & 0.86 & \color{Numa}{\textbf{0.97}} & 0.87 & 0.75 & \color{Numa}{\textbf{0.97}} \\
\hline
\large{\textbf{BFS WEB}} & \color{Numa}{\textbf{3.24}} & 3.17 & 2.66 & 3.08 & 2.81 & 3.15 & 2.96 & 3.06 & 2.48 & 3.12 & 2.72 \\
\hline
\large{\textbf{SSSP USA}} & 1.68 & \color{Numa}{\textbf{1.73}} & 1.68 & 1.72 & 1.66 & 1.66 & 1.71 & 1.72 & 1.71 & 1.70 & 1.71 \\
\hline
\large{\textbf{SSSP WEST}} & 1.52 & 1.54 & 1.53 & 1.55 & 1.55 & 1.53 & 1.54 & \color{Numa}{\textbf{1.56}} & 1.55 & \color{Numa}{\textbf{1.56}} & 1.55 \\
\hline
\large{\textbf{SSSP TWITTER}} & 1.87 & 1.89 & \color{Numa}{\textbf{1.99}} & 1.70 & 1.77 & 1.69 & 1.85 & 1.64 & 1.78 & 1.63 & 1.78 \\
\hline
\large{\textbf{SSSP WEB}} & 3.28 & 3.25 & 3.13 & 3.18 & \color{Numa}{\textbf{3.29}} & 3.23 & 3.20 & 3.19 & 3.15 & 3.20 & 3.20 \\
\hline
\large{\textbf{MST USA}} & 1.42 & 1.40 & 1.45 & \color{Numa}{\textbf{1.46}} & 1.44 & \color{Numa}{\textbf{1.46}} & 1.44 & 1.29 & 1.45 & 1.44 & 1.45 \\
\hline
\large{\textbf{MST WEST}} & 1.22 & \color{Numa}{\textbf{1.32}} & \color{Numa}{\textbf{1.32}} & 1.31 & 1.31 & 1.27 & 1.20 & 1.25 & 1.25 & 1.18 & 1.23 \\
\hline
\large{\textbf{A* USA}} & 1.49 & 1.50 & 1.50 & \color{Numa}{\textbf{1.51}} & 1.50 & 1.50 & 1.50 & 1.50 & 1.49 & 1.50 & 1.50 \\
\hline
\large{\textbf{A* WEST}} & \color{Numa}{\textbf{1.47}} & \color{Numa}{\textbf{1.47}} & 1.46 & \color{Numa}{\textbf{1.47}} & 1.45 & \color{Numa}{\textbf{1.47}} & 1.41 & \color{Numa}{\textbf{1.47}} & \color{Numa}{\textbf{1.47}} & \color{Numa}{\textbf{1.47}} & \color{Numa}{\textbf{1.47}} \\
\hline
\end{tabular}
\end{center}
\vspace{0.3em}
\caption{Speedups for the Stealing Multi-Queue implementation via skip lists with various weights $K$ for non-local NUMA node accesses, obtained on the \textbf{AMD} platform on $256$ threads; the best speedups are highlighted with {\color{Numa}{\textbf{\numacol{}}}}. The baseline is the classic Multi-Queue on $256$ threads with $C$ = 4. With $K = 1$ the algorithm is the same as without the NUMA-specific optimization. }
\label{table:slsmq_numa_amd}
\end{table}

\begin{table}[h]
\small
\begin{center}
\begin{tabular}{ |c|c|c|c|c|c|c|c|c|c|c|c| }
\hline
 & \large{\textbf{1}} & \large{\textbf{2}} & \large{\textbf{4}} & \large{\textbf{8}} & \large{\textbf{16}} & \large{\textbf{32}} & \large{\textbf{64}} & \large{\textbf{128}} & \large{\textbf{256}} & \large{\textbf{512}} & \large{\textbf{1024}} \\
\hline
\large{\textbf{BFS USA}} & 1.83 & 1.85 & 1.84 & 1.84 & 1.83 & 1.84 & 1.84 & 1.85 & \color{Numa}{\textbf{1.86}} & 1.85 & 1.84 \\
\hline
\large{\textbf{BFS WEST}} & 1.62 & 1.62 & 1.63 & 1.63 & 1.64 & \color{Numa}{\textbf{1.65}} & \color{Numa}{\textbf{1.65}} & \color{Numa}{\textbf{1.65}} & \color{Numa}{\textbf{1.65}} & \color{Numa}{\textbf{1.65}} & \color{Numa}{\textbf{1.65}} \\
\hline
\large{\textbf{BFS TWITTER}} & 0.99 & 0.95 & 1.03 & \color{Numa}{\textbf{1.07}} & 0.97 & 1.03 & 1.00 & 1.02 & 1.06 & 0.98 & 1.06 \\
\hline
\large{\textbf{BFS WEB}} & \color{Numa}{\textbf{2.84}} & 2.35 & 2.67 & \color{Numa}{\textbf{2.84}} & 2.82 & 2.79 & 2.80 & 2.68 & 2.83 & 2.68 & 2.65 \\
\hline
\large{\textbf{SSSP USA}} & 1.47 & 1.47 & 1.47 & \color{Numa}{\textbf{1.48}} & 1.47 & \color{Numa}{\textbf{1.48}} & \color{Numa}{\textbf{1.48}} & \color{Numa}{\textbf{1.48}} & \color{Numa}{\textbf{1.48}} & \color{Numa}{\textbf{1.48}} & \color{Numa}{\textbf{1.48}} \\
\hline
\large{\textbf{SSSP WEST}} & 1.35 & 1.33 & 1.36 & 1.36 & 1.36 & 1.36 & \color{Numa}{\textbf{1.37}} & 1.36 & 1.36 & \color{Numa}{\textbf{1.37}} & \color{Numa}{\textbf{1.37}} \\
\hline
\large{\textbf{SSSP TWITTER}} & 1.66 & 1.65 & 1.65 & 1.64 & 1.65 & \color{Numa}{\textbf{1.67}} & 1.65 & 1.66 & 1.66 & \color{Numa}{\textbf{1.67}} & 1.65 \\
\hline
\large{\textbf{SSSP WEB}} & 2.19 & 2.19 & 2.21 & 2.21 & 2.20 & 2.19 & 2.20 & 2.21 & 2.19 & \color{Numa}{\textbf{2.23}} & 2.20 \\
\hline
\large{\textbf{MST USA}} & 1.34 & 1.33 & 1.29 & 1.29 & 1.30 & 1.29 & 1.32 & 1.36 & 1.28 & \color{Numa}{\textbf{1.38}} & 1.28 \\
\hline
\large{\textbf{MST WEST}} & 1.41 & 1.47 & \color{Numa}{\textbf{1.48}} & 1.40 & 1.46 & \color{Numa}{\textbf{1.48}} & 1.47 & 1.38 & \color{Numa}{\textbf{1.48}} & 1.45 & \color{Numa}{\textbf{1.48}} \\
\hline
\large{\textbf{A* USA}} & 1.47 & 1.47 & \color{Numa}{\textbf{1.48}} & 1.47 & 1.47 & \color{Numa}{\textbf{1.48}} & 1.47 & \color{Numa}{\textbf{1.48}} & 1.47 & 1.47 & \color{Numa}{\textbf{1.48}} \\
\hline
\large{\textbf{A* WEST}} & 1.31 & \color{Numa}{\textbf{1.33}} & \color{Numa}{\textbf{1.33}} & 1.32 & \color{Numa}{\textbf{1.33}} & \color{Numa}{\textbf{1.33}} & \color{Numa}{\textbf{1.33}} & \color{Numa}{\textbf{1.33}} & \color{Numa}{\textbf{1.33}} & \color{Numa}{\textbf{1.33}} & \color{Numa}{\textbf{1.33}} \\
\hline
\end{tabular}
\end{center}
\vspace{0.3em}
\caption{Speedups for the Stealing Multi-Queue implementation via skip lists with various weights $K$ for non-local NUMA node accesses, obtained on the \textbf{Intel} platform on $128$ threads; the best speedups are highlighted with {\color{Numa}{\textbf{\numacol{}}}}. The baseline is the classic Multi-Queue on $128$ threads with $C$ = 4. With $K = 1$ the algorithm is the same as without the NUMA-specific optimization. }
\label{table:slsmq_numa_intel}
\end{table}



\clearpage
\section{Final Results: The Magnified Versions for AMD and Intel}\label{appendix:final:magnified}
\begin{figure*}[h]
    \centering
    \includegraphics[width=0.82\textwidth]{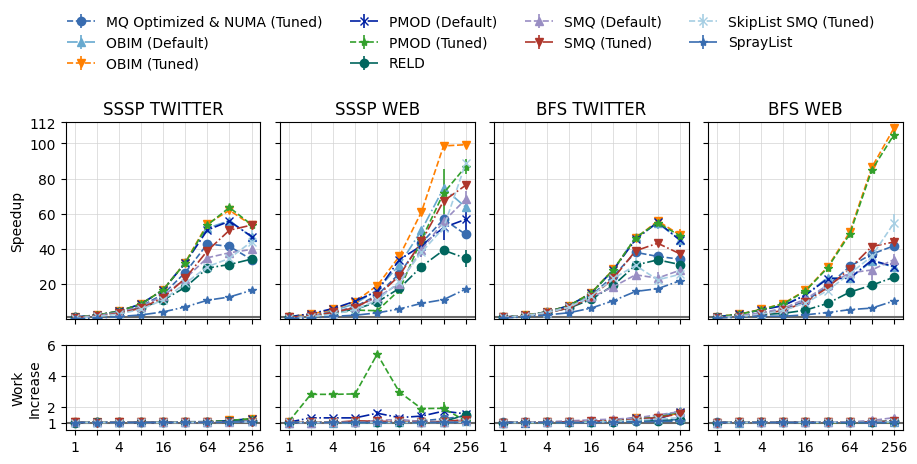}
    \includegraphics[width=0.82\textwidth]{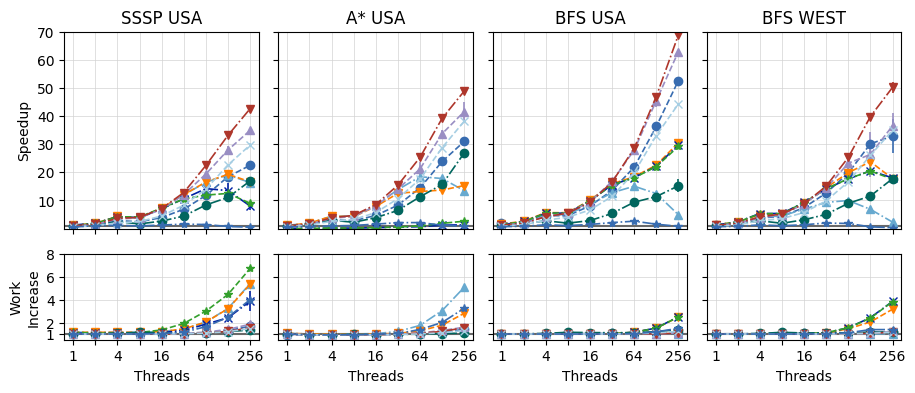}
    \includegraphics[width=0.82\textwidth]{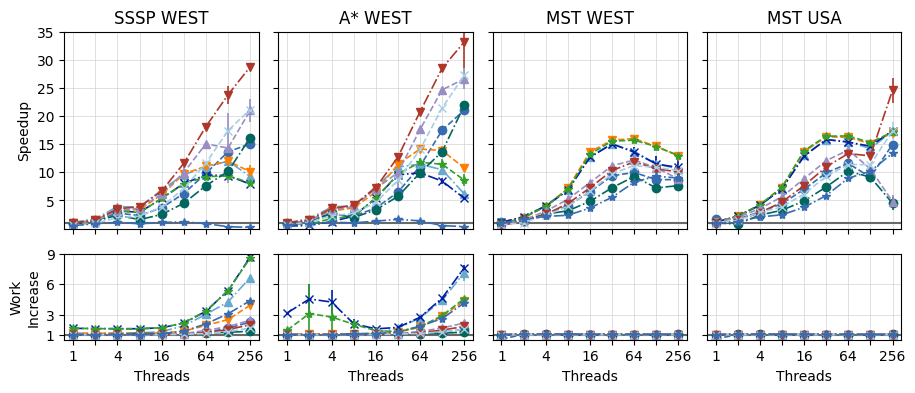}       
    \caption{Comparison of the tuned and default variants of SMQ, PMOD, and OBIM, an optimized version of the classic Multi-Queue, SprayList, and RELD schedulers on the \textbf{AMD} platform. Speedups are versus the baseline Multi-Queue executed on a single thread. Implementation details are provided in the main body.}
    \label{fig:big_amd} 
\end{figure*}

\begin{figure*}[h]
    \centering
    \includegraphics[width=0.82\textwidth]{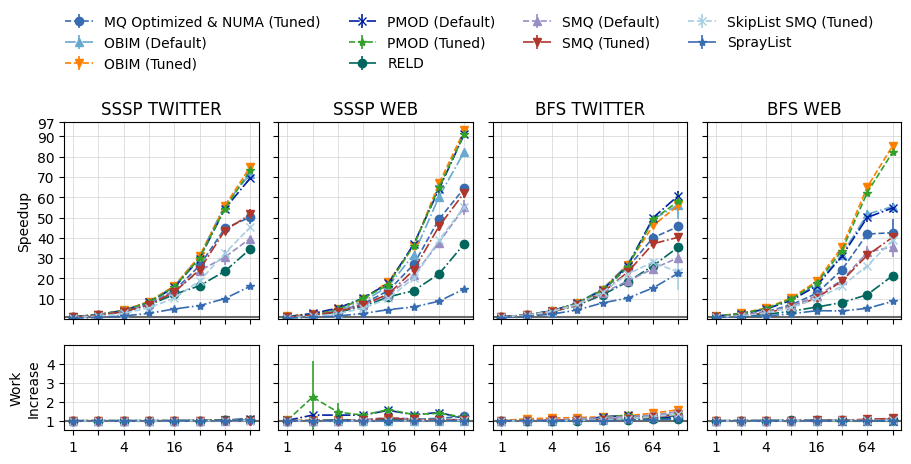}
    \includegraphics[width=0.82\textwidth]{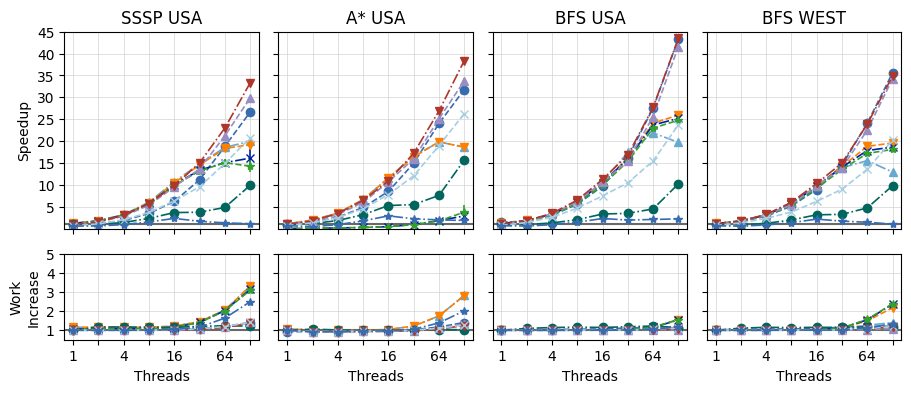}
    \includegraphics[width=0.82\textwidth]{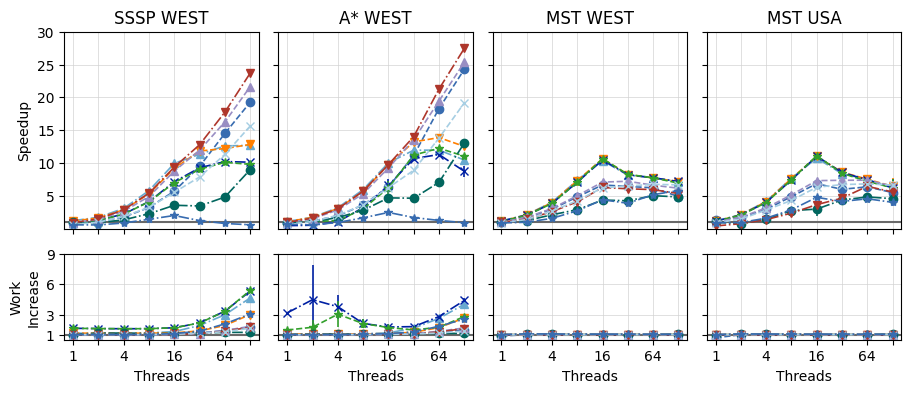}       
    \caption{Comparison of the tuned and default variants of SMQ, PMOD, and OBIM, an optimized version of the classic Multi-Queue, SprayList, and RELD schedulers on the \textbf{Intel} platform. Speedups are versus the baseline Multi-Queue executed on a single thread. Implementation details are provided in the main body.}
    \label{fig:big_intel} 
\end{figure*}

\end{document}